\documentclass{article}

\usepackage{fullpage}

\usepackage{graphicx}
\graphicspath{{figs/}}
\usepackage{amsmath,amssymb,amsfonts, amsthm}
\usepackage{xcolor}
\usepackage{xspace}
\usepackage{thm-restate}
\usepackage{mathtools}

\usepackage[colorlinks,linkcolor=blue,filecolor=blue,citecolor=blue,urlcolor=blue,pagebackref]{hyperref}

\usepackage{soul}

\usepackage{algorithm}
\usepackage{algpseudocode}
\algtext*{EndWhile}
\algtext*{EndIf}
\algtext*{EndFor}
\algrenewcommand\algorithmicrequire{\textbf{Input:}}
\algrenewcommand\algorithmicensure{\textbf{Output:}}

\usepackage[colorinlistoftodos,textsize=tiny,textwidth=2cm,color=red!25!white,obeyFinal]{todonotes}

 \newcommand{\snote}[1]{\todo[color=blue!25!white]{SL: #1}\xspace}

\usepackage{enumitem}
\setlist{topsep=3pt, itemsep=0pt}
\newcommand{\R}{\mathbb{R}}
\newcommand{\cost}{\mathrm{cost}}
\newcommand{\obj}{\mathrm{obj}}

\newtheorem{theorem}{Theorem}
\newtheorem{lemma}[theorem]{Lemma}
\newtheorem{definition}[theorem]{Definition}
\newtheorem{claim}[theorem]{Claim}
\newtheorem{coro}[theorem]{Corollary}

\newcommand{\ceil}[1]{{\left\lceil#1\right\rceil}}

\DeclareMathOperator*{\E}{{\mathbb{E}}}

\newcommand{\calC}{{\mathcal{C}}}
\newcommand{\calK}{{\mathcal{K}}}

\newcommand{\poly}{{\mathrm{poly}}}

\newcommand{\pedge}{{$+$edge}\xspace}
\newcommand{\medge}{{$-$edge}\xspace}
\newcommand{\pedges}{{$+$edges}\xspace}
\newcommand{\medges}{{$-$edges}\xspace}

\newcommand{\ppp}{$+\!+\!+$\xspace}
\newcommand{\ppm}{$+\!+\!-$\xspace}
\newcommand{\pmm}{$+\!-\!-$\xspace}
\newcommand{\mmm}{$-\!-\!-$\xspace}
\newcommand{\mmp}{$-\!-\!+$\xspace}

\newcommand{\adm}{\mathrm{adm}}
\newcommand{\err}{\mathrm{err}}

\newcommand{\opt}{{\mathrm{opt}}}
\newcommand{\optsdp}{{\mathrm{OPT}_{\mathrm{SDP}}}}

\newcommand{\epsrt}{\varepsilon_{\mathrm{rt}}}
\newcommand{\cc}{{Correlation Clustering}\xspace}
\newcommand{\hardnessratio}{24/23}

\newcommand{\lp}{\Delta}

\newcommand{\chbudget}{1.56}
\newcommand{\sdpbudget}{1.485}

\newcommand{\cov}{\mathrm{COV}}

\newcommand{\budgetp}{\mathrm{budget}^+}
\newcommand{\budgetn}{\mathrm{budget}^-}

\newcommand{\xuv}{x_{uv}}

\newcommand{\tuv}{a}
\newcommand{\tuw}{b}
\newcommand{\tvw}{c}

\newcommand{\eps}{\varepsilon}

\newcommand{\Co}{C_{\alpha}}

\newcommand{\woSArat}{1.56}

\newenvironment{cproof}
{\begin{proof}
 [Proof.]
 \vspace{-1.5\parsep}
}
{ \end{proof}}

\newcommand{\pureclusterlpratio}{1.485}

\allowdisplaybreaks

\begin{document}



\title{
Understanding the Cluster LP for Correlation Clustering\footnote{The conference version of this paper~\cite{CCLLNV24} claimed an approximation ratio of 1.437 whose proof currently has a gap. This version fixes the gap with a slightly worse approximation ratio of 1.485.}
}

\author{ 
{Nairen Cao\footnote{Supported by NSF grant CCF-2008422.}} \\ New York University \\
\texttt{nc1827@nyu.edu}
\and
{Vincent Cohen-Addad} \\ Google Research\\ \texttt{cohenaddad@google.com}
\and
{Euiwoong Lee\footnote{Supported in part by NSF grant CCF-2236669 and Google.}} \\ University of Michigan\\ \texttt{euiwoong@umich.edu}
\and {Shi Li \footnote{Affiliated with the Department of Computer Science and Technology in Nanjing University, and supported by the State Key Laboratory for Novel Software Technology, and the New Cornerstone Science Laboratory.}} \\ Nanjing University\\ \texttt{shili@nju.edu.cn} 
\and 
{Alantha Newman} \\ Université Grenoble Alpes\\ \texttt{alantha.newman@grenoble-inp.fr} 
\and
{Lukas Vogl \footnote{Supported by the Swiss National Science Foundation project 200021-184656 "Randomness in Problem Instances and Ran-
domized Algorithms".}} \\ EPFL \\ \texttt{lukas.vogl@epfl.ch}
}


\maketitle

\begin{abstract}
In the classic \cc problem introduced by Bansal, Blum, and Chawla~\cite{BBC04}, the input is a complete graph where edges are labeled either $+$ or $-$, and the goal is to find a partition of the vertices that minimizes the sum of the \pedges across parts plus the sum of the \medges within parts.
In recent years, Chawla, Makarychev, Schramm
and Yaroslavtsev~\cite{CMSY15} gave a 2.06-approximation by
providing a near-optimal rounding of the standard LP, and Cohen-Addad, Lee, Li, and Newman~\cite{CLN22, CLLN23} finally bypassed the integrality gap of 2 for this LP giving a $1.73$-approximation for the problem.

While introducing new ideas for \cc, their algorithm is more complicated than {\em typical} approximation algorithms in the following two aspects: (1) It is based on two different relaxations with separate rounding algorithms connected by the round-or-cut procedure. (2) Each of the rounding algorithms has to separately handle seemingly inevitable {\em correlated rounding errors}, coming from {\em correlated rounding} of Sherali-Adams and other strong LP relaxations~\cite{GS11, BRS11, RT12}. 

In order to create a simple and unified framework for \cc similar to those for {\em typical} approximate optimization tasks, 
we propose the {\em cluster LP} as a strong linear program for \cc.
It is exponential-sized, but we show that it can be $(1+\eps)$-approximately solved in polynomial time for any $\eps > 0$, providing the framework for designing rounding algorithms without worrying about correlated rounding errors; these errors are handled uniformly in solving the relaxation. 

We demonstrate the power of the cluster LP by presenting new rounding algorithms, and providing two analyses, one analytically proving a $\chbudget$-approximation and the other solving a factor-revealing SDP to show a $\sdpbudget$-approximation. Both proofs introduce principled methods by which to analyze the performance of the algorithm, resulting in a significantly improved approximation guarantee.

Finally, we prove an integrality gap of $4/3$ for the cluster LP, showing our $\sdpbudget$-upper bound cannot
be drastically improved. Our gap instance directly inspires an improved NP-hardness of approximation with a ratio $24/23 \approx 1.042$; no explicit hardness ratio was known before.

\end{abstract}

\newpage
\tableofcontents
\newpage

\section{Introduction}
\label{sec:intro}

Clustering is a classic problem in unsupervised machine learning and
data mining. Given a set of data elements and pairwise similarity
information between the elements, the task is to find a partition of the data elements into clusters to achieve (often contradictory) goals of placing similar elements in the same cluster and separating different elements in different clusters. Introduced by Bansal, Blum, and Chawla~\cite{BBC04}, \cc
elegantly models such tension and has become one of the most widely studied
formulations for graph clustering.
The input of the problem consists of a
complete graph $(V, E^+ \uplus E^-)$, where $E^+ \uplus E^- = {V
  \choose 2}$, $E^+$ representing the so-called \emph{positive} edges
and $E^-$ the so-called \emph{negative} edges.
The goal is to find a clustering (partition) of $V$, namely $(V_1, \dots, V_k)$,
that minimizes the number of unsatisfied edges, namely the $+$edges
between different clusters and the $-$edges within the same cluster.
Thanks to the simplicity and modularity of the formulation, \cc has
found a number of applications, e.g., finding clustering
ensembles \cite{bonchi2013overlapping}, duplicate detection
\cite{arasu2009large}, community mining \cite{chen2012clustering},
disambiguation tasks \cite{kalashnikov2008web}, automated labelling
\cite{agrawal2009generating, chakrabarti2008graph} and many more.

This problem is APX-Hard~\cite{CGW05}, and various $O(1)$-approximation algorithms~\cite{BBC04, CGW05} have been proposed in the literature. Ailon, Charikar and Newman introduced an influential {\em pivot-based} algorithm, which leads to a combinatorial $3$-approximation and a
$2.5$-approximation with respect to the standard LP relaxation~\cite{ACN08}. The LP-based rounding was improved by Chawla, Makarychev, Schramm
and Yaroslavtsev to a $2.06$-approximation~\cite{CMSY15}, nearly matching the LP integrality gap of 2 presented in~\cite{CGW05}. 


It turns out that (a high enough level of) the Sherali-Adams hierarchy can be used to design a strictly better than $2$-approximation. Cohen-Addad, Lee, and Newman~\cite{CLN22} showed that $O(1/\eps^2)$ rounds of the Sherali-Adams hierarchy have an integrality gap of at most $(1.994+\eps)$. This approximation ratio was improved by Cohen-Addad, Lee, Li, and Newman~\cite{CLLN23} to $(1.73+\eps)$ in $n^{\poly(1/\eps)}$-time, which combines {\em pivot-based rounding} and {\em set-based rounding}. 

One undesirable feature of~\cite{CLLN23} is the lack of a single convex relaxation with respect to which the approximation ratio is analyzed. For technical reasons, it combines the two rounding algorithms via a generic {\em round-or-cut} framework. Given $x \in [0,1]^E$, each of the two rounding algorithms outputs either an integral solution with some guarantee or a hyperplane separating $x$ from the convex hull of integral solutions; if both algorithms output integral solutions, one of them is guaranteed to achieve the desired approximation factor. 
Though each of the rounding procedures is based on some LP relaxations, they are different, so there is no single relaxation that can be compared to the value of the final solution. 

In this work, we propose the {\em cluster LP} as a single relaxation that captures all of the existing algorithmic results. Based on this new unified framework, we design a new rounding algorithm as well as principled tools for the analysis that significantly extend the previous ones, ultimately yielding a new approximation ratio of $\sdpbudget+\eps$. The study of the cluster LP sheds light on the hardness side as well, as we prove a $4/3 \approx 1.33$ gap for the cluster LP and a $24/23 \approx 1.042$ NP-hardness of approximation.

\subsection{Our Results}
We first state the cluster LP here. It is similar to {\em configuration LPs} used for scheduling and assignment problems~\cite{bansal2006santa, fleischer2006tight}.
In the cluster LP, we have a variable $z_S$ for every $S \subseteq V, S \neq \emptyset$, that indicates if $S$ is a cluster in the output clustering or not. As usual, $x_{uv}$ for every $uv \in {V \choose 2}$ indicates if $u$ and $v$ are separated in the clustering or not. For any $x \in [0, 1]^{{V \choose 2}}$, we define $\obj(x) := \sum_{uv \in E^+} x_{uv} + \sum_{uv \in E^-} (1-x_{uv})$ to be the fractional number of edges in disagreement in the solution $x$.
\begin{equation}
	\min \qquad \obj(x) \qquad \text{s.t.} \tag{cluster LP} \label{LP:clusterlp}
\end{equation} \vspace{-15pt}
\begin{align}
	\sum_{S \ni u} z_S &= 1   &\quad &\forall u \in V \label{LPC:set-covered}\\
	\sum_{S \supseteq \{u, v\}} z_S &= 1 - x_{uv} &\quad &\forall uv \in {V \choose 2} \label{LPC:set-define-x}\\
	z_S &\geq 0 &\qquad &\forall S \subseteq V, S\neq \emptyset
\end{align}
The objective of the LP is to minimize $\obj(x)$, which is a linear function. 
\eqref{LPC:set-covered} requires that every vertex $u$ appears in exactly one cluster, \eqref{LPC:set-define-x} gives the definition of $x_{uv}$ using $z$ variables. 

The idea behind this LP was used in \cite{CLLN23} to design their set-based rounding algorithm, though the LP was not formulated explicitly in that paper. 
%
%
%
%
Moreover, the paper did not provide an efficient algorithm to solve it approximately.  Our first result shows that we can approximately solve the cluster LP in polynomial time, despite it having an exponential number of variables.  We remark that unlike the configuration LPs for many problems, we do not know how to solve the cluster LP simply by considering its dual.

\begin{restatable}{theorem}{thmsolveclusterLP} \label{thm:solving-cluster-LP}
	Let $\eps > 0$ be a small enough constant and $\opt$ be the cost of the optimum solution to the given \cc instance. In time $n^{\poly(1/\eps)}$, we can output a solution $\big((z_S)_{S \subseteq V},(x_{uv})_{uv \in {V \choose 2}}\big)$ to the cluster LP with $\obj(x) \leq (1+\eps)\opt$, described using a list of non-zero coordinates.~\footnote{We remark that $\obj(x)$ given by the theorem is at most $1+\eps$ times $\opt$, instead of the value of the cluster LP. This is sufficient for our purpose. One should also be able to achieve the stronger guarantee of $(1+\eps)$-approximation to the optimum fractional solution. Instead of dealing with the optimum clustering $\calC^*$ in the analysis, we deal with the optimum fractional clustering to the LP. For simplicity, we choose to prove the theorem with the weaker guarantee.}
\end{restatable}


The cluster LP is one of the most powerful LPs that has been considered for the problem. 
A large portion of the algorithms and analysis in \cite{CLN22} and \cite{CLLN23} is devoted to handle the additive errors incurred by the correlated rounding procedure, which is inherited from the Raghavendra-Tan rounding technique \cite{RT12}.
Instead, we move the complication of handling rounding errors into the procedure of solving the cluster LP relaxation. 

With this single powerful relaxation, we believe that Theorem~\ref{thm:solving-cluster-LP} provides a useful framework for future work that may use more ingenious rounding of the exponential-sized cluster LP without worrying about errors.  Indeed, the constraints in the cluster LP imply that the matrix $(1 - x_{uv})_{u, v \in V}$ is PSD, \footnote{Consider the matrix $Y \in [0, 1]^{V \times V}$ where $y_{uv} = 1 - x_{uv}$ for every $u, v \in V$ ($Y_{uu}=1, \forall u \in V$). For every $w \in \R^V$, we have 
     $w^TYw = \sum_{u,v \in V} y_{uv} w_u w_v = \sum_{u, v}\sum_{S \supseteq \{u, v\}} z_S w_u w_v = \sum_{u, v} \sum_{S \subseteq V} z_S\cdot (w_u \cdot 1_{u \in S}) \cdot (w_v \cdot 1_{v \in S})
     = \sum_{S \subseteq V} z_S \left(\sum_{u \in S} w_u\right)\left(\sum_{v \in S} w_v\right) \geq 0$.}
 and thus the LP is at least as strong as the natural SDP for the problem. 
For the complementary version of maximizing the number of correct edges, the standard SDP is known to give a better approximation guarantee of $0.766$~\cite{swamy2004correlation, CGW05}. For the minimization version, the standard SDP has integrality gap at least 1.5 (see Appendix~\ref{app:sdp-int-gap}), but it is still open whether this program has an integrality gap strictly below 2 or not.

We demonstrate the power of the cluster LP by presenting and analyzing the following algorithm, significantly improving the previous best $1.73$-approximation. 

\begin{theorem}
There exists a $(\chbudget+\eps)$-approximation algorithm for \cc that runs in time $O(n^{\poly(1/\eps)})$.
\label{thm:main:algo:hand}
\end{theorem}

This is achieved by a key modification of the
pivot-based rounding algorithm that is used in conjunction with the
set-based algorithm as in~\cite{CLLN23}.  
In combination with more careful analysis, which
involves principled methods to obtain the best
{\em budget function}, we obtain a significantly improved approximation ratio.


In order to obtain an even tighter analysis of a similar algorithm, we introduce the new {\em factor revealing SDP} that searches over possible global distributions of triangles in valid \cc instances. By numerically solving such an SDP, we can further improve the approximation ratio of the same algorithm. 

\begin{theorem}
There exists a $(\sdpbudget+\eps)$-approximation algorithm for \cc that runs in time $O(n^{\poly(1/\eps)})$.
\label{thm:main:algo:computer}
\end{theorem}

While the proof includes a feasible solution to a large SDP and is not human-readable, we prove that our SDP gives an {\em upper bound} on the approximation ratio, so it is a complete proof modulo the SDP feasibility of the solution. Our program and solution can be found at \url{https://github.com/correlationClusteringSDP/SDP1437code/}.

We also study lower bounds and prove the following lower bound on the integrality gap of the cluster LP. 

\begin{restatable}{theorem}{thmgap} 
For any $\eps > 0$, the integrality gap of the cluster LP is at least $4/3 - \eps$.
\label{thm:main:gap}
\end{restatable}

This integrality gap for the cluster LP, after some (well-known) loss, directly translates to NP-hardness. Apart from the APX-hardness~\cite{CGW05}, it is the first hardness with an explicit hardness ratio. 

\begin{restatable}{theorem}{thmhardness} 
Unless $\mathbf{P} = \mathbf{BPP}$, for any $\eps > 0$, there is no $(\hardnessratio - \eps)$-approximation algorithm for \cc.
\label{thm:main:hardness}
\end{restatable}

\subsection{Further Related Work}
\label{sec:furtherrelatedwork}
The weighted version of \cc, where each pair of vertices has an associated weight
and unsatisfied edges contribute a cost proportional to their weight
to the objective, is shown to be 
equivalent to the Multicut problem~\cite{demaine2006correlation}, implying that there is an $O(\log n)$-approximation but no constant factor approximation is possible under the Unique Games Conjecture~\cite{CKKRS06}.

In the unweighted case, a PTAS exists when the number of clusters is a
fixed constant~\cite{giotis2006correlation, karpinski2009linear}. Much study has been devoted to the minimization version of
\cc in various computational models, for example in the online setting~\cite{mathieu2010online,NEURIPS2021_250dd568,CLMP22}, as well as in
other practical settings such as distributed,
parallel or streaming~\cite{chierichetti2014correlation,pmlr-v37-ahn15, CLMNP21, DBLP:conf/nips/PanPORRJ15, DBLP:conf/wdag/CambusCMU21, DBLP:conf/icml/Veldt22,DBLP:conf/www/VeldtGW18, DBLP:conf/innovations/Assadi022, DBLP:conf/focs/BehnezhadCMT22, DBLP:conf/soda/BehnezhadCMT23,
cao2023breaking, makarychev2023single,CambusKLPU-SODA24}.  Other recent work involves settings with fair or local guarantees~\cite{ahmadian2023improved,davies2023fast,heidrich20234}. 

\section{Overview and Algorithm}
In this section, we provide an overview of our algorithms and tools
for the analysis.  Section~\ref{sec:overview-solving} introduces some
techniques used for formulating and solving the cluster LP.
Section~\ref{sec:overview-rounding} presents our full
rounding algorithm and provides high-level ideas for the analysis.
Section~\ref{sec:overview-hardness} discusses ideas behind our gap and
hardness results.  These techniques were introduced over the course of
three conference papers~\cite{CLN22, CLLN23, CCLLNV24}, so we try to
highlight their connections and developments over time. For instance,
since the cluster LP was formalized in the last conference
paper~\cite{CCLLNV24}, several of the following techniques were
initially presented in relation to a direct rounding of our {\em
  initial} (Sherali-Adams-based) relaxation.

\subsection{Formulating and Solving Cluster LP}
\label{sec:overview-solving}

\paragraph{Previous Algorithms.}
To understand in more detail our contribution and the new techniques
introduced, we need to provide a brief summary of previous
approaches. We first recall the classic LP relaxation (whose
integrality gap is known to be in $[2, 2.06]$). There is a variable
$x_{uv}$ for each pair of vertices whose intended value is 1 if $u,v$
are not in the same cluster and $0$ otherwise. The goal is thus to
minimize $\sum_{uv \in E^+} x_{uv} + \sum_{uv \in E^-} (1-x_{uv})$,
under the classic triangle inequality constraint: $\forall u,v,w$,
$x_{uv} \le x_{uw} + x_{wv}$.

To obtain the 2.5-approximation algorithm, Ailon, Charikar and
Newman~\cite{ACN08} designed a pivot-based rounding method that
proceeds as follows: (1) Pick a random vertex $p$, called the
\emph{pivot}, and create the cluster containing $p$; and (2) Recurse
on the rest of the instance (the vertices not in the cluster).  This
hence builds the clustering in a sequential manner. Once we commit to
this scheme, the main design question that remains is how to construct
the cluster containing the pivot.  The solution of \cite{ACN08} was to
go over all the other vertices and for each other vertex $u$ place it
in the cluster of the pivot $p$ with probability $(1-x_{pu})$,
independently of the other random decisions made for the other
vertices. 

This was further improved by Chawla, Makarychev, Schramm and
Yaroslavtsev~\cite{CMSY15} who kept the same scheme but provided a
better rounding approach: For $+$edges, they replaced the probability
$(1-x_{pu})$ with a new rounding function $f^+$ to apply to the
quantity $(1-x_{pu})$ yielding the probability of incorporating $u$ in
$p$'s cluster.  Similar to the approach of \cite{ACN08}, the
(streamlined) analysis of this rounding scheme was triangle-based: The
analysis is a charging scheme that charges the cost paid by each pair
$u,v$ to the triangle $p,u,v$ where $p$ is the pivot which decided
edge $u,v$.  Then the crux of the analysis is to show that for any
triangle $p,u,v$, the charge is bounded compared to the LP cost.

\paragraph{Sherali-Adams and Correlated Rounding.}
One potential place for an improvement in the previous algorithms is
the {\em independent rounding} given a pivot $p$; in both~\cite{ACN08,
  CMSY15}, once we choose $p$, other vertices are put into the cluster
of $p$ independently of one another.  It might be more advantageous to
make these decisions in a correlated way.  However, the standard LP,
which only has variables $\{ x_{uv} \}_{uv \in \binom{V}{2}}$ 
for the pairs, is not suited to provide
information about such correlations.

Motivated by this fact, our first new idea, which appeared in the
first conference version~\cite{CLN22}, is the usage of
\emph{correlated rounding} based on the {\em Sherali-Adams
  hierarchy}. Namely, given a pivot $p$, the set of $+$neighbors of
$p$ that join the cluster of $p$ is chosen in a correlated manner,
using the techniques of Raghavendra and Tan~\cite{RT12}.  Concretely,
the Sherali-Adams hierarchy provides variables of the form $y_{S}$ for
any constant-sized set of vertices $S$ that indicates the probability
that all the vertices in $S$ are in the same cluster.\footnote{In
which case, $y_{uv} = 1-x_{uv}$.}  Given a pivot $p$, the correlated
rounding then allows to sample in such a way that the probability that
$u,v$ join the cluster of $p$ is $y_{puv} \pm \eps$, where $\eps$ is
an arbitrarily small constant.  Applying this rounding to a set of
vertices $S$, it ensures that we can sample a subset $C$ of $S$ both
being of arbitrary size (importantly not necessarily constant) such
that pairs $u,v$ ends up in $C$ with probability $y_{puv}$ on average,
up to losing an additive $\eps |S|^2$ error in the cost.

The main issue we need to face with the above plan of attack is that the additive error of $\eps|S|^2$, which we pay when we create a cluster $C$, may end up being much larger than the cost of the optimum solution. In fact, an optimum solution may have cost zero.  
The correlated rounding was initially developed for {\em maximization
  problems} like Max-CSPs~\cite{de2007linear, BRS11, GS11, RT12,
  yoshida2014approximation} where this additive error can be easily
ignored by only additively sacrificing $\eps$ in the approximation
ratio, but this is not the case for \cc.
\footnote{The only exception is~\cite{GS11} that studied minimization versions including Min Bisection and Min Uncut, but their technique heavily relies on the property of cuts not easily adaptable to \cc.}

\cite{CLN22} applied this rounding to +neighbors of $p$ to get a $(2 + \eps)$-approximation (with respect to a $O(1/\eps^2)$-level of Sherali-Adams) and further combined it with the {\em global triangle charging scheme} to achieve a $(1.994+\eps)$-approximation.  This charging scheme was later further developed and used to analyze {\em global triangle distributions}, which is discussed at the end of  Section~\ref{sec:overview-rounding}. 

Dealing with additive errors from the correlated rounding is one of the most technical
parts \cite{CLN22} and perhaps a limitation to
getting an improved approximation ratio (e.g., special attention must
be paid to edges with values close to $0$ and $1$ in both the
algorithm and in the analysis). That brings us to our second main idea of {\em preclustering}, so that we can (1) precluster the instance so that the cost of
the optimal solution becomes relatively large compared to the size of
the ``unsure part'', and (2) use correlated rounding only for the
unsure part of the instance so that the error from the correlated
rounding is negligible.


\paragraph{Preclustering.} 
Preclustering is a subroutine that aims at identifying ``clear clusters''. 
Formally, our preclustering outputs a pair $(\calK, E_\adm)$, where $\calK$ is a partition of $V$ into so-called \emph{atoms} and $E_\adm \subseteq {V \choose 2}$ is a set of \emph{admissible} edges.
The guarantee is that for any $\eps > 0$,  there is a $(1+\eps)$-approximate solution that is {\em consistent} with the preclustering; an atom is not broken, and if $u$ and $v$ are not in the same atom and $uv \notin E_\adm$, then $u$ and $v$ are in different clusters. We also ensure that $|E_{\adm}| \leq O_{\eps}(\opt)$, so the number of edges on which we apply the correlated rounding (and pay the additive error) is comparable to $\opt$. 

The algorithm first constructs atoms. 
Intuitively, consider a cluster of the
optimum (integral) solution that has $s$ vertices but contributes significantly smaller than $s^2$ to the objective function, say $\eps s^2$. (Otherwise, an $\eps^2 s^2$-additive error will already be a $\eps$-multiplicative error.)
For this to happen, it must be that the number of \medges internal to $S$
is small compared to ${s \choose 2}$ and that the number \pedges with
exactly one endpoint in $S$ is also much smaller than ${s \choose 2}$. This cluster is thus almost a clique with tiny
\pedge-expansion, which is easy to (approximately) identify. 
This idea was used by 
Cohen-Addad, Lattanzi, Mitrovic, Norouzi-Fard, Parotsidis, and Tarnawski for a massively-parallel algorithm with a large constant factor approximation~\cite{CLMNP21}. 
Our first version of preclustering~\cite{CLLN23} used \cite{CLMNP21} as a black-box to construct atoms, but in this paper, we present the simplified version of~\cite{CCLLNV24}, which starts from an arbitrary $O(1)$-approximate clustering and removes a vertex $v$ from a cluster $C$ when the $+$neighborhood of $v$ is somewhat different from $C$ (and makes $v$ a singleton cluster). 


The set of admissible edges is roughly defined as follows: we
construct a graph $(V, E^1)$ where two vertices are neighbors if their
$+$degrees are similar (up to a factor of $O(1/\eps)$). Then an edge
$uv$ is admissible if $u$ and $v$ have many (at least an
$\Omega(\eps)$ fraction) common neighbors in $E^+ \cap E^1$. Here, the
somewhat relaxed notions of degree-similarity and
large-common-neighborhood (compared to the ones for the atoms, where
we wanted $+$neighborhood of $u$ and $v$ to be almost identical for
$u$ and $v$ to be in the same atom) ensure that there is a
near-optimal solution respecting the admissible edges, while also
guaranting that the number of the admissible edges (crossing atoms) is
not too large compared to $\opt$.

\paragraph{Solving Cluster LP on Preclustered Instance by Sampling.}  
After constructing a preclustered instance $(\calK, E_\adm)$, we formulate an LP relaxation aimed at finding a $(1+\eps)$-approximate \emph{good clustering} for $(\calK, E_\adm)$, that we call the \emph{bounded sub-cluster LP}. 
First introduced in~\cite{CLLN23}, it is a strengthened version of the Sherali-Adams hierarchy used in~\cite{CLN22}. 
With a solution $(x, y)$ to the LP, we run a procedure that constructs a single cluster $C$ randomly. The probability that any vertex is in $C$ is precisely $1/y_\emptyset$, where $y_\emptyset$ is the fractional number of clusters in $y$.  The probabilities that exactly one of $u$ and $v$ is in $C$, and both of them are in $C$, are respectively $\frac{x_{uv}}{y_\emptyset}$ and $\frac{1 - x_{uv}}{y_\emptyset}$ up to some error terms arising from the correlated rounding. As usual, $x_{uv}$ is the extent in which $u$ and $v$ are separated.

To construct the solution $z=(z_S)_{S \subseteq V}$ for the cluster LP, we generate $y_\emptyset \Delta$ many clusters $C$ independently, for a large enough polynomial $\Delta$.  Roughly speaking, the solution $z$ is $\frac1\Delta$ times the multi-set of clusters $C$ we generated. The error incurred by the correlated rounding procedure can be bounded in terms of $|E_\adm|$, and the error from sampling can be bounded using concentration bounds. 

\subsection{Rounding Cluster LP}
\label{sec:overview-rounding}

In this section, we describe our algorithm for obtaining the improved approximation ratio for \cc. We solve the cluster LP using Theorem~\ref{thm:solving-cluster-LP} to get a fractional solution $z = (z_S)_{S \subseteq V}$, which determines $x \in [0, 1]^{V \choose 2}$ as in \eqref{LPC:set-define-x}: $x_{uv} := 1 - \sum_{S \supseteq \{u, v\}} z_S$ for every $uv \in {V \choose 2}$. We have $\obj(x) \leq (1+\eps) \opt$. The theorem will be proved in Section~\ref{sec:solve-LP}. With $z$, we then run two procedures: the cluster-based rounding and the pivot-based rounding with threshold $\tau=0.4$. We select the better result as the final clustering.  The two procedures are defined in Algorithms~\ref{alg:cluster-based} and \ref{alg:correlated-rounding} respectively.  We use $N^+(u)$ and $N^-(u)$ to denote the sets of $+$ and $-$neighbors of a vertex $u \in V$ respectively. They are stated as randomized algorithms, but they can be derandomized by the standard method of conditional expectation. 


\begin{algorithm}
    \caption{Cluster-Based Rounding}
        \label{alg:cluster-based}
    \begin{algorithmic}[1]
        \State $\calC \gets \emptyset, V' \gets V$
        \While{$V' \neq \emptyset$}
            \State randomly choose a cluster $S\subseteq V$, with probabilities $\frac{z_S}{\sum_{S'} z_{S'}}$
            \If{$V' \cap S \neq \emptyset$} $\calC \gets \calC \cup \{V' \cap S\}$, $V' \gets V' \setminus S$ \EndIf
        \EndWhile
        \State \Return $\calC$
    \end{algorithmic}    
\end{algorithm}

\begin{algorithm}
\caption{General pivot-based rounding with hybrid independent/dependent updates. 
Parameters: $I_c^+\subseteq[0,1]$ and functions $f^+,f^-:[0,1]\to[0,1]$. 
For a pivot $u$, $+$edges $(u,v)$ with $x_{uv}\notin I_c^+$ are rounded independently (include $v$ with probability $1-f^+(x_{uv})$); 
$+$edges with $x_{uv}\in I_c^+$ are rounded dependently via a sample $S\ni u$ drawn from $\{z_S\}$; 
$-$edges are always rounded independently (include $v$ with probability $1-f^-(x_{uv})$).}
\label{alg:general-correlated-rounding}
\begin{algorithmic}[1]
    \State $\mathcal{C} \gets \emptyset$, $V' \gets V$
    \While{$V' \neq \emptyset$}
        \State pick a pivot $u \in V'$ uniformly at random; \quad $C \gets \{u\}$
        \State sample a set $S \ni u$ according to probabilities $\{z_S\}$ \Comment{$\sum_{S \ni u} z_S = 1$}
        \For{each $v \in V' \cap N^+(u)$}
            \If{$x_{uv} \notin I^+_c$} \Comment{independent rounding on out-of-range $+$ edges}
                \State independently add $v$ to $C$ with probability $1 - f^+(x_{uv})$ 
            \ElsIf{$v \in S$} \Comment{dependent rounding on in-range $+$ edges}
                \State add $v$ to $C$ 
            \EndIf
        \EndFor
        \For{each $v \in V' \cap N^-(u)$}
            \State independently add $v$ to $C$ with probability $1 - f^-(x_{uv})$ \Comment{independent rounding on $-$ edges}
        \EndFor
        \State $\mathcal{C} \gets \mathcal{C} \cup \{C\}$; \quad $V' \gets V' \setminus C$
    \EndWhile
    \State \Return $\mathcal{C}$
\end{algorithmic}
\end{algorithm}
\paragraph{Analysis of Cluster-Based Rounding Procedure.} The cluster-based rounding procedure is easy to analyze.  The following lemma suffices.
\begin{lemma}
    For every $uv \in {V \choose 2}$, the probability that $u$ and $v$ are separated in the clustering $\calC$ output by the cluster-based rounding procedure is $\frac{2x_{uv}}{1+x_{uv}}$. So the probability they are in the same cluster is $\frac{1-x_{uv}}{1+x_{uv}}$.
\label{lem:cluster-based-rounding}
\end{lemma}
\begin{proof}
    We consider the first set $S$ chosen in the cluster-based rounding algorithm such that $\{u, v\} \cap S \neq \emptyset$. $u$ and $v$ will be separated iff $|S \cap \{u, v\}| = 1$. The probability that this happens is precisely $\frac{\sum_{|S \cap \{u, v\}| = 1}z_S}{\sum_{S \cap \{u, v\} \neq \emptyset}z_S} = \frac{2 x_{uv}}{1+x_{uv}}$.
\end{proof}
Therefore, a $+$edge $uv$ will incur a cost of $\frac{2x_{uv}}{1+x_{uv}}$ in expectation in the cluster-based rounding procedure, and a $-$edge will incur a cost of $\frac{1-x_{uv}}{1+x_{uv}}$.  The approximation ratios for a $+$edge $uv$ and a $-$edge $uv$ are respectively $\frac{2}{1+x_{uv}}$ and $\frac{1}{1+x_{uv}}$. Notice that the latter quantity is at most $1$. 

\paragraph{Notations and Analysis for Pivot-Based Rounding Procedure.} We now proceed to the pivot-based rounding procedure in Algorithm~\ref{algo:pivot}.\footnote{We remark that to recover the correlated rounding algorithm in \cite{CLN22} and \cite{CLLN23}, we can use $C \gets \emptyset$ in Step 4.
Then we can obtain their approximation ratios without the complication
of handling rounding errors.  The errors are handled in \cite{CLN22}
by distinguishing between the short, median and long $+$edges. In our
algorithm, we also distinguish between \emph{short} $+$edges (those
with $x_{uv} \leq \frac13$) and \emph{long} $+$edges (those with
$x_{uv} > \frac13$); however, the purpose of this distinction is to
get an improved approximation ratio, instead of to bound the rounding
errors.}
Our high-level setup of the analysis is based on~\cite{ACN08}
and~\cite{CMSY15} with some additional notions.  We consider a general
{\em budget} for every edge.
We shall define two \emph{budget functions}:
\begin{itemize}
    \item $\budgetp:[0, 1] \to \R_{\geq 0}$ and $\budgetn:[0, 1] \to \R_{\geq 0}$.
\end{itemize}
They determine the \emph{budget} $b_{uv}$ for the edge $uv$: if $uv \in E^+$, then $b_{uv} := \budgetp(x_{uv})$, and if $uv \in E^-$, then $b_{uv} := \budgetn(x_{uv})$.



Algorithm~\ref{alg:correlated-rounding} presents a generalized pivot-based rounding scheme parameterized by an interval $I_c^+\subseteq[0,1]$ and functions $f^+,f^-$. 
Unlike the standard pivot rounding method, our scheme uses $I_c^+$ to designate the regime in which we apply \emph{dependent} rounding (and applies independent rounding otherwise). 
By choosing different parameter triples $(I_c^+,f^+,f^-)$, the algorithm instantiates a family of rounding procedures. 
We give two concrete instantiations that achieve different approximation guarantees: one more elementary and amenable to by-hand analysis, and one slightly more involved that yields a stronger approximation ratio.

We now focus on one iteration of the while loop in Algorithm~\ref{alg:correlated-rounding}. Suppose $u, v, w \in V'$ at the beginning of the iteration, and let $C$ be the cluster constructed at the end. We use $u$ to denote the event that $u$ is chosen as the pivot.  We say $vw$ incurs a cost in the iteration, if $vw \in E^+$ and $|C \cap \{v, w\}| = 1$, or $vw \in E^-$ and $\{v, w\} \subseteq C$. Then we define 
\begin{align*}
  \cost_u(v, w) := \Pr[vw\text{ incurs a cost} \mid u].
\end{align*}
and 
\begin{align*}
    \lp_u(v, w):= \Pr[C \cap \{v, w\} \neq \emptyset \mid u] \cdot b_{vw}.
\end{align*}
$\cost_u(v, w)$ is the probability that $vw$ incurs a cost conditioned on the event $u$. When an edge $vw$ disappears, we say $vw$ \emph{releases} its budget. So, $\lp_u(v, w)$ is the expected budget released by $vw$ in the iteration when $u$ is the pivot.  Notice that both $\cost_u(v, w)$ and $\lp_u(v, w)$ do not depend on $V'$, provided that $u, v, w \in V'$.

We call a set of three distinct vertices a {\em triangle}. A set of two distinct vertices is called a {\em degenerate triangle}.  For triangle $(u, v, w)$, let
\begin{align*}
  \cost(u, v, w) := \cost_u(v, w) + \cost_v(u, w) + \cost_w(u, v),  \quad \text{and} \quad \lp(u, v, w) := \lp_u(v, w) + \lp_v(u, w) + \lp_w(u, v).
\end{align*}
For degenerate triangle $(u, v)$, let 
\begin{align*}
  \cost(u, v) := \cost_u(u, v) + \cost_v(u, v), \quad \text{and} \quad \lp(u, v) := \lp_u(u, v) + \lp_v(u, v).
\end{align*}
 
\begin{lemma}
    \label{lemma:triangle-implies-budget}
    Suppose that for every $V' \subseteq V$, we have 
        \begin{align}
        \sum_{(u, v, w) \in \binom{V'}{3}} \cost(u,v, w) + \sum_{(u, v) \in \binom{V'}{2}} 
        \cost(u,v)
        \leq 
        \sum_{(u, v, w) \in \binom{V'}{3}} \lp(u,v, w) + \sum_{(u, v) \in \binom{V'}{2}} \lp(u,v).
        \label{eq:triangle-sum}
    \end{align}
    Then, the expected cost of the clustering output by Algorithm~\ref{alg:correlated-rounding} is at most $\sum_{uv \in {V \choose 2} } b_{uv}$. 
\end{lemma}
\begin{proof} 
Focus on any iteration of Algorithm~\ref{alg:correlated-rounding}; $V'$ is the $V'$ at the beginning of the iteration.  

The expected cost incurred by all edges in the iteration is
\begin{align*}
    &\frac{1}{|V'|} \sum_{u \in V'} \sum_{(v, w) \in {V' \choose 2}} \cost_u(v, w) =\frac1{|V'|}\sum_{(u, v, w) \in {V' \choose 3}} \cost(u, v, w) + \frac1{|V'|}\sum_{(u, v) \in {V' \choose 2}} \cost(u, v).
\end{align*}

The expected budget released at this iteration is
\begin{align*}
    &\frac{1}{|V'|} \sum_{u \in V'} \sum_{(v, w) \in {V' \choose 2}} \lp_u(v, w)  = \frac1{|V'|}\sum_{(u, v, w) \in {V' \choose 3}} \lp(u, v, w) + \frac1{|V'|}\sum_{(u, v) \in {V' \choose 2}} \lp(u, v).
\end{align*}

Therefore, if the condition of the lemma holds, then at every iteration of Algorithm~\ref{alg:correlated-rounding}, the expected cost incurred is at most the expected budget released. Overall, the expected cost of the final clustering is at most the expected total budget released by all edges during the whole procedure, which is $\sum_{uv \in {V \choose 2}}b_{uv}$. This finishes the proof of the lemma.
\end{proof}

To obtain an approximation ratio of $\alpha \in [1, 2)$, we consider a variant of our algorithm, in which we run the cluster-based rounding procedure (Algorithm~\ref{alg:cluster-based}) with probability $\frac\alpha2$, and the pivot-based rounding procedure with threshold $\tau=0.4$ (Algorithm~\ref{alg:correlated-rounding}) with the remaining probability $1-\frac\alpha2$.    Clearly, the actual algorithm that picks the better of the two clusterings generated can only be better. We set up the budget functions $\budgetp$ and $\budgetn$ such that every edge pays a cost of at most $\alpha$ times its LP cost in expectation. That is, the following properties are satisfied for every $x \in [0, 1]$: 
\begin{align*}
    \frac{\alpha}{2} \cdot \frac{2x}{1 + x} + \left(1 - \frac\alpha 2\right) \budgetp(x) = \alpha x, \qquad \frac{\alpha}{2} \cdot \frac{1-x}{1 + x} + \left(1 - \frac\alpha 2\right) \budgetn(x) = \alpha (1-x).
\end{align*}
This gives us the following definitions: 
\begin{align}
    \budgetp_\alpha(x):= \frac{\alpha}{1-\alpha/2} \cdot \frac{x^2}{1+x}, \qquad \text{and} \qquad \budgetn_\alpha(x):=\frac{\alpha}{1-\alpha/2} \cdot \frac{(1+2x)(1-x)}{2(1+x)}, \qquad \forall x \in [0, 1].\label{def:budget_func}
\end{align}

\begin{lemma}
    \label{lemma:budgets-for-alpha}
    If the budget functions $\budgetp_\alpha$ and $\budgetn_\alpha$ satisfy \eqref{eq:triangle-sum} for some $\alpha \in [1, 2)$, then our algorithm has an approximation ratio of $\alpha$. 
\end{lemma}
\begin{proof}
    Consider the variant of the algorithm where we run the cluster-based rounding procedure with probability $\frac\alpha2$, and the pivot-based procedure with threshold $\tau=0.4$ with the remaining probability of $1-\frac\alpha2$. By Lemma~\ref{lemma:triangle-implies-budget}, the expected cost of the clustering given by the variant is at most
    \begin{align*}
        &\quad \sum_{uv \in E^+} \left(\frac\alpha2 \cdot \frac{2x_{uv}}{1+x_{uv}} + 
        \big(1-\frac\alpha2\big)\cdot \budgetp_\alpha(x_{uv}) \right) + \sum_{uv \in E^-} \left(\frac\alpha 2 \cdot \frac{1-x_{uv}}{1+x_{uv}} + \big(1-\frac\alpha 2\big) \cdot \budgetn_\alpha(x_{uv})\right) \\
        &=\alpha\left(\sum_{uv \in E^+} x_{uv} + \sum_{uv \in E^-} (1-x_{uv})\right) = \alpha \cdot \obj(x).
    \end{align*} 
    The actual algorithm we run can only be better than this variant. 
\end{proof}

\begin{algorithm}
    \caption{Pivot-Based Rounding with Threshold $\tau = 0.4$ and $I^+_c = (\tau, 1]$}
    \label{alg:correlated-rounding}
    \begin{algorithmic}[1]
        \State $\calC \gets \emptyset, V' \gets V$
        \While{$V' \neq \emptyset$}
            \State pick a pivot $u \in V'$ uniformly at random; \quad $C \gets \{u\}$
            \State sample a set $S \ni u$ according to probabilities $\{z_S\}$ \Comment{$\sum_{S \ni u} z_S = 1$}
            \For{each $v \in V' \cap N^+(u)$}
                    \If{$x_{uv} \leq \tau$} \Comment{independent rounding on out-of-range $+$ edges}
                        \State independently add $v$ to $C$  \Comment{$f^+(x) = 1$}
                    \ElsIf{$v \in S$} \Comment{dependent rounding on in-range $+$ edges}
                        \State add $v$ to $C$ 
                    \EndIf
                \EndFor
        \For{each $v \in V' \cap N^-(u)$}
            \State independently add $v$ to $C$ with probability $1 - x_{uv})$ \Comment{$f^-(x) = x$}
        \EndFor
        \State $\mathcal{C} \gets \mathcal{C} \cup \{C\}$; \quad $V' \gets V' \setminus C$
        \EndWhile
        \State \Return $\calC$
    \end{algorithmic}
    \label{algo:pivot}
\end{algorithm}


As a baseline, we instantiate Algorithm~\ref{alg:general-correlated-rounding} to obtain Algorithm~\ref{alg:correlated-rounding}. 
In this instantiation we set $I_c^+=(0.4,1]$: consequently, all $+$-edges $(u,v)$ with $x_{uv}\le 0.4$ are rounded independently, while $+$-edges with $x_{uv}\in I_c^+$ participate in the dependent rounding step. (All $-$-edges are rounded independently.)

We provide a per-triangle analysis leading to an approximation ratio of $\woSArat$ in Section~\ref{sec:pureClusterAnalyticalProof}:
\begin{restatable}{lemma}{lemmabudgetsforonehalf}
    \label{lemma:budgets-withoutSA}
    For Algorithm~\ref{alg:correlated-rounding} and budget functions $\budgetp \equiv \budgetp_{\woSArat}$ and $\budgetn \equiv \budgetn_{\woSArat}$, we have $\cost(T) \leq \lp(T)$ for every triangle $T$.
\end{restatable}

Clearly, the lemma implies that \eqref{eq:triangle-sum} holds for $\budgetp \equiv \budgetp_{\woSArat}$ and $\budgetn \equiv \budgetn_{\woSArat}$. By Lemma~\ref{lemma:budgets-for-alpha}, Algorithm~\ref{alg:correlated-rounding} gives an approximation ratio of $\woSArat$.  


\smallskip

To get a better approximation ratio, we provide an analysis that uses global distributions of triangles. 
It relies on solving a factor-revealing SDP.
The resulting instantiation is given in Algorithm~\ref{alg:pivot3}. 
For a \(+\) edge \((u,v)\), we proceed as follows:
(i) if \(x_{uv}\le 0.40\), include \(v\) in the pivot cluster deterministically;
(ii) if \(0.40 < x_{uv}\le 0.57\), include \(v\) via the dependent step using the pivot-centered set distribution \(\{z_S\}\);
(iii) if \(x_{uv} > 0.57\), include \(v\) independently with probability \(1-x_{uv}\).
For a \(-\) edge \((u,v)\), we split the endpoints with probability \(x_{uv}^2\) (equivalently, add \(v\) to the pivot’s cluster with probability \(1-x_{uv}^2\)).
For algorithm~\ref{alg:pivot3}, the following lemma is proved in Section~\ref{sec:1485approximate}.


\begin{algorithm}
    \caption{Pivot-Based Rounding (with correlated mid-range for \(+\) edges)}
    \label{alg:pivot3}
    \begin{algorithmic}[1]
        \State \(\mathcal{C} \gets \emptyset\), \(V' \gets V\)
        \While{\(V' \neq \emptyset\)}
                \State pick a pivot $u \in V'$ uniformly at random; \quad $C \gets \{u\}$
                 \State sample a set $S \ni u$ according to probabilities $\{z_S\}$ \Comment{$\sum_{S \ni u} z_S = 1$}
        \For{each $v \in V' \cap N^+(u)$}
            \If{$x_{uv} \leq 0.4 $} \Comment{independent rounding on out-of-range $+$ edges}
                \State independently add $v$ to $C$ \Comment{$f^+(x) = 1$ for $x \leq 0.4$} 
            \ElsIf{$x_{uv} > 0.57 $} \Comment{independent rounding on out-of-range $+$ edges}
                \State add $v$ to $C$ with probability \(1 - x_{uv}\) \Comment{$f^+(x) = x$ for $x > 0.57$} 
            \ElsIf{$v \in S$} \Comment{dependent rounding on in-range $+$ edges}
                \State add $v$ to $C$ 
            \EndIf
        \EndFor
        \For{each $v \in V' \cap N^-(u)$}
            \State independently add $v$ to $C$ with probability $1 - x_{uv}^2$ \Comment{$f^-(x) = x^2$}
        \EndFor
        \State $\mathcal{C} \gets \mathcal{C} \cup \{C\}$; \quad $V' \gets V' \setminus C$
        \EndWhile
        \State \Return \(\mathcal{C}\)
    \end{algorithmic}
\end{algorithm}

\begin{restatable}{lemma}{lemmabudgetsforpureclusterlpSDPmethod}
\label{lemma:budgets-for-pureclusterlp-sdpbudget}
For Algorithm~\ref{alg:pivot3}, \eqref{eq:triangle-sum} holds for budget functions $\budgetp \equiv \budgetp_{\pureclusterlpratio}$ and $\budgetn \equiv \budgetn_{\pureclusterlpratio}$. 
\label{lem:analysis_1.485}
\end{restatable}

Combined with Lemma~\ref{lemma:budgets-for-alpha}, it implies Theorem~\ref{thm:main:algo:computer}.

\paragraph{Global Triangle Distributions.}

The idea of global triangle distributions, which is a basis for Lemma~\ref{lem:analysis_1.485}, is systematically expressed by a {\em factor-revealing SDP}. 
Given a cluster LP solution $z_S$ and vertices $u, v, w$, let us define $y_{uv} := \sum_{S \supseteq  \{u,v\}} z_S$ (resp. $y_{uvw} := \sum_{S \supseteq \{u, v, w\}} z_S)$ be the probability that $u,v$ (resp. $u,v,w$) are in the same cluster. Given any quadruple $T = (a, b, c, d) \in [0, 1]^4$,  let $\eta_T$ represent the number of triangles $(u,v ,w)$ such that of $y_{uv} = a, y_{uw} = b, y_{vw} = c, y_{uvw} = d$. 

Consider a hypothetical rounding procedure, where given a pivot $u$, the cluster $C$ that contains $u$ is simply chosen with probability $z_C$ (note that $\sum_{C \ni u} z_C = 1$). 
Let $X_{v}$ denote the event that node $v$ is included in the cluster of node $u$ in this rounding. We can show $\E[X_{v}\cdot X_{w}] = y_{uvw}$ and $\E[X_{v}] \cdot \E[X_{w}] = y_{uv}y_{uw}$. The covariance matrix $COV_u$, where $COV_u(v, w) = \E[X_{v}\cdot X_{w}] - \E[X_{v}] \cdot \E[X_{w}] = y_{uvw} - y_{uv}y_{uw}$, must be positive semidefinite (PSD). This PSD constraint on the covariance matrix enforces a stronger constraint on $\eta_T$. For instance, if all non-degenerate triangles centered at $u$ are \ppm triangles with $y$ value $(y_{uv} = 0.5, y_{uw} = 0.5, y_{wv} = 0, y_{uvw} = 0)$, then the covariance matrix of $COV_u$ cannot be PSD because $COV_u(v,w) = y_{uvw} - y_{uv}y_{uw} = -0.25$ for almost all non-diagonal entries.

For a triangle $T = (y_{uv}, y_{uw}, y_{vw}, y_{uvw})$, we discretize $y_{uv}, y_{uw}, y_{vw}$ to incorporate the PSD constraint. We partition the interval $[0, 1]$ into numerous subintervals $I_1, I_2,...,I_t$. Each triangle with $y$ value $(y_{uv} \in I_i, y_{uw} \in I_j, y_{vw} \in I_k, y_{uvw})$ is placed in one of these interval combinations. We can rearrange $COV_u$ as $Q_u \in \mathbb{R}^{t \times t}$, where $Q_u(I_i, I_j) = \sum_{y_{uv} \in I_i, y_{uw} \in I_j} (y_{uvw} - y_{uv}y_{uw})$. Considering $Q = \sum_{u \in V} Q_u$, we can represent $Q$ using $T$ and $\eta_T$. The PSD property of $Q_u$ implies $Q$ is PSD, thus enforcing a constraint on $\eta_T$.

Despite there being infinitely many types of triangles in each range
$I_i, I_j, I_k$, our key observation is that $y_{uvw} - y_{uv}y_{uw}$
is multilinear. Therefore, we only need a few triangles in each range
to represent all possible triangles. We want to mention the triangles
we need are fixed so they can be precomputed and the only unsure variable is $\eta_T$. To compute a lower bound $\sum \eta_T(\Delta(T) - \text{cost}(T))$, we set up a semidefinite program (SDP) under the constraint that $Q$ is PSD. This SDP is independent of \ref{LP:clusterlp} and relies on the chosen interval and budget function. By employing a practical SDP solver, we demonstrate that $\sum \eta_T(\Delta(T) - \text{cost}(T)) \geq 0$.

\subsection{Gaps and Hardness}
\label{sec:overview-hardness}
A high-level intuition for the cluster LP is the following: (any) LP cannot distinguish between a random graph and a nearly bipartite graph. 
For the cluster LP, given a complete graph $H = (V_H, E_H)$ with $n = |V_H|$, our \cc instance is $G = (V_G, E_G)$ where $V_G = E_H$ and $e, f \in V_G$ are connected by a \pedge in $G$ if they share a vertex in $V$. Consider vertices of $H$ as {\em ideal clusters} in $G$ containing their incident edges.
The LP fractionally will think that it is nearly bipartite, implying that 
the entire $E_H$ can be partitioned into $n/2$ ideal clusters of the same size. Of course, integrally, such a partition is not possible in complete graphs. 

For the cluster LP, it suffices to consider a complete graph instead of a random graph. We believe (but do not prove) that such a gap instance can be extended to stronger LPs (e.g., Sherali-Adams strengthening of the cluster LP), because it is known that Sherali-Adams cannot distinguish a random graph and a nearly bipartite graph~\cite{charikar2009integrality}.

The idea for the NP-hardness of approximation is the same. The main difference, which results in a worse factor here, is that other polynomial-time algorithms (e.g., SDPs) can distinguish between random and nearly bipartite graphs! So, we are forced to work with slightly more involved structures.

Still, we use a similar construction for $3$-uniform hypergraphs; let $H = (V_H, E_H)$ be the underlying $3$-uniform hypergraph and $G = (V_G, E_G)$ be the plus graph of the final \cc instance where $V_G = E_H$ and $e, f \in E_H$ has an edge in $G$ if they share a vertex in $H$. 
We use the hardness result of Cohen-Addad, Karthik, and Lee~\cite{
cohen2022johnson} that shows that it is hard to distinguish whether $H$ is {\em nearly bipartite}, which implies that half of the vertices intersect every hyperedge, or close to a random hypergraph. 

\paragraph{Organization.}  
In Section~\ref{sec:precluster}, we describe our preclustering procedure. It will be used in our algorithm to solve the cluster LP described in Section~\ref{sec:solve-LP}. This finishes the proof of Theorem~\ref{thm:solving-cluster-LP}.
We prove Lemma~\ref{lem:analysis_1.485} in 
Section~\ref{sec:1485approximate}, which proves Theorem~\ref{thm:main:algo:computer}. 
It is a computer-assisted proof, and we present an analytic proof of 
Lemma~\ref{lemma:budgets-withoutSA}, giving a bound of $\woSArat$, in 
Section~\ref{sec:pureClusterAnalyticalProof}. 
We give the $(\frac43 - \eps)$-integrality gap of the cluster LP (Theorem~\ref{thm:main:gap}) in Section~\ref{sec:gap}, and the improved hardness of $24/23 - \eps$ (Theorem~\ref{thm:main:hardness}) in Section~\ref{sec:hardness}.

\paragraph{Global Notations.} For two sets $A$ and $B$, we use $A \triangle B = (A \setminus B) \cup (B \setminus A) $ to denote the symmetric difference between $A$ and $B$. We used $N^+_u$ and $N^-_u$ to denote the sets of $+$ and $-$neighbors of a vertex $u$ respectively in the \cc instance.  For a clustering $\calC$ of $V$, we define $\obj(\calC)$ to be the objective value of $\calC$.  For any $x \in [0, 1]^{{V \choose 2}}$, we already defined $\obj(x) = \sum_{uv \in E^+} x_{uv} + \sum_{uv\in E^-}(1-x_{uv})$. Recall that we defined $\cost_u(v, w), \lp_u(v, w), \cost(T)$ and $\lp(T)$ for a triangle $T = (u, v, w)$ or a degenerate triangle $T = (u, v)$ in this section; they depend on the budget functions $\budgetp$ and $\budgetn$.

\section{Preclustering}
\label{sec:precluster}

An important technique we use to solve the cluster LP is preclustering. We define a preclustered instance and a good clustering as follows. 

\begin{definition}\label{definition:prepro}
Given a \cc instance $(V, E^+ \uplus E^-)$, a {\em preclustered instance} is defined by a pair $(\calK, E_{\adm})$, where  $\calK$ is a partition of $V$ (which can also be viewed as a clustering), and $E_{\adm} \subseteq \binom{V}{2}$ is a set of pairs such that for every $uv \in E_\adm$, $u$ and $v$ are not in a same set in $\calK$. 

Each set $K \in \calK$ is called an {\em atom}. An (unordered) pair $uv$ between two vertices $u$ and $v$ in a same $K \in \calK$ is called an {\em atomic edge}; in particular, a self-loop $uu$ is an atomic edge. A pair that is neither an atomic nor an admissible edge is called a \emph{non-admissible} edge.   
\end{definition}


\begin{definition}
Given a preclustered instance $(\calK, E_\adm)$ for some \cc instance $(V, E^+ \uplus E^-)$,  a clustering $\calC$ of $V$ is called {\em good} with respect to $(\calK, E_{\adm})$ if 
	\begin{itemize}
	    \item $u$ and $v$ are in the same cluster in $\calC$ for an atomic edge $uv$, and
	    \item $u$ and $v$ are not in the same cluster in $\calC$ for a non-admissible edge $uv$.
	\end{itemize}
\end{definition}

We show that we can efficiently produce a preclustered instance for which there is a $(1+\epsilon)$-approximate good clustering. The rest of this section is dedicated to the proof of the following theorem:
\begin{restatable}{theorem}{thmpreclustering}
    \label{thm:preprocessed-wrapper}
    For any sufficiently small $\eps > 0$, there exists a $\text{poly}(n, \frac1\eps)$-time algorithm that, given a \cc instance $(V, E^+ \uplus E^-)$ with optimal value $\opt$ (which is not given to us), produces a preclustered instance $(\calK, E_{\adm})$ such that 
    \begin{itemize}
        \item there exists a good clustering w.r.t $(\calK, E_\adm)$, whose cost is at most $(1 + \eps) \opt$, and 
        \item $|E_{\adm}| \leq O\big(\frac1{\eps^2}\big)\cdot \opt$. 
    \end{itemize}        
\end{restatable}



For convenience, we assume every $u$ has a self-loop in $E^+$. Let $N^+_u$ be the set of $+$neighbors of $u$; so we have $u \in N^+_u$. Let $d^+_u = |N^+_u|$ be its $+$degree. Let $\calC^*$ be the optimum clustering, and $\opt = \obj(\calC^*)$ be its cost. We assume $\eps > 0$ is at most a small enough constant.

\subsection{Constructing Atoms $\calK$}
In the first step of the algorithm for the proof of Theorem~\ref{thm:preprocessed-wrapper}, we define the set $\calK$ of atoms.  We use any known $O(1)$-approximation algorithm for the \cc problem to obtain a clustering; for example, we can use the $3$-approximation combinatorial algorithm of \cite{ACN08} to obtain a clustering $\calC$ with $\obj(\calC) \leq 3\cdot \opt$. 

Our $\calK$ is obtained from $\calC$ by marking some vertices and creating singletons for them. We view $\calK$ both as a clustering and as the set of atoms. The algorithm is described in Algorithm~\ref{alg:construct-calK}, where $\beta = 0.1$.

\begin{algorithm}[H]
	\caption{Construction of $\calK$}
	\label{alg:construct-calK}
	\begin{algorithmic}[1]
		\For{every non-singleton $C \in \calC$}
			\For{every $u \in C$}: mark $u$ if $|N^+_u \triangle C| > \frac{\beta}2\cdot |C|$  \EndFor
			\If {at least $\frac{\beta |C|}{3}$ vertices in $C$ are marked} \label{step:mark-individual} mark all vertices in $C$ \EndIf
		\EndFor
		\State let $\calK$ be the clustering obtained from $\calC$ by removing marked vertices and creating a singleton cluster for each of them
		\State \Return $\calK$
	\end{algorithmic}
\end{algorithm}

\begin{lemma}
	\label{lemma:cost-calC'}
	$\obj(\calK) \leq O\left(1\right) \cdot \opt$, where we view $\calK$ as a clustering.
\end{lemma}
\begin{proof}
	If a cluster $C$ has fewer than $\frac{\beta |C|}{3}$ marked vertices before Step~\ref{step:mark-individual}, then the cost incurred by separating the marked vertices in $C$ is at most $\frac2{\beta} = O(1)$ times the cost of edges incident to these vertices in $\calC$.  On the other hand, if $C$ has at least $\frac{\beta |C|}{3}$ marked vertices before Step~\ref{step:mark-individual}, the cost of edges incident to $C$ in $\calC$ is at least $\frac{\beta}{2}\cdot |C| \cdot \frac{\beta |C|}{3} \cdot \frac12  = \Omega(|C|^2)$.  The cost incurred by breaking $C$ into singletons is at most $|C| \choose 2$.  As every edge is charged at most twice, the cost incurred by creating singletons for all marked vertices is at most $O(1) \cdot \obj(\calC)$.  So $\obj(\calK) \leq O(1) \cdot \opt$ as $\obj(\calC) \leq 3\cdot \opt$.
\end{proof}

\begin{lemma}
	\label{lemma:K-dense}
	For every non-singleton atom $K \in \calK$, and every $u \in K$, we have $|N^+_u \triangle K| < \beta |K|$.
\end{lemma}

\begin{proof}
	Assume $K \subseteq C$ for some $C \in \calC$. So the vertices in $C \setminus K$ are marked, the vertices in $K$  are unmarked, and $|C \setminus K| < \frac{\beta |C|}{3}$. For every $u \in K$, we have $|N^+_u \triangle C| \leq \frac{\beta}{2} \cdot |C|$ because it is unmarked. Then $|N^+_u \triangle K| \leq \frac{\beta}{2} \cdot |C| + \frac{\beta |C|}{3} = \frac{5\beta}{6} \cdot |C| \leq \frac{5\beta}{6} \cdot \frac{|K|}{1-\beta/3} \leq \beta |K|$ for our choice of $\beta$.
\end{proof}

\begin{lemma}
	\label{lemma:optimum-not-break-atom}
	Consider the optimum clustering $\calC^*$. Any atom $K \in \calK$ is completely inside a  cluster in $\calC^*$. 
\end{lemma}

\begin{proof}
    The lemma holds trivially if $K$ is a singleton.  Assume towards the contradiction that $K$ is a non-singleton and not inside a cluster in $\calC^*$.  By Lemma~\ref{lemma:K-dense}, we have that $|N^+_u \triangle K| < \beta |K|$ for every $u \in K$.  In particular, this implies that $|K \setminus N^+_u| < \beta |K|$ and $|N^+_u \setminus K| < \beta |K|$.  We consider two cases as in \cite{CLLN23}. 
	
	First, suppose no cluster in $\calC^*$ contains at least $\frac{2|K|}{3}$ vertices in $K$.   We consider the operation of removing all vertices in $K$ from their respective clusters in $\calC^*$, and creating a single cluster $K$.  The saving in cost is at least $\frac12 \cdot |K| \cdot (\frac{|K|}{3} - \beta |K|) - \frac12 \cdot |K| \cdot \beta |K| = (\frac{1}{6}-\beta)|K|^2 > 0$, as $\beta = 0.1$. This contradicts that $\calC^*$ is optimum. 
	
	Then, consider the other case where there is some $C \in \calC$ with $|C\cap K|\geq \frac{2|K|}{3}$. Let $u$ be any vertex in $K \setminus C$; it must exist as $K$ is not completely inside $C$.  Then, we consider the operation of moving $u$ from its cluster to $C$. The saving in cost is at least $(\frac{2}{3} - 2\beta)|K| - \frac{|K|}{3} = (\frac13-2\beta) |K| > 0$ as $\beta = 0.1$, contradicting that $\calC^*$ is optimum. 
\end{proof} \smallskip

With Lemma~\ref{lemma:optimum-not-break-atom}, we now restrict ourselves to clusterings that do not break atoms.  As a result, we can then assume all the edges between two vertices in a same atom $K \in \calK$ are $+$edges. The assumption only decreases $\opt$ and can only make the two properties of Theorem~\ref{thm:preprocessed-wrapper} harder to satisfy.    We use $K_u$ for every $u \in V$ to denote the atom that contains $u$. 
We let $k_u = |K_u|$. 

It is convenient for us to distribute the $+$edges incident to an atom $K \in \calK$ equally to the vertices in $K$. For every $u, v \in V$, we define
\begin{align*}
	w_{uv} := \frac{1}{k_uk_v}\sum_{u' \in K_u, v' \in K_v}1_{u'v' \in E^+}.
\end{align*}
to be the probability that an edge between a random vertex in $K_u$ and a random vertex in $K_v$ is a $+$edge. So $w_{uv} \in [0, 1]$ for every pair $uv$; in particular, if $v \in K_u$, we have $w_{uv} = 1$. We shall call $w_{uv}$ the weight of the edge $uv$. For any set $u \in V, V' \subseteq V$, we define $w(u, V'):=\sum_{v \in V'} w_{uv}$ be the total weight of edges between $u$ and $V'$.  Let $w_u := w(u, V)$ be the total weight of all edges incident to $u$.    In this new instance,  $w_{uv}$ fraction of edge $uv$ has a $+$ sign, and the remaining $1-w_{uv}$ fraction has a $-$ sign; $w_{u}$ is the total fractional number of $+$edges incident to $u$, which can be treated as the $+$degree of $u$. This new instance is equivalent to the original one. Till the end of this section we focus on this instance and call it the \emph{averaged instance} to distinguish it from the original one.

\subsection{Defining Admissible Edges $E_\adm$}
In this section, we define the set $E_\adm$ of admissible edges, such that the following holds. There is a $(1+O(\eps))$-approximate clustering $\calC^*_1$ that does not break atoms, and all edges between two different clusters in $\calC^*_1$ are admissible. We shall bound $|E_\adm|$ in terms of $\obj(\calK)$, which is at most $O(1) \cdot \opt$. Notice that the construction of $E_\adm$ is a part of our algorithm for the proof of Theorem~\ref{thm:preprocessed-wrapper}, but the construction of $\calC^*_1$ is not. So, the latter procedure can depend on the optimum clustering $\calC^*$. 

We give the definition of $E_\adm$ upfront: 
\begin{align}
    E^1 &:= \Big\{uv: \eps w_v < w_u < \frac{w_v}{\eps}\Big\}, \label{equ:define-E1} \\
    E^2 &:= \Big\{uv: K_u = K_v \text{ or } uv \in E^1 \text{ and } \sum_{p \in N^1_u \cap N^1_v} w_{up}w_{vp}> \eps (w_u + w_v)\Big\},\label{equ:define-E2}\\
    E_\adm &:= \Big\{uv \in E^2: K_u \neq K_v\Big\}. \label{equ:define-Eadm}
\end{align}
Above, $uv$ is an unordered pair in $V$, with $u = v$ allowed. As a sanity check, notice that the conditions for the three definitions are symmetric between $u$ and $v$.  Notice that $E_\adm \subseteq E^2 \subseteq E^1$. We use $N^1_u, N^2_u$ to denote the neighbor sets of $u$ in graphs $(V, E^1)$ and $(V, E^2)$ respectively. We use $d^2_u = |N^2_u|$ to denote the degree of $u$ in $(V, E^2)$. \medskip

We then proceed to show the existence of a clustering $\calC^*_1$ with desired properties. 
\begin{lemma}
	\label{lemma:calC*1}
	There is a clustering $\calC^*_1$ of cost at most $(1+O(\eps))\opt$ that does not break any atom $K \in \calK$, and satisfies the following condition: For every $u, C$ with $K_u \subsetneq C \in \calC^*_1$, we have  $w_{u, C} > \frac{|C|}{2} +   \eps \cdot w_u$.
\end{lemma}


\begin{proof}
	Start with $\calC^*_1 = \calC^*$; notice that it does not break any atom.  While the condition does not hold for some $C \in \calC^*_1$ and $u \in C$, i.e., $w_{u, C} \leq \frac{|C|}{2} +   \eps \cdot w_u$,  we update $\calC^*_1 \gets \calC^*_1 \setminus \{C\} \cup \{C \setminus K_u, K_u\}$. This finishes the construction of $\calC^*_1$; clearly it does not break any atom. 
	
	We then consider how much cost increment the procedure incurs. Focus on any iteration of the while loop.  At the beginning of the iteration, the cost of edges incident to any vertex in $K_u$ in the clustering $\calC^*_1$ (w.r.t the averaged instance) is
	\begin{align*}
		&\quad |C|  - w(u, C) + w(u, V \setminus C) = w(u, V) + |C| - 2w(u, C)\geq w_u + |C| - 
		2\left(\frac{|C|}2 + \eps \cdot w_u\right) = (1 - 2\eps) w_u.
	\end{align*}
	The inequality is by the condition of while loop. 
	
	 The cost increment incurred by separating $K_u$ and $C \setminus K_u$ in the iteration is
	 \begin{align*}
		&\quad k_u \cdot \big(w(u, C \setminus K_u) - (|C \setminus K_u| - w(u, C \setminus K_u))\big) = k_u (2w(u, C \setminus K_u) - |C \setminus K_u|)\\
   &\leq k_u \cdot \big(2 w (u, C) - |C|\big) \leq 2\eps k_u\cdot w_u.  
	\end{align*}
	Again, the second inequality is by the condition of the while loop.
	
	We can charge the cost increment using the $(1-2\eps)k_u \cdot  w_u$ cost from edges incident to $u$: every unit cost is used to charge $\frac{2\eps}{1-2\eps}$ unit cost increment. Notice that every edge is charged at most twice.  So, we have $\obj(\calC^*_1) - \obj(\calC) \leq \frac{4\eps}{1-2\eps} \cdot \obj(\calC^*_1)$. This implies that $\obj(\calC^*_1) \leq \big(1+O(\eps)\big) \obj(\calC^*) = (1+O(\eps))\opt$. 
\end{proof}

We let $\calC^*_1$ be the $(1 + O(\eps))$-approximate clustering satisfying the properties of Lemma~\ref{lemma:calC*1}.  

\begin{lemma}
	\label{lemma:degree-close}
	For every cluster $C \in \calC^*_1$ and two vertices $u, v \in C$, we have $w_u > \eps w_v$. 
\end{lemma}
\begin{proof}
	Assume $v \notin K_u$ since otherwise $w_u = w_v$. So $C$ contains both $K_u$ and $K_v$. Lemma~\ref{lemma:calC*1} implies $w_u \geq w_{u, C} > \frac{|C|}{2}$. On the other hand, it implies $\eps \cdot w_v < w(v, C) - \frac{|C|}{2} \leq  |C| - \frac{|C|}2 = \frac{|C|}2$. Therefore,  $w_u > \eps w_v$.
\end{proof}

By Lemma~\ref{lemma:degree-close} and the definition of $E^1$ in \ref{equ:define-E1}, we have 
\begin{coro}
	\label{coro:E1}
	Every $C \in \calC^*_1$ forms a clique (with self-loops) in $(V, E^1)$. 
\end{coro}

\begin{lemma}
	\label{lemma:large-overlap}
	For any cluster $C \in \calC^*_1$, any two vertices $u, v \in C$ with $v \notin K_u$ have $\sum_{p \in N^1_u \cap N^1_v}  w_{up} w_{vp}> \eps \cdot (w_u + w_v )$. 
\end{lemma}

\begin{proof}
	Notice that $K_u, K_v \subseteq C$. By Corollary~\ref{coro:E1}, we have $C \subseteq N^1_u$ and $C \subseteq N^1_v$.  By Lemma~\ref{lemma:calC*1}, we have $w(u, C) > \frac{|C|}{2} + \eps w_u$, and $w(v, C) > \frac{|C|}{2} + \eps w_v$. Therefore,
	\begin{align*}
		\sum_{p \in N^1_u \cap N^1_v} w_{up}w_{vp} &\geq \sum_{p \in C} w_{up}w_{vp} \geq \sum_{p \in C} (w_{up} + w_{vp} - 1) = w(u, C) + w(v, C) - |C|\\
		&> \frac{|C|}{2} + \eps w_u + \frac{|C|}{2} + \eps w_v - |C| = \eps(w_u + w_v). 
	\end{align*}
	The second inequality is by that $w_{vp}, w_{up} \in [0, 1]$. 
\end{proof}

Corollary~\ref{coro:E1}, Lemma~\ref{lemma:large-overlap} and the definition of $E^2$ in \eqref{equ:define-E2} imply
\begin{coro}
	\label{coro:E2}
	Every $C \in \calC^*_1$ forms a clique (with self-loops) in $(V, E^2)$.
\end{coro}
So, by the definition of $E_\adm$ in \eqref{equ:define-Eadm}, $\calC^*_1$ is a good clustering w.r.t the preclustered instance $(\calK, E_\adm)$.  \medskip

We then proceed to bound $|E_\adm|$.
\begin{lemma}
	\label{lemma:d2-small}
	$d^2_u - k_u \leq O(\frac{1}{\eps^2}) (w_u - k_u)$.
\end{lemma}
\begin{proof}
	Notice that $d^2_u - k_u$ is precisely the number of edges $uv \in E^2$ with $v \notin K_u$. 
	Consider the graph $(V, E^1)$ with edge weights $w$: Every edge $vp \in E^1$ has weight $w_{vp}$. We define the weight of a $2$-edge path $u$-$p$-$v$ in $(V, E^1)$ to be $w_{up}w_{pv}$. By the definition of $E^2$, if $v \notin K_u$, then $uv \in E^2$ only if the total weight of $2$-edge paths of the form $u$-$p$-$v$ is at least $\eps (w_u + w_v) > \eps w_u$. 
	
	First, consider the paths $u$-$p$-$v$ with $p \notin K_u$. The total weight of all such paths (over all $p$ and $v$) is at most $(w_u - k_u) \cdot \frac{w_u}{\eps}$. This holds since the total weight of edges between $u$ and $V \setminus K_u$ in the complete graph is $w_u - k_u$, and any neighbor $p$ of $u$ in $(V, E^1)$ has $w_p < \frac{w_u}{\eps}$ by the definition of $E^1$.  Then consider the paths $u$-$p$-$v$ with $p \in K_u$ and $v \notin K_u$. The total weight of all such paths is at most $k_u (w_u - k_u)$ as any $p \in K_u$ has $w_p = w_u$ and $k_p = k_u$.  
	
	Therefore, the total number of verticies $v \notin K_u$ with $uv \in E^2$ is at most 
	\begin{equation*}
		\frac{(w_u - k_u) \cdot \frac{w_u}{\eps} + k_u(w_u - k_u)}{\eps w_u} \leq \frac1{\eps^2}(w_u - k_u) + \frac{1}{\eps}(w_u - k_u) = O\big(\frac{1}{\eps^2}\big) (w_u - k_u). \hfill \qedhere
	\end{equation*}	
\end{proof}

The total number of admissible edges is 
\begin{align*}
	|E_\adm| = \frac{1}{2} \sum_{u \in V} (d^2_u - k_u) \leq O\big(\frac{1}{\eps^2}\big) \sum_{u \in V}(w_u - k_u) = O\big(\frac{1}{\eps^2}\big)  \cdot \obj(\calK) \leq O\big(\frac{1}{\eps^2}\big) \cdot \opt.
\end{align*}
The second equality is by that $\obj(\calK) = \frac{1}{2}\sum_{u \in V}(w_u - k_u)$ and the last inequality is by Lemma~\ref{lemma:cost-calC'}. We can scale $\eps$ down by a constant at the beginning, so that the cost of $\calC^*_1$ is at most $(1+\eps)\opt$. This finishes the proof of Theorem~\ref{thm:preprocessed-wrapper}.

\section{Solving Cluster LP Approximately}
\label{sec:solve-LP}


In this section, we show how to solve the cluster LP in polynomial time, by proving Theorem~\ref{thm:solving-cluster-LP}, which is repeated below.
{\renewcommand{\footnote}[1]{}\thmsolveclusterLP*}

We define some global parameters used across this section.  Let $\eps_1 = \eps^3, \epsrt = \eps_1^2 = \eps^6$, and $r = \Theta(1/\epsrt^2) = \Theta(1/\eps^{12})$ be an integer, with some large enough hidden constant.  The subscript ``rt'' stands for Raghavendra-Tan.

We apply Theorem~\ref{thm:preprocessed-wrapper} to obtain a preclustered instance $(\calK, E_\adm)$, with the unknown good clustering $\calC^*_1$. We can assume in the preclustered instance $(\calK, E_\adm)$, the edges between two different atoms $K$ and $K'$ are all admissible, or all non-admissible. If one edge between them is non-admissible, we can change all other edges to non-admissible edges. This will not change the set of good clusterings, and it will decrease $|E_\adm|$.

We define $K_u$ to be the atom that contains $u$, and $k_u = |K_u|$.  We shall use $N_\adm(u)$ to be the set of vertices $v$ such that $uv \in E_\adm$; so $N_\adm(u) = N_\adm(v)$ if $v \in K_u$. We further process the good clustering $\calC^*_1$ using the following procedure. It is not a part of our algorithm; it is only for analysis purpose. 
\begin{algorithm}[H]
    \begin{algorithmic}[1]
        \While {there exists some $K_u$ in a cluster $C \in \calC^*_1$ with $k_u < |C| \leq k_u + \eps_1 \cdot |N_\adm(u)|$}
            \State $\calC^*_1 \gets \calC^*_1 \setminus \{C\} \cup \{K_u, C \setminus K_u\}$
        \EndWhile
    \end{algorithmic}
\end{algorithm}

\begin{claim}
    The procedure increases $\obj(\calC^*_1)$ by at most $2\eps_1\cdot|E_\adm|$.
\end{claim}
\begin{proof}
    Whenever we break $C$ into $K_u$ and $C \setminus K_u$ in the procedure, the cost increase is at most $k_u \cdot (|C|  - k_u) \leq k_u \cdot \eps_1 \cdot |N_\adm(u)| = \eps_1 \sum_{v \in K_u} |N_\adm(v)|$. We separate each atom $K_u$ at most once. Therefore, the total cost increase is at most $\eps_1 \sum_{v \in V} |N_\adm(v)| = 2\eps_1\cdot |E_\adm|$.
\end{proof}

So, the cost of $\calC^*_1$ after the procedure will be at most $(1 + \eps)\opt + O(\eps_1) |E_\adm|$. Crucially, the following property is satisfied:
\begin{enumerate}[label=(A\arabic*)]
	\item \label{property:forbid} For every $u \in V$, $K_u$ is either a cluster in $\calC^*_1$, or in a cluster of size more than $k_u + \eps_1\cdot |N_\adm(u)|$. 
\end{enumerate} \smallskip

\subsection{Bounded Sub-Cluster LP Relaxation for Preclustered Instances}
We form an LP relaxation aiming at finding the good clustering
$\calC^*_1$. In the LP, we have a variable $y^s_S$, for every $s \in
[n]$, and $S \subseteq V$ of size at most $r$ (recall that $r =
\Theta(1/\eps^{12})$), that denotes the number of clusters in
$\calC^*_1$ of size $s$ containing $S$ as a subset. When $S \neq
\emptyset$, there is at most one such cluster and thus $y^s_S \in \{0,
1\}$ indicates if $S$ is a subset of a cluster of size $s$ in
$\calC^*_1$.  For every $S \subseteq V$ of size at most $r$, let $y_S
:= \sum_s y^s_S$ denote the number of clusters (of any size) in $
\calC^*_1$ containing $S$ as a subset. Again, if $S \neq \emptyset$,
then $y_S \in \{0, 1\}$ indicates if $S$ is a subset of a cluster in $
\calC^*_1$. For every $uv \in {V \choose 2}$, we have a variable $
x_{uv}$ indicating if $u$ and $v$ are separated or not in $
\calC^*_1$. We call the LP the \emph{bounded sub-cluster LP
relaxation}, as we have variables indicating if a small set $S$ is a
subset of a cluster or not.  \smallskip

We use the following type of shorthand:  $y^s_{u}$ for $y^s_{\{u\}}$, $y^s_{uv}$ for $y^s_{\{u, v\}}$, and $y^s_{Su}$ for $y^s_{S \cup \{u\}}$. The bounded sub-cluster LP is defined as follows. In the description, we always have $s \in [n], u \in V$ and $uv \in {V \choose 2}$. For convenience, we omit the restrictions.    By default, any variable of the form $y_S$ or $y^s_S$ has $|S| \leq r$; if not, we do not have the variable and the constraint involving it. 
	\begin{equation}
		\min \qquad \obj(x)  \tag{bounded sub-cluster LP} \label{LP:subset}
	\end{equation} \vspace*{-28pt}
	
	\noindent\begin{minipage}[t]{0.3\textwidth}
	    \begin{align}
	    	\sum_{s = 1}^{n} y^s_S &= y_S   & &\forall S \label{LPC:subset-define-yS}\\
	    	y_u &= 1 & &\forall u \label{LPC:subset-a-contained}\\
	    	y_{uv} + x_{uv} &= 1 & &\forall uv \label{LPC:subset-define-x}\\
	    	\frac1s\sum_{u} y^s_{Su} &=  y^s_S &\quad &\forall s, S \label{LPC:subset-size-s}\\
	    	y^s_S &\geq 0 & &\forall s, S \label{LPC:subset-non-negative}
	    \end{align}
	\end{minipage}\hfill
	\begin{minipage}[t]{0.6\textwidth}
	    \begin{align}
	         x_{uv} &= 0 & & \forall u, v \text{ in a same } K \in \calK \label{LPC:subset-atom} \\[8pt]
	         x_{uv} &=1 & &\forall \text{non-admissible edge }uv \label{LPC:subset-non-adm}\\[3pt]
	        y^s_u &=0 & &\forall u, s \in  [k_u-1] \cup \big[k_u+1, k_u + \eps_1 |N_\adm(u)|\big]\label{LPC:subset-forbid} 
	    \end{align}\vspace*{-20pt}
	    \begin{align}
	       \sum_{T' \subseteq T}(-1)^{|T'|}y^s_{S\cup T'} &\in [0, y^s_S] & &\forall s, S\cap T =\emptyset
	       \label{LPC:subset-correlation}
	    \end{align}
	\end{minipage} \bigskip

 \eqref{LPC:subset-define-yS} gives the definition of $y_S$,
 \eqref{LPC:subset-a-contained} requires $u$ to be contained in some
 cluster, and \eqref{LPC:subset-define-x} gives the definition of
 $x_{uv}$.  \eqref{LPC:subset-size-s} says if $y^s_S = 1$, then there
 are exactly $s$ elements $u \in V$ with $y^s_{Su} = 1$. (An exception
 is when $S = \emptyset$; the equality also holds in this case.)
 \eqref{LPC:subset-non-negative} is the non-negativity
 constraint. \eqref{LPC:subset-atom} and \eqref{LPC:subset-non-adm}
 follows from that $\calC^*_1$ is a good clustering, and
 \eqref{LPC:subset-forbid} follows from \ref{property:forbid}. The
 left side of \eqref{LPC:subset-correlation} is the number of clusters
 of size $s$ containing $S$ but does not contain any vertex in $T$ by
 inclusion-exclusion principle.  So the inequality holds. This
 corresponds to a Sherali-Adams relaxation needed for the correlated
 rounding~\cite{RT12}, see Lemma~\ref{lem:RT-KT}.  The running time
 for solving the LP is $n^{O(r)} = n^{O(1/\eps^{12})}$.

\subsection{Sampling One Cluster Using LP Solution to the Bounded Sub-Cluster LP}
We solve the bounded sub-cluster LP to obtain the $y$ and $x$ vectors.  Given $y$, we can use the procedure construct-cluster described in Algorithm~\ref{alg:subset-construct-C}, to produce a random cluster $C$. Notice that steps 1 and 2 of the algorithm are well-defined due to equalities~\eqref{LPC:subset-define-yS} and \eqref{LPC:subset-size-s}.
\begin{algorithm}
	\caption{Construct-Cluster$(y)$}
	\label{alg:subset-construct-C}
	\begin{algorithmic}[1]
		\State randomly choose a cardinality $s$, so that $s$ is chosen with probability $\frac{y^s_\emptyset}{y_\emptyset}$
		\State randomly choose a vertex $u \in V$, so that $u$ is chosen with probability $\frac{y^s_u}{s y^s_\emptyset}$; call $u$ the \emph{pivot}
		\State define a vector $y'$ such that $y'_S = \frac{y^s_{Su}}{y^s_u}$ for every $S \subseteq V$ of size at most $r-1$
		\State\label{step:subset-RT}apply the Raghavendra-Tan correlated rounding technique over the fractional set $y'$ to construct a cluster $C \subseteq V$ that does not break any atom, and \Return $C$
	\end{algorithmic}
\end{algorithm}

With \eqref{LPC:subset-correlation}, the Raghavendra-Tan technique can be applied: 
\begin{lemma}[\cite{RT12}]
\label{lem:RT-KT}
  In Step~\ref{step:subset-RT} of Algorithm~\ref{alg:subset-construct-C}, one can sample a set $C \subseteq V$ that does not break atoms in time $n^{O(r)}$ such that
  \begin{itemize}
  \item For each $v \in V$, $\Pr[v \in C] = y'_{v}$. 
  \item $\frac{1}{|N_\adm(u)|^2 }\sum_{v, w \in N_\adm(u)}\big|\Pr[v, w \in C] - y'_{vw}\big| \leq \epsrt$.
  \end{itemize}
\end{lemma}
Recall that $\epsrt = \Theta(1/\sqrt{r})$ and the hidden constant inside $\Theta(\cdot)$ is large enough. 
\medskip

We define $\err^s_{vw|u}$ to be the error generated by the procedure when we choose $s$ as the cardinality and $u$ as the pivot:
\begin{align*}
    \err^s_{vw|u} := \left|\Pr\big[v, w \in C|s, u\big] - \frac{y^s_{uvw}}{y^s_u}\right|, \forall vw \in {V \choose 2},
\end{align*}
and 
\begin{align*}
	\err^s_{vw}:=\frac{1}{sy^s_\emptyset}\sum_{u \in V}y^s_u\cdot\err^s_{vw|u}
	 \text{ and } \err_{vw} := \sum_{s} \frac{y^s_\emptyset}{y_\emptyset}\cdot \err^s_{vw}
\end{align*}
as the error for $vw$ conditioned on $s$, and the unconditioned error. Notice that all these quantities are expectations of random variables, and thus deterministic.

We prove the following two lemmas.
\begin{lemma}
    \label{lemma:subset-prob-clustered}
    For any $v \in V$, we have $\Pr[v \in C] =\frac1{y_\emptyset}$.
\end{lemma}
\begin{proof}
    The probability is 
    \begin{align*}
        \sum_{s}\frac{y^s_\emptyset}{y_\emptyset}\sum_{u \in V'}\frac{y^s_u}{sy^s_\emptyset} \cdot \frac{y^s_{uv}}{y^s_u} = \frac1{y_\emptyset}\sum_{s}\frac1s\sum_{u \in V'}y^s_{uv} = \frac1{y_\emptyset}\sum_{s} y^s_v = \frac1{y_\emptyset}y_v = \frac1{y_\emptyset}. 
    \end{align*}
    The second equality is by \eqref{LPC:subset-size-s}. The third and the last inequalities are by \eqref{LPC:subset-define-yS} and \eqref{LPC:subset-a-contained} respectively. 
\end{proof}

\begin{lemma}
	\label{lemma:bound-correlation-probabilities}
    Focus on an edge $vw \in {V \choose 2}$. 
    \begin{enumerate}
        \item $\Pr\left[v \in C, w \notin C\right] \leq \frac{1}{y_\emptyset} \cdot  x_{vw} + \err_{vw}$.
        \item $\Pr\left[\{v, w\} \subseteq C \right] \leq \frac{1}{y_\emptyset} \cdot y_{vw} + \err_{vw}$.
    \end{enumerate}
\end{lemma}

\begin{proof}
    We focus on the first statement. 
    The probability that $v \in C$ and $w \notin C$ conditioned on $s$ is at most
    \begin{align*}
    	\sum_{u \in V}\frac{y^s_u}{s y^s_\emptyset}\cdot
        \left(\frac{1}{y^s_u} \cdot \big(y^s_{uv} -
        y^s_{uvw}\big) + \err^s_{vw|u}\right) & =\sum_{u \in V}\left(\frac{1}{sy^s_\emptyset}\cdot(y^s_{uv} - y^s_{uvw}) + \frac{y^s_u}{sy^s_{\emptyset}}\cdot\err^s_{vw|u}\right)\\
    	&=\frac{1}{y^s_\emptyset}(y^s_{v}-y^s_{vw}) + \err^s_{vw}.
    \end{align*}

        To see the second equality, we apply \eqref{LPC:subset-size-s}
        with $S = \{v\},\{w\}$ and $\{v, w\}$ respectively, and use the definition
    of $\err^s_{vw}$.
        Deconditioning on $s$, we have that the probability $vw$ is
        split by $C$ 
        is at most
    \begin{align*}
    	\sum_{s}\frac{y^s_\emptyset}{y_\emptyset}\cdot
        \left(\frac{1}{y^s_\emptyset}(y^s_v-y^s_{vw}) +
        \err^s_{vw}\right) &=  \frac{1}{y_\emptyset}\sum_{s}(y^s_v -y^s_{vw}) + \err_{vw}\\
    	&=  \frac1{y_\emptyset} (1 - y_{vw} ) + \err_{vw} =   \frac1{y_\emptyset} (x_{vw}) + \err_{vw}.
    \end{align*}
    The second equality used that $\sum_{s} y^s_v = y_v = 1$ and $\sum_{s}y^s_{vw} = y_{vw} = 1-x_{vw}$.

For the second statement, the probability that $v,w \in C$ conditioned
on $s$ is at most
    \begin{align*}
    	\sum_{u \in V}\frac{y^s_u}{s y^s_\emptyset}\cdot
      \left(\frac{1}{y^s_u} \cdot y^s_{uvw}  + \err^2_{vw|u}\right)
&=\sum_{u \in V}\left(\frac{1}{sy^s_\emptyset}\cdot y^s_{uvw} 
+ \frac{y^s_u}{sy^s_{\emptyset}}\cdot\err^s_{vw|u}\right)\\
    	&=\frac{1}{y^s_\emptyset}\cdot y^s_{vw} + \err^s_{vw}.
    \end{align*}
Deconditioning on $s$, as before, results in the desired bound of
$y_{vw}/y_\emptyset + \err_{vw}$.
\end{proof}

We can then bound the total unconditional error over all $vw$ pairs:
\begin{lemma}
    \label{lemma:bound-err}
    $\displaystyle \sum_{vw \in {V \choose 2}} \err_{vw}\leq O(\eps_1) \cdot \frac1{y_\emptyset} |E_\adm|$.
\end{lemma}

\begin{proof} Throughout the proof, we assume $u, v, w$ are all in $V$, $vw$ and $uw$ are in $V \choose 2$.  

Fix some $s \in [n], u \in V$ with $y^s_u > 0$, and we now bound $\sum_{vw} \err^s_{vw|u}$.  If $s = k_u$, then $C = K_u$; no errors will be created and the quantity is $0$. Assume $s > k_u$.  By \eqref{LPC:subset-forbid}, we have that $s > k_u + \eps_1\cdot |N_\adm(u)|$, since otherwise we shall have $y^s_u = 0$. 
 By the second property of Lemma~\ref{lem:RT-KT}, we have $\sum_{vw} \err^s_{vw|u} \leq \frac\epsrt2 |N_\adm(u) |^2$. (Notice that if one of $v$ and $w$ is not in $N_\adm(u)$, then $\err^s_{vw|u} = 0$.) Recall that $\epsrt = \eps_1^2$.
 Therefore, 
 \begin{align}
		\sum_{vw \in {V \choose 2}} \err^s_{vw|u} &\leq \frac\epsrt2 \cdot |N_\adm(u) |^2  \leq  \frac{\epsrt}{2\eps_1} \cdot |N_\adm(u) |  \cdot (s - k_u) \nonumber\\
		&= \frac{\eps_1}2 \cdot |N_\adm(u) | \cdot \sum_{v \in N_\adm(u)}\frac{y^s_{uv}}{y^s_u}= \frac{\eps_1}2 \cdot \sum_{v, w \in N_\adm(u) } \frac{y^s_{uv}}{y^s_u}. \label{inequ:sum-vw-err-s-vw|u}
	\end{align} 
    The first equality is by \eqref{LPC:subset-size-s} and $y^s_{uv} = y^s_u$ for every $v \in K_u$. (To see this, notice that $y^s_{uv}  \leq  y^s_u$ is implied by \eqref{LPC:subset-correlation}. We have $y_{uv} = \sum_s y^s_{uv}$, $y_u = \sum_s y^s_u$, and $y_{uv} = y_u = 1$ if $v \in K_u$.) The second equality is obtained by replacing $|N_{\adm}(u)|$ with $\sum_{w \in N_\adm(u)}1$.

	Considering the inequalities over all $u \in V$, we have
	\begin{align}
		\sum_{vw} \err^s_{vw} &=  \frac1{sy^s_\emptyset} \sum_{u}y^s_u\cdot\sum_{vw}\err^s_{vw|u} \leq \frac{1}{sy^s_\emptyset} \sum_{u} y^s_u \cdot \sum_{v, w \in N_\adm(u) } \frac{\eps_1}{2} \cdot \frac{y^s_{uv}}{y^s_u}  = \frac{\eps_1}{2} \cdot \frac{1}{sy^s_\emptyset}\cdot \sum_{u \in V, v, w \in N_\adm(u) } y^s_{uv} \nonumber\\
		&=  \frac{\eps_1}{2}\cdot\sum_{v \in V}\frac{y^s_v}{s y^s_\emptyset} \sum_{u \in N_\adm(v), w \in N_\adm(u)}\frac{y^s_{uv}}{y^s_v} \quad\leq\quad \frac{\eps_1}{2} \cdot \sum_{v\in V}\frac{y^s_v}{s y^s_\emptyset}\sum_{uw \in E_\adm}  \left(\frac{y^s_{uv} + y^s_{vw}}{y^s_v}\right)\nonumber\\
		&\leq \eps_1 \cdot \sum_{v\in V}\frac{y^s_v}{s y^s_\emptyset}\sum_{uw \in E_\adm} \Pr[C \cap \{u, w\} \neq \emptyset \mid s, v \text{ is pivot}] \nonumber\\
		&= \eps_1 \sum_{uw \in E_\adm}\Pr[C \cap \{u, w\} \neq \emptyset \mid s]. \label{equ:err-sum-s-vw}
	\end{align}
    The first inequality is by \eqref{inequ:sum-vw-err-s-vw|u}. 
    For the second inequality in the above sequence, notice that we view $uw$ as an unordered pair. Thus we keep both $y^s_{uv}$ and $y^s_{vw}$ in the numerator. 
    To see the third inequality, notice that $\frac{y^s_{uv}}{y^s_v} = \Pr[u \in C|s, v\text{ is pivot}] \leq \Pr[C \cap \{u, w\} \neq \emptyset | s, v\text{ is pivot}]$. The same inequality holds for $\frac{y^s_{vw}}{y^s_v}$.

    Finally, we take all $s$ into consideration:
    \begin{align*}
        \sum_{vw} \err_{vw} &= \sum_{s}\frac{y^s_\emptyset}{y_\emptyset}\cdot \sum_{vw} \err^s_{vw} \leq \eps_1 \cdot \sum_s \frac{y^s_\emptyset}{y_\emptyset} \sum_{uw \in E_\adm} \Pr[C \cap \{u, w\} \neq \emptyset |s] \\
        &= \eps_1 \cdot \sum_{uw \in E_\adm} \Pr[C \cap \{u, w\} \neq \emptyset] \leq \frac{2\eps_1}{y_\emptyset} |E_\adm| + 3\eps_1 \sum_{uw \in {V\choose 2}} \err_{uw}.
    \end{align*}
    The first inequality is by \eqref{equ:err-sum-s-vw}. 
    To see the second inequality, we notice that $C \cap \{u, w\} \neq \emptyset$ is the union of the 3 disjoint events: $u \in C$ and $w \notin C$, $u \notin C$ and $w \in C$, and $\{u, w\} \subseteq C$. By Lemma~\ref{lemma:bound-correlation-probabilities}, we have $\Pr[C \cap \{u, w\} \neq \emptyset] \leq \frac{2x_{uw} + y_{uw}}{y_\emptyset} + 3\cdot \err_{uw} \leq \frac{2}{y_\emptyset} + 3\cdot \err_{uw}$.  So, we have $\sum_{vw} \err_{vw} \leq \frac{1}{1-3\eps_1}\cdot \frac{2\eps_1}{y_\emptyset}|E_\adm|$. This proves the lemma. 
\end{proof}


\subsection{Using Independently Sampled Clusters to Construct Cluster LP Solution}

With all the ingredients, we can now describe our algorithm for solving the cluster LP approximately, finishing the proof of Theorem~\ref{thm:solving-cluster-LP}.  Let $\Delta = \Theta\left(\frac{n^2 \log n}{\eps_1^2 |E_\adm|}\right)$ with a large enough hidden constant, and $\Delta y_\emptyset$ being an integer. (We assume $|E_\adm| \geq 1$ since otherwise the preclustered instance is trivial.)  We run Algorithm~\ref{alg:subset-construct-C} $\Delta y_\emptyset$ times independently to obtain clusters $C_1, C_2, \cdots, C_{\Delta y_\emptyset}$. 

We use the following variant of Chernoff bound.
\begin{theorem}
	Let $X_1, X_2, X_3, \cdots, X_n$ be independent (not necessarily iid) random varibles which take values in $[0, 1]$. Let $X = \sum_{i = 1}^n X_i, \mu = \E[X]$, and $\mu' \geq \mu$ be a real. Then for any $\delta \in (0, 1)$, we have 
	\begin{align*}
		\Pr[X < (1-\delta) \mu ] < e^{-\delta^2\mu/2}\qquad \text{and} \qquad \Pr[X > \mu + \delta \mu'] < e^{-\delta^2 \mu' /3}.
	\end{align*}
\end{theorem}

For every $u \in V$, let $R_u = \{t: u \in C_t\}$.  Notice that $\Delta y_\emptyset \cdot \frac{|E_\adm|}{y_\emptyset n^2} = \Theta\big(\frac{\log n}{\eps_1^2}\big)$, with a large enough hidden constant. Using Chernoff bound and union bound, we can prove that with probability at least $1-1/n$, the following conditions hold.
\begin{itemize}
	\item For every $u \in V$, we have $|R_u| \geq (1-\eps_1)\Delta y_\emptyset \cdot \frac{1}{y_\emptyset} = (1-\eps_1)\Delta$.
	\item For every $u, v \in V$ such that $uv \in E^+$, we have 
	\begin{align}
		|R_u \setminus R_v|  &\leq \Delta y_\emptyset \left(\frac{x_{uv}}{y_\emptyset} + \err_{uv} + \eps_1\cdot \max\left\{\frac{x_{uv}}{y_\emptyset} + \err_{uv}, \frac{|E_\adm|}{y_\emptyset n^2}\right\}\right) \nonumber\\
		&\leq (1+\eps_1)\Delta (x_{uv} +  y_\emptyset \err_{uv}) + \frac{\eps_1 \Delta|E_\adm|}{n^2}. \label{inequ:Ru-minus-Rv-upper-bound}
	\end{align}
	\item For every $uv \in E^-$, we have 
	\begin{align*}
		|R_u \cap R_v| &\leq \Delta y_\emptyset \left(\frac{y_{uv}}{y_\emptyset} + \err_{uv} + \eps_1\cdot 	\max\left\{\frac{y_{uv}}{y_\emptyset} + \err_{uv}, \frac{|E_\adm|}{y_\emptyset n^2}\right\}\right)\\
		&\leq (1+\eps_1)\Delta (y_{uv} + y_\emptyset \err_{uv}) + \frac{\eps_1 \Delta|E_\adm|}{n^2}.
	\end{align*}
\end{itemize}

From now on we assume the conditions hold. For every $u \in V$, we let $R'_u$ be the set of the $\ceil{(1-\eps)\Delta}$ smallest indices in $R_u$. Clearly, $|R'_u \cap R'_v| \leq |R_u \cap R_v|$.  We show $|R'_u \setminus R'_v|$ is still upper bounded by \eqref{inequ:Ru-minus-Rv-upper-bound}. 
\begin{claim}
    For every $uv \in E^+$ we have $\max\{|R'_u \setminus R'_v|, |R'_v \setminus R'_u|\} \leq (1+\eps_1)\Delta (x_{uv} +  y_\emptyset \err_{uv}) + \frac{\eps_1 \Delta|E_\adm|}{n^2}$.
\end{claim}
\begin{proof}
	For convenience, we use $B$ to denote the upper bound $(1+\eps_1)\Delta (x_{uv} +  y_\emptyset \err_{uv}) + \frac{\eps_1 \Delta|E_\adm|}{n^2}$. We think of $R'_u$ ($R'_v$ resp.) as obtained from the set $R_u$ ($R_v$ resp.) by removing the largest indices one by one.  Wlog we assume $|R_u| \geq |R_v|$; and thus initially $|R_v \setminus R_u|\leq |R_u \setminus R_v| \leq B$. We remove the elements from $R_u$ and $R_v$ in two stages.  
	
	In the first stage we do the following. While $|R_u| > |R_v|$, we remove the largest index from $R_u$. This can not increase $|R_u \setminus R_v|$. After the first stage, we have $|R_u \setminus R_v| = |R_v \setminus R_u| \leq B$.  
	 
	 In the second stage we do the following. While $|R_u| = |R_v| > \ceil{(1-\eps_1)\Delta}$, we remove the largest index in $R_u$ from $R_u$, and do the same for $R_v$.  Consider one iteration of the while loop. If the two indices are the same, then $|R_u \setminus R_v| = |R_v \setminus R_u|$ does not change.  Otherwise, wlog we assume the  index we removed from $R_u$ is larger.  Then removing the index in $R_u$ will decrease $|R_u \setminus R_v|$. So the iteration can not increase $|R_u \setminus R_v| = |R_v \setminus R_u|$. 
\end{proof}

Then, for every $t \in [1, \Delta y_\emptyset]$, we define $C'_t = \{u: t \in R'_u\} \subseteq C_t$; then every $v$ is contained in $C'_t$ for exactly $\ceil{(1-\eps)\Delta}$ values of $t$. We define $z_S = \frac{1}{\ceil{(1-\eps)\Delta}} \cdot |\{t:C'_t = S\}|$ for every $S \subseteq V$ with $S \neq \emptyset$.  
Define $\tilde x_{uv} = 1 - \sum_{\{u, v\} \subseteq S} z_S$ 
for every $uv \in {V \choose 2}$. Then $(\tilde x, z)$ is a valid solution to the cluster LP.  

For a $uv \in E^+$, we have 
\begin{align*}
	\tilde x_{uv} = \frac{1}{\ceil{(1-\eps)\Delta}} \cdot  |R'_u \setminus R'_v| \leq \frac{1+\eps_1}{1-\eps}(x_{uv} + y_\emptyset \err_{uv}) + \frac{\eps_1 |E_{\adm}|}{(1-\eps)n^2}.
\end{align*}
For a $uv \in E^-$, we have 
\begin{align*}
	(1 - \tilde x_{uv}) = \frac{1}{\ceil{(1-\eps)\Delta}} \cdot |R'_u \cap R'_v| \leq \frac{1+\eps_1}{1-\eps}(1-x_{uv} + y_\emptyset \err_{uv}) + \frac{\eps_1|E_\adm|}{(1-\eps)n^2}.
\end{align*}

Therefore, 
\begin{align*}
	\obj(\tilde x) &\leq (1+O(\eps))\left(\obj(x) + y_\emptyset \sum_{uv \in {V \choose 2}} \err_{uv}\right) + O(\eps_1) |E_\adm|\leq (1+O(\eps))\obj(x)  + O(\eps_1) |E_\adm|\\
	&\leq (1+O(\eps))\cdot \opt + O(\eps^3) \cdot O\big(\frac1{\eps^2}\big)\cdot\opt = (1+O(\eps)) \opt. 
\end{align*}

The second inequality is due to Lemma~\ref{lemma:bound-err}, and the third one used that $|E_\adm| \leq O\big(\frac1{\eps^2}\big)\cdot \opt$.  By scaling $\eps$, the upper bound can be made to $(1+\eps)\opt$. This finishes the proof of Theorem~\ref{thm:solving-cluster-LP}.

\section{Triangle Analysis used in 1.56-Approximation of Cluster LP}\label{sec:pureClusterAnalyticalProof}


The goal of this section is to prove Lemma \ref{lemma:budgets-withoutSA}, which is restated here.

\lemmabudgetsforonehalf*
The proof of Lemma \ref{lemma:budgets-withoutSA} follows from Lemma \ref{lemma:degtrianglescost}, \ref{lemma:+++trianglecost}, \ref{lemma:---trianglescost}, \ref{lemma:+--trianglescost}, \ref{lemma:++-trianglescost} which we prove in this section. First, we recall some notation.

\paragraph{Notation and Useful Observations.}
We define some notation and make some observations that will be useful
in this section and in Section \ref{sec:1485approximate}. Recall that
we are given a polynomial-sized LP solution $(z = (z_S)_{S \subseteq
  V}, x \in [0, 1]^{V \choose 2})$ to \eqref{LP:clusterlp}. Moreover, remember that we choose the threshold $\tau = 0.4$. For a
$+$edge $uv$, we say $uv$ is \emph{short} if $\xuv \leq \tau$ and
\emph{long} if $\xuv > \tau$.
Notice that the pivot-based rounding
procedure (Algorithm~\ref{alg:cluster-based}) treats short and long
$+$edges differently.  For vertices $u, v, w \in V$, we define
\begin{align*}
    y_{uv} &:= \sum_{S \supseteq \{u, v\}} z_S \in [0, 1], \qquad  y_{uvw} := \sum_{S \supseteq \{u, v, w\}} z_S \in [0, 1], \qquad  y_{uv|w} := y_{uv} - y_{uvw} \in [0, 1], \\
    y_{u|v|w} &:= 1 - y_{uv|w} - y_{uw|v} - y_{vw|u} - y_{uvw}.
\end{align*}
In an integral solution $z$, $y_{uv}, y_{uvw}, y_{uv|w}$ and $y_{u|v|w}$, respectively, indicate if $u$ and $v$ are in the same cluster, if $u, v$ and $w$ are all in the same cluster, if $u$ and $v$ are in the same cluster not containing $w$, and if $u, v, w$ are in three different clusters.  These definitions also hold when $u, v$ and $w$ are not distinct vertices. 

If a solution $(z_S)$ satisfies the 3-round Sherali-Adams constraints, then we would have $$y_{uv} + y_{uw} + y_{vw} = 1 + (2 y_{uvw} - y_{u|v|w})$$ and $y_{u|v|w} \in [0, 1]$.  However, for a solution to \ref{LP:clusterlp}, we have the following, weaker, claim.

\begin{lemma}\label{lem:weakerLemCLP}
The following holds in a solution of the cluster LP.  For any $\{u,v,w\}$, we have 
$$3 y_{uvw} \leq y_{uv} + y_{uw} + y_{vw} \leq \frac{3}{2} + \frac{3}{2}y_{uvw}.$$
\end{lemma}

\begin{proof}
For a vertex $u$, we have: $1-y_{uv} - y_{uw} + y_{uvw} \in [0,1]$.   This follows from the fact that $$\sum_{S \ni u} z_S \geq \sum_{S \ni u,v} z_S + \sum_{S \ni u,w}z_S - \sum_{S \ni u,v,w} z_S \quad \iff \quad 
1 \geq y_{uv} + y_{uw} - y_{uvw}.$$

Now we consider three equalities, one each for $u, v$ and $w$.  Adding up and dividing by 3 implies the lemma.
\end{proof}

In the following two claims, we focus on one iteration of the while loop in Algorithm~\ref{alg:correlated-rounding}. $V'$ is the vertex set at the beginning of the iteration, and $C$ is the cluster obtained at the end.  We use $u$ to denote the event that $u$ is the chosen pivot. 
\begin{claim} 
	Let $u, v \in V'$ (it is possible that $u = v$).  Then 
	\begin{align*}
		\Pr[v \in C \mid u]= \begin{cases}
			1 & \text{if  $uv$ is a short $+$edge}\\
			1 - \xuv = y_{uv} & \text{otherwise}
		\end{cases}.
	\end{align*}
\end{claim}

\begin{claim}
	Let $u, v, w \in V'$ with $v \neq w$. Then 
	\begin{align*}
		\Pr[\{v, w\} \subseteq C \mid u] &= \begin{cases}
			y_{uvw} & \text{if both $uv$ and $uw$ are long $+$edges}\\
			\Pr[v \in C \mid u] \cdot \Pr[w \in C \mid u] & \text{otherwise}
		\end{cases},\\
		\Pr[v \in C \wedge w \notin C \mid u] &= \begin{cases}
			y_{uv|w} & \text{if both $uv$ and $uw$ are long $+$edges}\\
			\Pr[v \in C \mid u] \cdot \big(1 - \Pr[w \in C \mid u]\big) & \text{otherwise}
		\end{cases}.
	\end{align*}

\end{claim}

We now prove some useful facts about the budget functions initially defined in Section
\ref{sec:overview-rounding}.
\begin{align*}
    \budgetp_\alpha(x) = \frac{\alpha}{1-\alpha/2} \cdot \frac{x^2}{1+x}, \qquad \text{and} \qquad \budgetn_\alpha(x)=\frac{\alpha}{1-\alpha/2} \cdot \frac{(1+2x)(1-x)}{2(1+x)}, \qquad \forall x \in [0, 1].
\end{align*}
\begin{lemma}
    For any $\alpha \in [1,2)$, $\budgetp_\alpha(x)$ is convex and increasing. 
    Moreover, $\budgetn_\alpha(x)$ is decreasing. 
    \label{lem:b_conv}
\end{lemma}
\begin{proof}
    Consider the first and second derivative of $\budgetp_\alpha(x)$,
    \begin{align*}
        \frac{d}{d x} \budgetp_\alpha(x) = \frac{\alpha}{1-\alpha/2} \cdot \frac{x (2 + x)}{(1 + x)^2} \geq 0,
        \\  \frac{d^2}{d^2 x} \budgetp_\alpha(x) = \frac{\alpha}{1-\alpha/2} \cdot \frac{2}{(x+1)^3} \geq 0.
    \end{align*}
    Here, we used that $\frac{\alpha}{1-\alpha/2} \geq 0$ and $x\geq 0$.
    Similarly, for $\budgetn_\alpha(x)$,
    \begin{flalign*}
        && \frac{d}{d x} \budgetn_\alpha(x) = \frac{\alpha}{1-\alpha/2} \cdot \frac{-x(x+2)}{(1+x)^2} \leq 0. && \qedhere
    \end{flalign*}
\end{proof} \medskip


We will use the following claim several times.
\begin{claim}\label{clm:budgetBound}
  For $x \geq \tau$, 
$\budgetp_{\woSArat}(x) \geq 2x.$
\end{claim}

\begin{cproof}
Our goal is to show that
  \begin{align*}
    \frac{2 \alpha}{2-\alpha} \frac{x^2}{1+x} \geq 2x.
  \end{align*}
Simplifying, we have
    \begin{align*}
    \frac{x}{1+x} \geq \frac{2-\alpha}{\alpha}.
    \end{align*}
    Solving this inequality, we have
    $$x \geq \frac{2-\alpha}{2\alpha-2}.$$
When $\alpha = \woSArat$, the inequality holds when $x \geq 0.392857$.
\end{cproof}

We will prove the lemma by considering each possible triangle type separately. That is, degenerate triangles, triangles with no $+$edge (\mmm~triangles), triangles with one $+$edge (\pmm~triangles),  triangles with two $+$edges (\ppm~triangles) and lastly, triangles with three $+$edges (\ppp~triangles). Moreover, recall that the pivot rounding algorithm treats short $+$edges differently than long $+$edges. Therefore, for each type of triangle, we consider different cases that specify how many short and long $+$edges the triangle contains.

\paragraph{Degenerate triangles.}
For the case when two of the vertices $u,v,w$ are identical, we want to make sure that $\cost(u,v) \leq \Delta(u,v)$. 
\begin{lemma}
\label{lemma:degtrianglescost}
Fix budget functions $\budgetp \equiv \budgetp_{\alpha}$ and $\budgetn \equiv \budgetn_{\alpha}$, where $\alpha \geq 4/3$. For any degenerate triangle $T$, we have $\cost(T) \leq \Delta(T)$.
\end{lemma}
\begin{proof}
For a $-$edge with value $x$ this is equivalent to
\begin{align*}
    (1-x) \leq \frac{\alpha(1+2x)(1-x)}{(2-\alpha)(1+x)}. 
\end{align*}
This holds true since $\frac{\alpha(1+2x)}{(2-\alpha)(1+x)} \geq \frac{\alpha}{(2-\alpha)} \geq 1$, for any $\alpha \geq 1$.
For short $+$edges the cost is 0 and the inequality holds trivially. For long $+$edges, with value $x > 0.4$,
we need,
\begin{align*}
    x \leq \frac{2\alpha x^2}{(2-\alpha)(1+x)}. 
\end{align*}
This holds true since $\frac{2\alpha x}{(2-\alpha)(1+x)} \geq \frac{2\alpha/3}{(2-\alpha)(1+1/3)} \geq 1$ for any $\alpha \geq 4/3$.
\end{proof}

\subsection{\ppp Triangles}
\begin{lemma}
\label{lemma:+++trianglecost}
Fix the budget function $\budgetp \equiv \budgetp_{\alpha}$, where $\alpha = \woSArat$. 
For any \ppp~triangle $T$, we have $\cost(T) \leq \Delta(T)$.
\end{lemma}
\begin{proof}
We distinguish between the following four cases: the triangle has three short edges, two short edges, one short edge or no short edge.
\paragraph{(s,s,s)}
$\cost(T) = 0$ and $\Delta(T) \geq 0$.
\paragraph{(s,s,l)} 
For a \ppp~triangle $T=(u,v,w)$, assume that $uv$ is long and $uw$, $vw$ are short.  Let $a = x_{uv}, b= x_{uw}$ and $c = x_{vw}$.  Then
\begin{align*}
\cost(T) & = 2\tuv,
\\ \Delta(T) & = \budgetp(\tuv) + \budgetp(\tuw) + \budgetp(\tvw) \geq \budgetp(\tuv) \geq 2\tuv.
\end{align*}
Here, we used that $\budgetp(\tuv) \ge 2\tuv$ by Claim \ref{clm:budgetBound}. 
\paragraph{(s,l,l)}
For a \ppp triangle $T=(u,v,w)$, assume that $uv$ is short and $uw$ and $vw$ are long, and again let $a = x_{uv}, b = x_{uw}$ and $c = x_{vw}$.
\begin{align*}
\cost(T) & = \tuw + \tvw + y_{vw} + y_{uw} - 2 \cdot y_{uvw} = 2 - 2 \cdot y_{uvw}
\\ \Delta(T) & = \budgetp(\tuw) + \budgetp(\tvw) + \budgetp(\tuv) \cdot (y_{vw} + y_{uw} - y_{uvw})  \\& \geq \budgetp(\tuw) + \budgetp(\tvw).
\end{align*}
Because $\tuw > \tau$, we have $\budgetp(\tuw) \geq 2\tuw$, by Claim \ref{clm:budgetBound}.  Hence,
\begin{align*}
\Delta(T) - \cost(T) & \geq 2 \cdot (\tuw + \tvw)-2+2\cdot y_{uvw} 
    \\ & = 2 \cdot (2 - y_{uw} - y_{vw} + y_{uvw}) - 2 
    \\ & \geq 0.
\end{align*}
In the last inequality, we used that $y_{uw} + y_{vw} - y_{uvw} \leq 1$

\paragraph{(l,l,l)}
In this case, we use correlated rounding.
For a \ppp triangle $T = (u,v,w)$, let $a = x_{uv}, b = x_{uw}$ and $c = x_{vw}$. We have
\begin{align*}
  \cost(T) &= 2(y_{uv} + y_{uw} + y_{vw}) - 6y_{uvw},\\
  \Delta(T) &= \budgetp(a)(y_{uw} + y_{vw} - y_{uvw}) + \budgetp(b)(y_{uv} + y_{vw} - y_{uvw}) + \budgetp(c)(y_{uv} + y_{uw} - y_{uvw}).
  \end{align*}
Let us fix $s = y_{uv} + y_{uw} + y_{vw}$ (which implies $a+b+c = 3-s$), and let us also fix $y_{uvw}$.  Then we have
\begin{align*}
  \Delta(T) = & \budgetp(a)(s-1 +a - y_{uvw}) + \budgetp(b)(s-1 +b - y_{uvw}) + \budgetp(c)(s-1 + c - y_{uvw}).
\end{align*}
Since $\Delta(T)$ is convex in terms of $a,b$ and $c$, then subject to the constraint $a+b+c = 3-s$, the function $\Delta(T)$ is minimized when
$a=b=c = \frac{3-s}{3}$. Let us set $p = 1-\frac{s}{3}$.  We have
\begin{align*}
  \Delta(T) & \geq  3 \cdot \budgetp(p)(s-1 +p - y_{uvw}) = 3 \cdot \budgetp(p)(\frac{2s}{3} - y_{uvw}) = 2s \cdot \budgetp(p) - 3 y_{uvw} \cdot \budgetp(p).
\end{align*}
Now we want to show that $\Delta(T) - \cost(T) \geq 0$.
\begin{align*}
\Delta(T) - \cost(T) & \geq 2s \cdot \budgetp(p) - 3y_{uvw}\cdot \budgetp(p) - 2s + 6y_{uvw}.
\end{align*}
So we want to show the following inequality.
\begin{align}
2s \cdot \budgetp(p) + 6y_{uvw} & \geq 2s + 3y_{uvw}\cdot \budgetp(p).\label{lllppp-ineq-main}
\end{align}
Since all edges are long, we have $x_{uv} + x_{uw} + x_{vw} \geq 3 \tau$, which implies that $s \leq 3-3\tau = 1.8$, since $\tau = 0.4$.  Hence, $s \in [0,1.8]$ and $p \in [0.4,1]$, so $\budgetp(p)$ is in the range $[\budgetp(\tau),\budgetp(1)]$.
We will consider three cases, according to the value of $\budgetp(p)$.  The first case is when $\budgetp(p) \in [1,2]$, in which case Inequality \eqref{lllppp-ineq-main} clearly holds.

The second case is when $\budgetp(p) \in (2, \budgetp(1)]$.  For ease of notation, let $\bar{b} = \budgetp(p)$ and assume $\bar{b} > 2$.  Rearranging, we derive the next inequalities, which are all equivalent.
\begin{align*}
  2s \cdot \bar{b} + 6y_{uvw} & \geq 2s + 3y_{uvw}\cdot \bar{b}\\
  2s + 2s (\bar{b}-2) & \geq 3y_{uvw}\cdot \bar{b} - 6y_{uvw} \\
  2s + 2s (\bar{b}-2) & \geq 3 y_{uvw} (\bar{b}-2).
  \end{align*}
Since $s = y_{uv} + y_{uw} + y_{vw} \geq 3 y_{uvw}$, the inequality holds.

The third and last case is when $\budgetp(p) \in [\budgetp(\tau),1)$.  Again, for simplicity, let $\bar{b} = \budgetp(p)$.  Then our goal is to prove the following for $\bar{b} \in [\budgetp(\tau),1)$.
\begin{align*}
  2s \cdot \bar{b} + 6y_{uvw} & \geq 2s + 3y_{uvw}\cdot \bar{b}.
\end{align*}
Rearranging, we have
\begin{align*}
  3y_{uvw} +  3y_{uvw}(1-\bar{b}) & \geq 2s(1-\bar{b}).
\end{align*}
Again, rearranging, we have
\begin{align}
  2s \leq 3y_{uvw} + \frac{3 y_{uvw}}{1 - \bar{b}}. \label{lllpppSandY}
\end{align}
By Lemma \ref{lem:weakerLemCLP}, we have $2s \leq 3 + 3y_{uvw}$.  Thus, if we can show $$\frac{y_{uvw}}{1-\bar{b}} \geq 1 \quad \iff \quad y_{uvw} \geq 1 - \bar{b}$$ then we have proved \eqref{lllpppSandY}.  To prove this, we will show
\begin{align*}
1- \bar{b} \leq \frac{2s}{3} -1 \leq y_{uvw}.
  \end{align*}
The second inequality follows from Lemma \ref{lem:weakerLemCLP}.  Recalling that $p = 1- \frac{s}{3}$, we just need to show
\begin{align*}
  1- \bar{b} = 1 - \budgetp(1 - \frac{s}{3}) \leq \frac{2s}{3} -1,
  \end{align*}
which is established next.

We want to show
  \begin{align*}
1 - \budgetp_{\woSArat}(1 - \frac{s}{3}) \leq \frac{2s}{3} -1.
  \end{align*}
Rearranging, we have
    \begin{align*}
2 - \frac{2s}{3} \leq \budgetp_{\woSArat}(1 - \frac{s}{3}),
    \end{align*}
which follows from Claim \ref{clm:budgetBound}, since $1- \frac{s}{3} \in [\tau, 1]$.
\end{proof}

\subsection{\mmm Triangles}
\begin{lemma}\label{lemma:---trianglescost}
Fix the budget function
$\budgetn \equiv \budgetn_{\alpha}$, where $\alpha \geq 1$. For any \mmm triangle $T$, we have $\cost(T) \leq \Delta(T)$.
\end{lemma}
\begin{proof}
  We have that  $\budgetn_\alpha(x) = \frac{\alpha}{1-\alpha/2} \cdot \frac{(1+2x)(1-x)}{2(1+x)} \geq (1-x)$. As before, in the proof of Lemma \ref{lemma:degtrianglescost}, we have
  $\frac{\alpha(1+2x)}{(2-\alpha)(1+x)} \geq \frac{\alpha}{(2-\alpha)} \geq 1$, for any $\alpha \geq 1$.
The inequality holds true if the coefficient in the budget is $1$.
\begin{flalign*}  &&  \cost(T) &= y_{uv}y_{uw} + y_{uv}y_{vw} + y_{uw}y_{vw}.&&\\ 
&&\Delta(T) &\geq (y_{uv}+y_{uw} - y_{uv}y_{uw})y_{vw} + (y_{uv} + y_{vw} - y_{uv}y_{vw})y_{uw} + (y_{uw} + y_{vw} - y_{uw}y_{vw})y_{uv} &&\\ 
&&         &= 2(y_{uv}y_{uw} + y_{uv}y_{vw} + y_{uw}y_{vw}) - 3y_{uv}y_{uw}y_{vw} \geq y_{uv}y_{uw} + y_{uv}y_{vw} + y_{uw}y_{vw}. &&\qedhere
     \end{flalign*}
\end{proof}

\subsection{\pmm Triangles}
\begin{lemma}
Fix budget functions $\budgetp \equiv \budgetp_{\alpha}$ and $\budgetn \equiv \budgetn_{\alpha}$, where $\alpha \geq \frac{4}{3}$. For any \pmm triangle
$T$, we have $\cost(T) \leq \Delta(T)$.
\label{lemma:+--trianglescost}
\end{lemma}
\begin{proof}
We will prove the following two cases.
\paragraph{short $+$edge}
For a \pmm triangle $T=(u,v,w)$, assume that $uv$ is a short $+$edge and $uw$ and $vw$ are $-$edges. The following inequality holds when $\alpha \geq \frac{4}{3}$.
\begin{align*}
     \budgetn_\alpha(x) =\frac{\alpha}{1-\alpha/2} \cdot \frac{(1+2x)(1-x)}{2(1+x)} \geq 2(1-x).
\end{align*}
Therefore, we have for $\cost(T)$ and $\Delta(T)$,
\begin{align*}
    \cost(T) & = y_{uw} + y_{vw} + y_{uw} + y_{vw} - 2 \cdot y_{uw}y_{vw} \leq 2 (y_{uw} + y_{vw}), 
    \\ \Delta(T) & \geq \budgetn(\tuw) + \budgetn(\tvw) \geq 2 (y_{uw} + y_{vw}).
\end{align*}

\paragraph{long $+$edge}
For a \pmm triangle $T=(u,v,w)$, assume that $uv$ is a long $+$edge and $uw$ and $vw$ are $-$edges. For $x \geq \tau$, as observed in the proof of Lemma \ref{lemma:degtrianglescost}, we can lower bound the budget of the $+$edge as follows.
\begin{align*}
\budgetp_\alpha(x) = \frac{\alpha}{1-\alpha/2} \cdot \frac{x^2}{1+x} \geq \frac{2\alpha \tau}{(2-\alpha)(1+\tau)} x \geq x.
\end{align*}
As in the previous case, $\budgetn(x) \geq 2(1-x)$.  The inequality $\cost(T) \leq \Delta(T) $ is true if the coefficient for the $+$edge is $1$ and for the $-$edge is $2$.
\begin{flalign*}
&&\cost(T) &= y_{uw}y_{uv} + y_{vw}y_{uv} + y_{uw} + y_{vw} - 2y_{uw}y_{vw} = (2 - \tuv)(y_{uw} + y_{vw}) - 2y_{uw}y_{vw}. &&\\
&&\Delta(T) &\geq (y_{uw} + y_{vw} - y_{uw}y_{vw})\tuv + 2(y_{uw} + y_{uv} - y_{uw}y_{uv}) y_{vw} + 2(y_{vw} + y_{uv} - y_{vw}y_{uv}) y_{uw}&&\\
 && &=(y_{uw} + y_{vw} - y_{uw}y_{vw})\tuv  + 2 (1 - \tuv+ y_{uw}\tuv )y_{vw} + 2(1 - \tuv + y_{vw}\tuv)y_{uw}&&\\
&& &= (2-\tuv)(y_{uw} + y_{vw}) + 3 y_{uw}y_{vw}\tuv. && \qedhere
\end{flalign*}
\end{proof}

\subsection{\ppm Triangles}

\begin{lemma}
Fix budget functions $\budgetp \equiv \budgetp_{\alpha}$ and $\budgetn \equiv \budgetn_{\alpha}$, where $\alpha = \woSArat$. For any \ppm~triangle $T$, we have $\cost(T) \leq \Delta(T)$.
\label{lemma:++-trianglescost}
\end{lemma}
\begin{proof}
    Again, we distinguish cases depending on the number of short $+$edges.
\paragraph{(s,s)}

For a \ppm triangle $T=(u,v,w)$, assume that $uv$ is a $-$edge and $uw$, $vw$ are short $+$edges.  As usual, let $x_{uv} = a, x_{uw} = b$ and $x_{vw} = c$.
\begin{align}
    \cost(T) & = 2\cdot \tuv + 1,\nonumber \\
    \Delta(T) & = \budgetn(\tuv) + \budgetp(\tuw) + \budgetp(\tvw) \geq \budgetn(\tuv) + 2 \cdot \budgetp\left(\frac{\tuw + \tvw}{2} \right) \nonumber \\ & \geq \budgetn(\tuv)+ 2\cdot \budgetp\left ({\frac{\tuv}{2}} \right) ~ = \frac{2\alpha}{2-\alpha}\left(\frac{(1+2a)(1-a)}{2(1+a)} + \frac{a^2}{2+a}\right).
    \label{equ:delta1.5}
\end{align}
As before, we used the convexity of $\budgetp$ (Lemma \ref{lem:b_conv}) together with the triangle inequality, $\tuv \leq \tuw + \tvw$. The derivative of \eqref{equ:delta1.5} w.r.t. $\tuv$ is less than $0$ for $\tuv \geq 0$,
\[
- \frac{2\alpha}{2-\alpha} \left( \frac{a(3a+4)}{(1+a)^2(2+a)^2} \right) \leq 0.
\]
Hence, \eqref{equ:delta1.5} is decreasing with $\tuv$. Observe that $\tuv \leq \tuw + \tvw \leq 2 \tau$. Since the cost is increasing with $\tuv$ we only need to check the case $\tuv = 2 \tau$,
\begin{align*}
    \Delta(T) - \cost(T) & \geq \budgetn\left(2 \tau\right)+ 2\cdot \budgetp\left (\tau \right) - 4 \tau -1\\
    & \geq \budgetn\left(2 \tau\right)+ 4 \tau - 4 \tau -1 \geq .02.
    \end{align*}
Notice that we use Claim \ref{clm:budgetBound} to lower bound $\budgetp(\tau)$.

\paragraph{(s,l)}
For a \ppm triangle $T=(u,v,w)$, assume that $uv$ is a $-$edge, $uw$ is a short $+$edge and $vw$ is a long $+$edge.
\begin{align*}
    \cost(T) & = \tuv + y_{vw} + y_{vw} + y_{uv} - 2 y_{vw} y_{uv} = 1
    + 2 y_{vw} (1-y_{uv}) = 1 + 2 (1-x_{vw}) \tuv = 1 + 2 (1- \tvw) \tuv.
    \\ \Delta(T) & = \budgetp(\tvw) + \budgetn(\tuv) + \budgetp(\tuw) (y_{vw} + y_{uv} - y_{vw} y_{uv}) 
    \\ & = \budgetp(\tvw) + \budgetn(\tuv) + \budgetp(\tuw) (2 -\tvw -
    \tuv - (1-\tvw)(1-\tuv))
    \\ & = \budgetp(\tvw) + \budgetn(\tuv) + \budgetp(\tuw) (1 - \tuv \tvw).
\end{align*}
\paragraph{Case 1} $\tvw \geq \tuv$. 
\begin{align*}
    \Delta(T) & \geq \budgetp(\tuv) + \budgetn(\tuv) = \frac{\alpha}{2-\alpha} \geq  3,
\end{align*}
while $\cost(T)$ is always less than $3$.

\paragraph{Case 2} $\tvw \leq \tuv$. We have $\tuw \geq \tuv - \tvw$
by triangle inequality.  Since $\tuw \leq \tau$, we also have $\tuv -
\tvw \leq \tau$.
 Since $\budgetp(x)$ is increasing by Lemma \ref{lem:b_conv},
\[
    \Delta(T) \geq \budgetp(\tvw) + \budgetn(\tuv) + \budgetp(\tuv - \tvw) (1-\tuv \tvw).
\]
Define $\Co = \frac{\alpha}{2-\alpha}$. For a fixed $x \coloneqq \tvw \geq \tau$, define $z \coloneqq \tuv$. For $-$edges we can use a coefficient of $\Co$, since $\budgetn(x) \geq
\Co(1-x)$.
We need to argue that the difference on line
\eqref{equ:1.56diff6} is non-negative.
\begin{align}
    \Delta(T) - \cost(T) 
& \geq \budgetp(x) + \budgetn(z) + \budgetp(z - x)
    (1-xz) - 1 - 2 (1-x) z \label{equ:1.56diff6}
\\ & \geq \budgetp(x) + \Co (1-z) + \budgetp(z - x) (1-xz) - 1 - 2 (1-x) z
\\& = 2\Co \frac{x^2}{1+x} + \Co (1-z) + 2\Co \frac{(z - x)^2}{1+z - x} (1-xz) - 1 - 2 (1-x) z.
  \label{equ:1.56diff7}
\end{align}
Define $$f(z) := \Co (1-z) + 2\Co \frac{(z - x)^2}{1+z - x} (1-xz) - 1 - 2 (1-x) z.$$

\begin{claim}\label{clm:1.56fdecreasing}
$f$ is decreasing. 
\end{claim}  
\begin{cproof}
We will show that the derivative of $f$
is negative. We have
\begin{align*}
    f'(z) & = \Co \left(2x^2-4xz+2x+1+\frac{2(x-1)x-2}{(1+z-x)^2}\right) +2(x-1) \\ & \leq \Co \left(-2xz+2x+1+\frac{2(x-1)x-2}{(1+z-x)^2}\right). 
\end{align*}
We will show that $g(z) \coloneqq \left(-2xz+2x+1+\frac{2(x-1)x-2}{(1+z-x)^2}\right) < 0$ and thus, $f'(z) < 0$.
We have,
\[
    g'(z) = - \frac{4(x^2-x-1)}{(1+z-x)^3} - 2x,
\]
and,
\[
    g''(z) = \frac{12(x^2-x-1)}{(1+z-x)^4}.
\]
Note that $g''(z) \leq 0$ since $C > 0$, the numerator is negative for $x \in [0,1]$ and the denominator is always positive.
Thus, $g'(z)$ is decreasing and minimized for $z=1$,
\[
    g'(1) = -\frac{4(x^2-x-1)}{(2-x)^3} - 2x \ge 0,
\]
for all $x \in [0,1]$. Hence, $g(z)$ is increasing and maximized for $z=\min\{1, x+\tau\}$. We have,
\[
    g(1) = \left(-2x+2x+1+\frac{2(x-1)x-2}{(2-x)^2}\right) < 0
\]
for $x \ge 1 - \tau$ and,
\[
    g(x+\tau) = \left(-2x(x+\tau)+2x+1+\frac{2(x-1)x-2}{(1+\tau)^2}\right) < 0,
\]
for $x\in [0,1]$.
\end{cproof}


Now we return to our goal, which is to show that 
\begin{align}
    2\Co \frac{x^2}{1+x} + \Co (1-z) + 2\Co \frac{(z - x)^2}{1+z - x} (1-xz) - 1 - 2 (1-x) z \geq 0.
    \label{eqn:1.56diff8}
\end{align}
By Claim \ref{clm:1.56fdecreasing}, we can assume that $z= \min\{1, \tau +x\}$.  First,
we consider the case in which $\tau + x \geq 1$, in which case we assume that $z= 1$.  
Thus, we have
\begin{align}
(\ref{eqn:1.56diff8}) \geq 2 \Co \left (x^2 + \frac{1}{x-2} + \frac{1}{x+1}\right ) + 2x-3 \geq 0
\label{eqn:1.56casez1}
\end{align}
for $x\geq 1-\tau$.




Next, we consider the case in which $z = x + \tau < 1$,
\begin{align}
(\ref{eqn:1.56diff8}) \geq 2\Co \frac{x^2}{1+x} + \Co (1-x-\tau) + 2\Co \frac{\tau^2}{1+\tau} (1-x(x+\tau)) - 1 - 2 (1-x) (x+\tau) \geq 0,
\end{align}
for $x\in[0,1]$.


\paragraph{(l,l)}

For a \ppm triangle $T=(u,v,w)$, assume that $uv$ is a $-$edge and
$uw$, $vw$ are long $+$edges.  Define $\Co = \frac{\alpha}{2-\alpha}$.
For $-$edges we can use a coefficient of $\Co$, since $\budgetn(x) \geq
\Co(1-x)$. For $+$edges define $f(x) = \frac{2\Co \cdot
x^2}{1+x}$ so that $\budgetp(x)= f(x)$.   Recall that $a=x_{uv},
b=x_{uw}$ and $c = x_{vw}$.  Without loss of generality, we assume
that $b \leq c$, which implies $y_{uw} \geq y_{vw}$.

\begin{align*}
    \Delta(T) & = \Co(y_{uw} + y_{vw} - y_{uvw})y_{uv} + f(\tvw)(y_{uw} + y_{uv} - y_{uw}y_{uv})  
    + f(\tuw)(y_{vw} + y_{uv} - y_{vw}y_{uv})\\
		& \geq 
    \Co(y_{uw} + y_{vw} - y_{uvw})y_{uv} + 2\tvw (y_{uw} + y_{uv} - y_{uw}y_{uv})  
    + 2\tuw (y_{vw} + y_{uv} - y_{vw}y_{uv}).\\
 \cost(T) & = y_{uvw} + y_{uw} + y_{uv} - 2y_{uw}y_{uv} + y_{vw} +
    y_{uv} - 2y_{vw}y_{uv}\\ 
& = y_{uvw} + y_{uw} + y_{vw} + 2y_{uv} - 2(y_{uw} + y_{vw})y_{uv}.
  \end{align*}
The inequality follows from Claim \ref{clm:budgetBound}.
Using $y_{uv} = 1-a, y_{uw} = 1-b, y_{vw} = 1-c$ and $y=y_{uvw}$,
we have
\begin{align*}
    \Delta(T) & \geq 
    \Co(y_{uw} + y_{vw} - y_{uvw})y_{uv} + 2\tvw (y_{uw} + y_{uv} - y_{uw}y_{uv})  
    + 2\tuw (y_{vw} + y_{uv} - y_{vw}y_{uv})
    \\ & =  \Co(2-b -c - y)(1-a) + 2c + 2b - 4abc.
\end{align*}
For the difference we get,
\begin{align*}
    \Delta(T) -\cost(T) 
		\geq 
    \Co(2-b -c - y)(1-a) + 3c + 3b +2a -4abc -y - 4 + 2(2-b
		-c)(1-a)\\
		=
2\Co + (1 - \Co) b + (1 - \Co) c + (-\Co-1 + \Co a) y +( -2 \Co  -2)a + (\Co+2) ab + (\Co+2) ac    
-4abc.
  \end{align*}
Define the previous line to be the function $g(a,b,c,y)$.  Our goal is to that that
$$\Delta(T) - \cost(T) \geq g(a,b,c,y) \geq 0.$$
To minimize $g$, we can set $y$ as large as possible, since $y$ has a
negative coefficient in $g$.  So we have
$$y = \min\{y_{uv}, y_{vw}\} = \min\{1-a,1-c\},$$
because we assumed $b \leq c$ and hence $y_{vw} \leq y_{uw}$.

\paragraph{Case 1: $a\leq b \leq c$.}  First, we consider the case in which
$a \leq b \leq c$.  In this case, we can set
$y=1-c$.  So we have
\begin{align*}
g_1(a,b,c) & := 2\Co + (1 - \Co)(b+c) + (-\Co-1 + \Co a)(1-c) +( -2 \Co  -2)a + (\Co+2) ab + (\Co+2) ac    
-4abc \\
& = (1 - \Co) b + (\Co-1) + (2 - \Co a)c
+( -\Co  -2)a + (\Co+2) ab + (\Co+2) ac    
-4abc\\
& = (1 - \Co) b + (\Co-1) + 2c 
+( -\Co  -2)a + (\Co+2) ab + 2 ac    
-4abc.
  \end{align*}

\begin{claim}\label{clm:increaseWithC}
$g_1$ is increasing in $c$.
\end{claim}

\begin{cproof}
The derivative of $g_1$ with respect to $c$ is
$$2 + 2a - 4ab \geq 2+2a-4a = 2(1-a) \geq 0.$$
Here, we used that $b \leq 1$.
\end{cproof}

So we can make $c$ as small as possible and we thus set $c=b$.

\begin{align*}
g_1(a,b,c) \geq g_1(a,b,b) = (3 - \Co) b + (\Co-1) 
+( -\Co  -2)a + (\Co+4) ab 
-4ab^2.
  \end{align*}

\begin{claim}
$g_1(a,b,b)$ is decreasing with $a$.
\end{claim}
\begin{cproof}
The derivative is 
$$ - \Co - 2 + \Co b + 4b -4b^2.$$
The above quadratic is maximized for $b = \frac{\Co +4}{8}$,
$$- \Co - 2 + \Co b + 4b -4b^2 \leq \frac{1}{16}(\Co^2 - 8\Co -16) \leq 0$$
\end{cproof}
Now we can increase $a$ as much as possible and set $a = b$,
\begin{align}
g_1(b,b,b) =  (\Co-1) + (1 - 2\Co) b 
 + (\Co+4) b^2 -4b^3.\label{case1:final}
  \end{align}

\begin{claim}
    $g_1(b,b,b)$ is decreasing in $b$.
\end{claim}
\begin{cproof}
    The derivative is, $$1-2\Co+(2\Co+8)b - 12b^2.$$
    The above quadratic is maximized for $b=\frac{C+4}{12}$,
    $$1-2\Co+(2\Co+8)b - 12b^2 \leq \frac{1}{12}(\Co^2-16\Co+28) \leq 0$$.
\end{cproof}
Thus, we have, $$g_1(b,b,b) \geq g_1(1,1,1) = 0$$

\paragraph{Case 2: $b \leq a \leq c$.}
Our goal is to show that in this case, 
$g_1(a,b,c) \geq 0$.  Claim \ref{clm:increaseWithC} still applies.
When we decrease $c$ as much as possible, we have $c=a$.  So we have

\begin{align*}
g_1(a,b,c) \geq g_1(a,b,a) = g_2(a,b) 
& := (1 - \Co) b + (\Co-1) -\Co a + (\Co+2) ab + 2 a^2    
-4a^2b.
\end{align*}

\begin{claim}
$g_2$ is decreasing in $b$.
\end{claim}

\begin{cproof}
The derivative is $$1 - \Co + (\Co + 2)a -4a^2.$$
The above quadratic is maximized for $a = \frac{\Co+2}{8}$,
$$1 - \Co + (\Co + 2)a -4a^2 \leq \frac{1}{16} (\Co^2 - 12 \Co + 20) \leq 0.$$
\end{cproof}

So we increase $b$ as much as possible, setting $b=a$.  Then we have
\begin{align*}
g_2(a,b) \geq g_2(a,a) = g_3(a) 
& := a + (\Co-1) -2\Co a + (\Co+4) a^2 
-4a^3.
\end{align*}

\begin{claim}\label{clm:case2Last}
$g_3$ is decreasing in $a$.
\end{claim}

\begin{cproof}
The derivative is
\begin{align*}
1 - 2 \Co + 2 (4 + \Co) a  - 12 a^2.
\end{align*}
The above quadratic is maximized for $a = \frac{\Co + 4}{12}$, $$1 - 2 \Co + 2 (4 + \Co) a  - 12 a^2 \leq \frac{1}{12} (\Co^2 -16 \Co + 28) \leq 0.$$
\end{cproof}
By Claim \ref{clm:case2Last}, we have $g_3(a) \geq g_3(1)$, and we
have that $g_3(1) = 0$.

\paragraph{Case 3: $b \leq c \leq a$.}


Recall that our goal is to prove the following quantity is $\geq 0$.
\begin{align*}
2\Co + (1 - \Co) b + (1 - \Co) c + (-\Co-1 + \Co a) y +( -2 \Co  -2)a + (\Co+2) ab + (\Co+2) ac    
-4abc.
  \end{align*}
In this case, we can set $y = 1-a$, so we have
\begin{align}
\Co -1 + (1 - \Co) b + (1 - \Co) c -a - \Co a^2 + (\Co+2) ab + (\Co+2) ac -4abc.\label{eqn:case3-1}
\end{align}

\begin{claim}
\eqref{eqn:case3-1} is decreasing with $a$.
\end{claim}

\begin{cproof}
The derivative w.r.t. $a$ is
\begin{align}
   -1 - 2\Co a + (\Co +2)(b+c) - 4bc. 
   \label{eqn:case3deriv}
\end{align}
The coefficient of $a$ in \ref{eqn:case3deriv} is $-2\Co \leq 0$. Similar, the coefficient of $b$ in \ref{eqn:case3deriv} is $C+2-4c \geq C-2 \geq 0$. Thus, we can assume that $a = b =c$,
\begin{align*}
    (\ref{eqn:case3deriv}) \leq  -1 + 4 c - 4c^2 = - (2c - 1)^2 \leq 0. 
\end{align*}
\end{cproof}

Thus, to minimize \eqref{eqn:case3-1}, 
we can make $a$ as large as possible, setting $a =\min\{b+c, 1\}$.
Let us consider the first subcase in which $b+c > 1$ and so $a=1$.  In
this case, \eqref{eqn:case3-1} becomes
\begin{align}
-2 + 3(b+c) -4bc = -(2b-1)(2c-1) + b + c -1 \geq -(2b-1)(2c-1). 
\label{eqn:case3-2}
\end{align}
Note that $c \geq \frac{1}{2}$ since $b+c \geq 1$ and $c \leq b$. Thus \ref{eqn:case3-2} is decreasing with $b$,
$$-(2b-1)(2c-1) \geq -(2(1-c)-1)(2c-1) = (1-2c)^2 \geq 0$$

The second subcase is when $b+c \leq 1$ and we set $a=b+c$ or equivalently, $c = a - b$.
Then \eqref{eqn:case3-1} becomes
\begin{align}
\Co - 1 + (4b^2 - \Co)a + (2-4b)a^2.
\label{eqn:case3-3}
\end{align}
\begin{claim}
    \eqref{eqn:case3-3} is decreasing with $a$.
\end{claim}
\begin{cproof}
    The derivative of $\eqref{eqn:case3-3}$ w.r.t. $a$ is, $$4b^2 - \Co + 2(2-4b)a.$$ 
    Note that $b \leq \frac{1}{2}$ since $b+c \leq 1$ and $b \leq c$. Thus the derivative is increasing with $a$ and maximized for $a=1$, $$4b^2 - \Co + 2(2-4b)a \leq 4b^2 - 8b - \Co + 4.$$ 
     which is decreasing for $b \leq 1$ and thus maximized for $b=\tau$,
     $$4\tau^2 - 8 \tau - \Co + 4 \leq 0$$
\end{cproof}
By the above claim, $\eqref{eqn:case3-3}$ is minimized for $a = 1$,
$$\Co - 1 + (4b^2 - \Co)a + (2-4b)a^2 \geq 4b^2 - 4b + 1 = (2b-1)^2 \geq 0.$$


\end{proof}


\section{1.485-Approximation: Computer-assisted proof}
\label{sec:1485approximate}

In this section, by introducing and solving the {\em factor-revealing SDP}, we will prove Lemma~
\ref{lemma:budgets-for-pureclusterlp-sdpbudget}, which is restated as follows,
\lemmabudgetsforpureclusterlpSDPmethod*

Let $\mathcal{T}$ be defined as the set of all possible triangles, represented by tuples of values $(y_{uv}, y_{uw}, y_{vw}, y_{uvw})$ satisfying the condition $y_{uv} \leq y_{uw} \leq y_{vw}$. We consider positive degenerate triangle $uv$ as a \ppp triangle and negative degenerate triangle $uv$ as a \mmp triangle, both having the same $y$ value $(y_{uv}, y_{uv}, 1, y_{uv})$. This substitution simplifies our computations, and we will later show that it does not increase the value of $\Delta(T) - \cost(T)$ for degenerate triangles. Consequently, $\mathcal{T}$ includes both degenerate and non-degenerate triangles and it can be used to lower bound $\Delta(T) - \cost(T)$. Given graph $G$ and the \ref{LP:clusterlp} solution $z_S$, along with their corresponding $y$ values, we define $\eta_T(y)$ for any triangle $T \in \mathcal{T}$ as the number of such triangles present in graph $G$.

In Section \ref{sec:pureClusterAnalyticalProof}, we showed that, for any triangle $T \in \mathcal{T}$, the inequality $\Delta(T) - \cost(T) \geq 0$ holds. In this section, we show, for any solution $y$, it is possible to determine a reduced budget such that the summation $\sum_{T} \eta_T(y) \cdot (\Delta(T) - \cost(T)) \geq 0$. The key insight is that $\eta_T(y)$ for a given solution $y$ cannot take arbitrary values. In this section, we explain how stronger constraints can be imposed on $\eta_T(y)$, enabling the use of a smaller budget for $\sum_{T} \eta_T(y) \cdot (\Delta(T) - \cost(T))$.

Here, we first outline the constraint we intend to impose on $\eta_T$. Subsequently, we will employ discretization techniques to encompass all possible triangles $T$. We will then formulate a semidefinite program to demonstrate that it is possible to achieve $\sum_{T} \eta_T(y) \cdot (\Delta(T) - \cost(T)) \geq 0$ within a reduced budget. 



\paragraph{Covariance Constraint.}
Given~\ref{LP:clusterlp} solution $z_S$ and $y$, for each node $u$, the covariance matrix $\cov_u \in \R^{V \times V}$ where
\begin{align*}
    \cov_u(v, w) = y_{uvw} - y_{uv}y_{uw}
\end{align*}
is positive semidefinite (PSD). Actually, we can even show a stronger version of the covariance constraint.

\begin{lemma}
    \label{lemma:PSD}
    We define $y_S = \sum_{S' \supseteq S} z_{S'}$ for every $S \subseteq V$.  For any $T \subseteq V, T \neq \emptyset$, the matrix $M \in \R^{V \times V}$ where $M_{uv} = y_{T \cup \{u, v\}} - y_{T \cup \{u\}} y_{T \cup \{v\}}$ is PSD. 
\end{lemma}
\begin{proof}
    Fix any vector $a \in \R^{V}$. For every $S \subseteq V$, we define $a(S):= \sum_{v \in S} a_v$. Then 
    \begin{align*}
        a^\intercal M a &= \sum_{u, v \in V} (y_{T \cup \{u, v\}} - y_{T \cup \{u\}} y_{T \cup \{v\}})a_u a_v \\
        &= \sum_{u, v \in V} a_ua_v \left(\sum_{S \supseteq T \cup \{u, v\}} z_S - \Big(\sum_{S \supseteq T \cup \{u\}} z_S\Big) \Big(\sum_{S' \supseteq T \cup \{v\}} z_{S'}\Big) \right)\\
        &=\sum_{S \supseteq T} z_S \cdot a(S) \cdot a(S) - \sum_{S \supseteq T}\sum_{S' \supseteq T} z_S\cdot z_{S'}\cdot a(S) \cdot a(S')\\
        &\geq \sum_{S \supseteq T, S' \supseteq T} z_S z_{S'} \cdot a^2(S)- \sum_{S \supseteq T}\sum_{S' \supseteq T} z_S\cdot z_{S'}\cdot a(S) \cdot a(S')\\
        &= \sum_{S \supseteq T, S' \supseteq T} z_Sz_{S'}\left(\frac12 a^2(S) + \frac12 a^2(S') - a(S)a(S')\right) \geq 0.
    \end{align*}
    The first inequality used that $\sum_{S' \supseteq T} z_{S'} \leq 1$ as $T \neq \emptyset$.
    Therefore, $M$ is PSD. 
\end{proof}

Lemma~\ref{lemma:PSD} places a significant constraint on the distribution of triangles. We want to mention that the covariance matrix is implied by \ref{LP:clusterlp}, so we do not need to add extra constraints to the linear program. To obtain Lemma~\ref{lemma:PSD}, we only require that the solution is from~\ref{LP:clusterlp}.

To leverage the property of the covariance matrix being positive semidefinite, we begin by partitioning the interval $[0, 1]$ into numerous subintervals denoted as $I_1, I_2, ..., I_t$ where each subinterval $I_j \in [0, 1]$ is continuous. Given any $y_{uv} \in [0, 1]$, we use $I(y_{uv})$ to represent the interval containing $y_{uv}$. Let $l(I_j)$ and $r(I_j)$ be the left and right boundary of interval $I_j$. Given a node $u$, consider the matrix $Q_u \in \R^{t\times t}$, where $Q_u(I_j, I_k)$ is defined as
\begin{equation*}
    Q_u(I_j, I_k) = \sum_{\substack{
    y_{uv} \in I_j, y_{uw} \in I_k \\ v \in V,  w \in V}} y_{uvw} - y_{uv}y_{uw}.
\end{equation*}
When $v = u$ or $w = u$, we have $ y_{uvw} - y_{uv}y_{uw} = 0$ and $Q_u$ remains well-defined. We first show that $Q_u$ is also PSD.
\begin{lemma}
\label{lem:Qusdp}
    Let $Q_u$ be the matrix defined above, then $Q_u \succeq 0$.
\end{lemma}
\begin{proof}
    For any given vector $a \in \R^{t}$, let $a(I_j)$ be the value for index $I_j$. 
\begin{align*}
    a^\intercal \cdot Q_u \cdot a &= \sum_{I_j, I_k} a(I_j)Q_u(I_j, I_k)a(I_k) \\
    &= \sum_{I_j, I_k} \sum_{y_{uv} \in I_j, y_{uw} \in I_k} a(I_j)a(I_k)(y_{uvw} - y_{uv}y_{uw}) \\
    &=\sum_{v, w}\sum_{ I_j \ni y_{uv}, I_k \ni y_{uw} } a(I_j)a(I_k)(y_{uvw} - y_{uv}y_{uw}) \\
    &= \sum_{v, w}\sum_{y_{uv}, y_{uw} } a(I(y_{uv}))a(I({y_{uw}}))(y_{uvw} - y_{uv}y_{uw}) \\
    &= b^T \cdot \cov_u \cdot b \geq 0
\end{align*}
where $b \in \R^{n}$ and each entry of $b$ is defined as $b(v) = a(I(y_{uv}))$. The last inequality is using the fact that the covariance matrix is PSD.
\end{proof}
Given a \ref{LP:clusterlp} solution $y$, let $Q(y) = \sum_u Q_u$ be the sum of $Q_u$ for all vertices $u$ in the set $V$. According to Lemma \ref{lem:Qusdp}, we can assert that $Q(y) \succeq 0$.

We can now express the matrix $Q(y)$ using triangles $T \in \mathcal{T}$ and the corresponding counts $\eta_T(y)$. To achieve this, we need to consider the increase caused by $T$ in the matrices $Q_u$, $Q_v$, and $Q_w$.

Specifically, a non-degenerate triangle $T$ increases the entries $Q_u(I(y_{uv}), I(y_{uw}))$ and $Q_u(I(y_{uw}), I(y_{uv}))$ by $y_{uvw} - y_{uv}y_{uw}$ for the matrix $Q_u$,. Similar computations are performed for $Q_v$ and $Q_w$. For degenerate triangle $T = (y_{uv}, y_{uv}, 1, y_{uv})$, it increases $Q_u(I(y_{uv}), I(y_{uv}))$ by $(y_{uv} - y_{uv}^2)$, increases $Q_v(I(y_{uv}), I(y_{uv}))$ by $(y_{uv} - y_{uv}^2)$ and does not affect $Q_w$. In either case, we increase $Q(I(y_{uv}), I(y_{uw}))$ and $Q_u(I(y_{uw}), I(y_{uv}))$ by $y_{uvw} - y_{uv}y_{uw}$, $Q(I(y_{vu}), I(y_{vw}))$ and $Q(I(y_{vw}), I(y_{vu}))$ by $y_{uvw} - y_{vu}y_{vw}$, $Q(I(y_{wu}), I(y_{wv}))$ and $Q(I(y_{wv}), I(y_{wu}))$ by $y_{uvw} - y_{wu}y_{wv}$.

More formally, to compute $Q(y)$. Let $C(T) \in \R^{6 \times 6}$ be matrix that
\begin{align*}
    C(T) = diag(C^{(1)}(T), C^{(2)}(T), ... ,C^{(6)}(T))
\end{align*}
where
\begin{align*}
    C^{(1)}(T) = C^{(2)}(T) = \cov_u(v, w) = y_{uvw} - y_{uv}y_{uw},\\
    C^{(3)}(T)=C^{(4)}(T)= \cov_v(u, w) = y_{uvw} - y_{vu}y_{vw},\\
    C^{(5)}(T)=C^{(6)}(T)= \cov_w(v, u) = y_{uvw} - y_{wv}y_{wu}.\\
\end{align*}
Let $Head(T) \in \mathbb{R}^{6 \times t}$ be a binary matrix, defined as $Head(T)(i, I(a)) = 1$ when
\begin{align*}
(i, a) \in \{(1, y_{uv}), (2, y_{uw}), (3, y_{vu}), (4, y_{vw}), (5, y_{wu}), (6, y_{wv})\},
\end{align*}
and 0 otherwise. Similarly, we define $Tail(T) \in \mathbb{R}^{6 \times t}$ as a binary matrix, defined as $Tail(T)(i, I(a)) = 1$ when
\begin{align*}
(i, a) \in \{(1, y_{uw}), (2, y_{uv}), (3, y_{vw}), (4, y_{vu}), (5, y_{wv}), (6, y_{wu})\}.
\end{align*}
We can then express $Q(y)$ as
\begin{align*}
Q(y) = \sum_{T \in \mathcal{T}} \eta_T(y) \cdot \left(Head^\intercal(T) \cdot C(T) \cdot Tail(T)\right).
\end{align*}
Let $\mathcal{T}(I_i, I_j, I_k) = \{(y_{uv}, y_{uw}, y_{vw}, y_{uvw}) \in \mathcal{T} \mid y_{uv} \leq y_{uw} \leq y_{vw}, y_{uv} \in I_i, y_{vw} \in I_j, y_{uv} \in I_k\}$ be the set of triangles within the intervals $I_i, I_j, I_k$. For any two triangles $T$ with $(y_{uv} = a, y_{uw} = b, y_{vw} = c)$ and $T'$ with $(y_{uv} = a', y_{uw} = b', y_{vw} = c')$, if $T$ and $T'$ are located within the same interval, meaning $I(a) = I(a')$, $I(b) = I(b')$, and $I(c) = I(c')$, then we have $Head(T) = Head(T')$ and $Tail(T) = Tail(T')$. Given $I_i, I_j, I_k$, define $Head(I_i, I_j, I_k) = Head(T)$ (or $Tail(I_i, I_j, I_k) = Tail(T)$), where $T$ is an arbitrary triangle such that $T \in \mathcal{T}(I_i, I_j, I_k)$. We can then express $Q(y)$ as
\begin{align*}
    Q(y) = \sum_{I_i, I_j, I_k} Head^\intercal(I_i, I_j, I_k) \cdot \left( \sum_{T \in \mathcal{T}(I_i, I_j, I_k)} \eta_T(y) \cdot C(T)  \right) \cdot Tail(I_i, I_j, I_k).
\end{align*}


Since $Q(y) \succeq 0$, it provides a constraint for $\eta_T(y)$. Our target is to show that $\sum_{T} \eta_T(y) \cdot (\Delta(T) - \cost(T)) \geq 0 $ under the constraint that $Q(y) \succeq 0$. We will consider all possible triangles in each range $\mathcal{T}(I_i, I_j, I_k)$ and there will be at most $O(t^3)$ different ranges. For any solution $y$, the triangle $T = \{(y_{uv}, y_{uw}, y_{vw}, y_{uvw}) \mid y_{uv} \leq y_{uw} \leq y_{vw}\}$ satisfies the triangle inequality, that is $1 - y_{uw} + 1 - y_{vw} \geq 1 - y_{uv}$. Therefore, given range $I_i, I_j, I_k$, we assume that $l(I_i) \geq l(I_j) + l(I_k) - 1$, $r(I_j) \leq 1 + l(I_i) - l(I_k) $ and $r(I_k) \leq 1 + l(I_i) - l(I_j) $, since there will not be any triangles in the range when the above inequality does not hold. If necessary, we can truncate intervals to ensure that these inequalities hold true. 

In each range $I_i, I_j, I_k$, we will use $\Tilde{d}(I_i, I_j, I_k) = \min_{T \in \mathcal{T}(I_i, I_j, I_k)}(\Delta(T) - \cost(T))$ to  establish a lower bound on  $\Delta(T) - \cost(T)$. $\Tilde{d}(I_i, I_j, I_k)$ can be precomputed and there will be at most $O(t^3)$ different $\Tilde{d}$ value. We first show that $\Tilde{d}(I_i, I_j, I_k)$ is a lower bound even for degenerate triangle. Recall that we treat degenerate triangle as triangle with $y$ value $(y_{uv}, y_{uv}, 1, y_{uv})$.
\begin{lemma}
    Let $T = (u, v)$ be a degenerate triangle, for any budget function $b^{+}, b^{-}$ such that $b^{+}(0) = 0$, then 
    \begin{align*}
        \cost(u, v) - \Delta(u, v) \geq \Tilde{d}(I(y_{uv}), I(y_{uv}), I(1)) 
    \end{align*}
\end{lemma}
\begin{proof}

Let $T_1$ represent the \ppp triangle with $(y_{uv}, y_{uv}, 1, y_{uv})$. If $(u, v)$ is a positive edge and $y_{uv} \geq 2/3$, then $\cost(u, v) = 0$ and $\Delta(u, v) = 2 \cdot \budgetp(x_{uv})$. In this case, $\cost(T_1) = 0$ and $\Delta(T_1) = 2\cdot \budgetp(x_{uv}) + \budgetp(0)$. If $(u, v)$ is a positive edge and $y_{uv} < 2/3$, we have $\cost(u, v) = 2(1 - y_{uv})$ and $\Delta(u, v) = 2\cdot \budgetp(x_{uv})$. For $T_1$ in this scenario, $\cost(T_1) = 2(1 - y_{uv})$ and $\Delta(T_1) = 2\cdot \budgetp(x_{uv}) + \budgetp(0)y_{uv}$. In either case, we observe that $\Delta(u, v) - \cost(u, v) = \Delta(T_1) - \cost(T_1) \geq \Tilde{d}(I(y_{uv}), I(y_{uv}), I(1))$.

Let $T_2$ be the \mmp triangle with $(y_{uv}, y_{uv}, 1, y_{uv})$. If $(u,v)$ is a negative edge, then $\cost(u, v) = 2y_{uv}$ and $\Delta(u, v)= 2\cdot\budgetn(x_{uv})$. For $T_2$, we have $\cost(T_2) = 2y_{uv} + 2y_{uv} - 2y_{uv}^2$ and $\Delta(T_2)= 2\cdot \budgetn(x_{uv}) + b^{+}(0)(2y_{uv} - y_{uv}^2)$. Therefore, we conclude that $\Delta(u, v) - \cost(u, v) \geq \Delta(T_2) - \cost(T_2) \geq \Tilde{d}(I(y_{uv}), I(y_{uv}), I(1))$.
\end{proof}


We still need to represent $Q$ using $\eta_T$. The problem is there might be infinite types of triangles in each range $\mathcal{T}(I_i, I_j, I_k)$, and we are not able to list every possible triangle. To solve this problem, our last observation is that, for $i \in [1, 6]$, $C^{(i)}(T)$ is a multilinear function in terms of $y_{uv}, y_{uw}, y_{vw}$ and $y_{uvw}$.


Since $C^{(i)}(T)$ is multilinear, we only need to consider 16 different triangles in each range to cover all possible triangles, where those triangles' $y$ value are 

\begin{align*}
    y_{uv} \in \{ l(I_i), r(I_i) \}&, y_{uw} \in \{ l(I_j), r(I_j) \}, y_{vw} \in \{ l(I_k), r(I_k) \} \\ 
    y_{uvw} \in \{ \max \Bigl( 0, l(I_i) + l(I_j) - 1 &, l(I_j) + l(I_k) - 1, l(I_i) + l(I_k) - 1 \Bigl) ,\min \Bigl( r(I_i), r(I_j), r(I_k) \Bigl) \}.
\end{align*}


The $y_{uvw}$ constraints comes from the fact that $y_{uvw} \leq y_{vw}$ and $y_{uw} + y_{uv} - y_{uvw} \leq 1$. Each triangle might be considered in different ranges multiple times and we shall treat them as different types of triangles. Let $\mathcal{T_D}(I_i, I_j, I_k)$ be the set of triangles containing this 16 triangles in $\mathcal{T}(I_i, I_j, I_k)$ and $\mathcal{T_D} = \cup_{I_i, I_j, I_k}\mathcal{T_D}(I_i, I_j, I_k)$, we can set up a semidefinite programming to lower bound $\sum_{T} \eta_T(y) \cdot (\Delta(T) - \cost(T))$. In \ref{LP:sdplp}, the variables are $\eta_T$ that represent the ratio of triangles $T \in \mathcal{T_D}$ and $N\cdot \eta_T$ is the number of triangle of $T$, where $N = \frac{n(n-1)(n-2)}{6} + \frac{n(n-1)}{2} = \frac{n(n-1)(n+1)}{6}$. $\Tilde{d}(I_i, I_j, I_k)$, $Head^\intercal(I_i, I_j, I_k)$, $C(T)$ and  $Tail(I_i, I_j, I_k)$ are constant and can be precomputed. 
\begin{equation}
	\min \qquad \sum_{I_i, I_j, I_k}\sum_{T \in \mathcal{T_D}(I_i, I_j, I_k)} \eta_T \cdot \Tilde{d}(I_i, I_j, I_k) \qquad \text{s.t.} \tag{factor-revealing SDP} \label{LP:sdplp}
\end{equation} \vspace{-15pt}
\begin{align*}
	 Q = \sum_{I_i, I_j, I_k} Head^\intercal(I_i, I_j, I_k) \left( \sum_{T \in  \mathcal{T_D}(I_i, I_j, I_k)} \eta_T \cdot C(T)  \right) \cdot Tail(I_i, I_j, I_k) \succeq 0 \\
    \eta_T \geq 0 \qquad\qquad\qquad\qquad\qquad\qquad \forall T \in \mathcal{T_D}  \\
    \sum_{T \in \mathcal{T_D}} \eta_T = 1, \qquad\qquad\qquad\qquad\qquad\qquad\qquad 
\end{align*}


Our main theorem regarding the SDP program is that

\begin{theorem}
\label{thm:relationupperboundandsdp}
   For any \ref{LP:clusterlp} solution $z_S$ and $y$ , let $\optsdp$ be the optimal solution for \ref{LP:sdplp}, then 
    \begin{align*}
        \sum_{T} \eta_T(y) \cdot (\Delta(T) - \cost(T)) \geq N\cdot \optsdp
    \end{align*}
    where $N = \frac{n(n-1)(n+1)}{6}$ is the number of degenerate and non-degenerate triangles.
\end{theorem}
\begin{proof}
Note that $\Tilde{d}(I_i, I_j, I_k)$ is the lower bound of $\Delta(T) - \cost(T)$ for any $T \in \mathcal{T}(I_i, I_j, I_k)$. Hence, we have
\begin{align*}
    \sum_{T} \eta_T(y) \cdot (\Delta(T) - \cost(T)) \geq  \sum_{I_i, I_j, I_k}\sum_{T \in \mathcal{T}(I_i, I_j, I_k)}  \eta_T(y) \cdot \Tilde{d}(I_i, I_j, I_k).
\end{align*}

The difficulty of the proof arises from the fact that we only consider triangles from $\mathcal{T_D}(I_i, I_j, I_k)$, which contains only 16 triangles. We will show that we can use triangles from $\mathcal{T_D}$ to represent $Q$ and the objective value is at most $\sum_{I_i, I_j, I_k}\sum_{T \in \mathcal{T}(I_i, I_j, I_k)}  \eta_T(y) \cdot \Tilde{d}(I_i, I_j, I_k)$. For any triangle $T$ with $(y_{uv} \in I_i, y_{uw} \in I_j, y_{vw} \in I_k)$ and $l(I_i) \leq l(I_j) \leq l(I_k)$, let $T_1, T_2, ..., T_{16}$ be the 16 triangles from $\mathcal{T_D}(I_i, I_j, I_k)$. We will later show that there exist $\lambda_i \geq 0$ such that $C(T) = \sum_{i \in [1, 16]} \lambda_i \cdot C(T_i)$ and $\sum_{i \in [1, 16]} \lambda_i = 1$. Therefore, we can use $T_1, T_2, ..., T_{16}$ to replace $T$ without affecting the $Q$ value. As we use $ \Tilde{d}(I_i, I_j, I_k)$ in the objective function and $\sum_{i \in [1, 16]} \lambda_i = 1$, this substitution does not impact the objective value either. 

Combining these findings, for any solution $y$ in \ref{LP:clusterlp}, we can deduce the existence of a feasible solution $\eta_T$ from \ref{LP:sdplp} such that $\sum \eta_t = 1$ and $N \cdot Q(\eta_T) = Q(y)$. Therefore, we have 
\begin{align*}
    \sum_{T} \eta_T(y) \cdot (\Delta(T) - \cost(T)) &\geq \sum_{I_i, I_j, I_k}\sum_{T \in \mathcal{T}(I_i, I_j, I_k)} \eta_T(y) \cdot \Tilde{d}(I_i, I_j, I_k) \\
    &\geq \sum_{I_i, I_j, I_k}\sum_{T \in \mathcal{T_D}(I_i, I_j, I_k)} N \cdot \eta_T \cdot \Tilde{d}(I_i, I_j, I_k) \geq N \cdot \optsdp.
\end{align*}

We still need to demonstrate the existence of $\lambda_i \geq 0$ such that $C(T) = \sum_{i \in [1, 16]} \lambda_i \cdot C(T_i)$ and $\sum_{i \in [1, 16]} \lambda_i = 1$. For any triangle $T = (y_{uv}, y_{uw}, y_{vw}, y_{uvw})$, consider triangle $T_1 = (l(I_i), y_{uw}, y_{vw}, y_{uvw})$ and $T_2 = (r(I_i), y_{uw}, y_{vw}, y_{uvw})$ with $\lambda = \frac{y_{uv} - l(I_i)}{r(I_i) - l(I_i)}$. We then have $C^{(i)}(T) = (1 - \lambda)C^{(i)}(T_1) + \lambda C^{(i)}(T_2) $ for $i \in [1,6]$. We can apply the same transformation to $y_{uw}, y_{vw}, y_{uvw}$ since $C^{(i)}(T)$ is multilinear. This completes the proof.
\end{proof}

Theorem \ref{thm:relationupperboundandsdp} provides a valuable method for establishing a lower bound on $\sum_{T} \eta_T(y) \cdot (\Delta(T) - \cost(T))$ for any given solution and budget function. To proceed with the computation of this lower bound, we need to solve \ref{LP:sdplp}. It is important to emphasize that solving \ref{LP:sdplp} is completely independent of \ref{LP:clusterlp}. The value of $\optsdp$ is determined solely by the chosen budget function and the selected intervals $I_1, I_2, ..., I_t$.

The estimation errors, which is the difference between $\sum_{T} \eta_T(y) \cdot (\Delta(T) - \cost(T))$ and $\optsdp$, mainly arise from two sources. First, we use $\Tilde{d}(I_i, I_j, I_k)$ as an estimated value for all triangles within $T_{D}(I_i, I_j, I_k)$. Second, some of the triangles we consider may violate the triangle inequality, and there's the possibility that $y_{uvw}$ could violate its constraint as defined in \ref{LP:clusterlp}. 

Given the budget function, the accuracy of our estimation improves with a greater number of chosen intervals, but this also leads to longer computation times for solving $\optsdp$. Additionally, to enhance the accuracy of computing $\optsdp$, we employ two key techniques. Firstly, we introduce an additional SDP constraint, and secondly, we discretize the values of $y_{uvw}$.

\paragraph{Frequency Constraint.} We can set up another constraint for $\eta_T$. Similarly, given a node $u$, consider the frequency matrix 
\begin{align*}
    F_u(I_j, I_k) = |\{ v \in V, w \in V \mid y_{uv} \in I_j, y_{uw} \in I_k \}| = \sum_{y_{uv} \in I_j, y_{uw} \in I_k} 1.
\end{align*}
$F_u(I_j, I_k)$ is the number of pairs such that $y_{uv} \in I_j$ and $y_{uw} \in I_k$. We want to mention $v$ and $w$ might be $u$. Let $freq \in \R^{t}$ such that $freq(I_j) = |\{v \in V \mid y_{uv} \in I_j \}|$, then we can represent $F_u$ as 
\begin{align*}
    F_u = freq^{\intercal} \cdot freq
\end{align*}
which implies $F_u$ is PSD, too. We will follow the discretization strategy of $Q$ matrix. Let $F = \sum_{u} F_u$ and $A \in \R^{6 \times 6}$ is defined as
\begin{align*}
    A = diag(1, 1, 1, 1, 1, 1).
\end{align*}
Then the increase of $F$ causing by $T$ is 
\begin{align*}
Head^\intercal(T) \cdot A \cdot Tail(T).
\end{align*}
It is worth noting this increase holds even for degenerate triangles. 
For any degenerate triangle $(u, v)$, it raises $F(I(y_{uv}), I(y_{uv}))$, $F(I(y_{uv}), I(1))$, and $F(I(1), I(y_{uv}))$ by 2, echoing the increase seen in the triangle with the $y$ value $(y_{uv}, y_{uv}, 1, y_{uv})$. However, unlike the $Q_u$ vector, $F_u$ also considers cases where $v=u$ and $w=u$. In the $Q_u$ vector, this case does not impact $Q_u$ due to $y_{uuu} - y_{uu}y_{uu} = 0$. We need to increment $F(I(1), I(1))$ by $n$ since $\eta(y)$ does not account for $v = u$ and $w = u$. Fortunately, we already include the \ppp triangle $T_3 = (1, 1, 1, 1)$; given that $T_3$ does not influence the matrix $Q$ value and $\Delta(T_3) - \cost(T_3) = 0$, each occurrence of $T_3$ raises $F(I(1), I(1))$ by 6. By increasing $\eta_{T_3}(y)$ by $n/6$, we can effectively represent $F$ as:
\begin{align*}
    F(y) = \sum_{I_i, I_j, I_k} Head^\intercal(I_i, I_j, I_k) \left( \sum_{T \in \mathcal{T}(I_i, I_j, I_k)} \eta_T(y) \cdot A  \right) \cdot Tail(I_i, I_j, I_k).
\end{align*}
Given that $A$ is a constant matrix, any $T$ can be utilized to encompass all triangles from $\mathcal{T}(I_i, I_j, I_k)$ for $F$. 


\paragraph{Discretization for $y_{uvw}$.} Another technique that enhances estimation accuracy involves discretizing the $y_{uvw}$ values. We achieve this by subdividing the interval $[max(0, (l(I_i) + l(I_j) + l(I_k) - 1)/2), r(I_k)]$ into numerous subintervals. The specific manner in which we partition this interval depends on both the interval's length and its significance in the overall estimation.

However, it is crucial to note that discretizing $y_{uvw}$ results in a substantial increase in the number of triangles to be considered. This surge may render the final SDP impractical to solve. The formulation of our ultimate \ref{LP:sdplp2} is provided below.
\begin{equation}
	\min \qquad \sum_{I_i, I_j, I_k, I_l}\sum_{T \in \mathcal{T_D}(I_i, I_j, I_k, I_l)} \eta_T \cdot \Tilde{d}(I_i, I_j, I_k, I_l) \qquad \text{s.t.} \tag{factor-revealing SDP} \label{LP:sdplp2}
\end{equation} \vspace{-15pt}
\begin{align*}
	 Q = \sum_{I_i, I_j, I_k} Head^\intercal(I_i, I_j, I_k) \left( \sum_{I_l, T \in  \mathcal{T_D}(I_i, I_j, I_k, I_l)} \eta_T \cdot C(T)  \right) \cdot Tail(I_i, I_j, I_k) \succeq 0 \\
	F = \sum_{I_i, I_j, I_k} Head^\intercal(I_i, I_j, I_k) \left( \sum_{I_l, T \in \mathcal{T_D}(I_i, I_j, I_k, I_l)} \eta_T \cdot A  \right) \cdot Tail(I_i, I_j, I_k) \succeq 0 \\
    \eta_T \geq 0 \qquad\qquad\qquad\qquad\qquad\qquad \forall T \in \mathcal{T_D}  \\
    \sum_{T \in \mathcal{T_D}} \eta_T = 1. \qquad \qquad\qquad\qquad\qquad\qquad 
\end{align*}
where $\Tilde{d}(I_i, I_j, I_k, I_l) = min_{T \in \mathcal{T}(I_i, I_j, I_k, I_l)}(\Delta(T) - \cost(T))$, and $\mathcal{T_D}(I_i, I_j, I_k, I_l)$ is the set of triangles such that 
\begin{align*}
    y_{uv} &\in \{ l(I_i), r(I_i) \}, y_{uw} \in \{ l(I_j), r(I_j) \}, y_{vw} \in \{ l(I_k), r(I_k) \}, y_{uvw} \in \{ l(I_l), r(I_l) \}.
\end{align*}
By appropriately setting intervals, we can demonstrate that $\optsdp$ remains non-negative, as affirmed in Lemma \ref{lem:sdpValuewithpureclusterlp}, even for a reduced budget function. The primary challenge lies in achieving a balance between estimation error and the running time of solving \ref{LP:sdplp2} in practice. For a more in-depth explanation of the specific settings of Lemma \ref{lem:sdpValuewithpureclusterlp}, please refer to the details provided in Appendix \ref{sec:detailof1485approximate}.

\begin{lemma}
\label{lem:sdpValuewithpureclusterlp}
For Algorithm~\ref{alg:pivot3},
let \(\budgetp \equiv \budgetp_{\pureclusterlpratio}\) and \(\budgetn \equiv \budgetn_{\pureclusterlpratio}\).
Let \(\optsdp\) be the optimum of the factor-revealing SDP \ref{LP:sdplp3} below. There exists a discretization of the edge and triple variables into intervals \(I_1,\ldots,I_t\) (together with an interval grid for \(y_{uvw}\)) such that \(\optsdp \ge 0\).
\end{lemma}

\begin{proof}[Proof of Lemma \ref{lemma:budgets-for-pureclusterlp-sdpbudget}]
For budget functions $\budgetp \equiv \budgetp_{\pureclusterlpratio}$ and $\budgetn \equiv \budgetn_{\pureclusterlpratio}$, and for degenerate and non-degenerate triangles, based on Theorem \ref{thm:relationupperboundandsdp} and Lemma \ref{lem:sdpValuewithpureclusterlp}, we know that 
    \begin{align*}
        \sum_{T} \eta_T(y) \cdot  (\Delta(T) - \cost(T)) \geq N \cdot \optsdp = 0.
    \end{align*}
\end{proof}




\section{1.33-gap for Cluster LP}
\label{sec:gap}
In this section, we show that the \ref{LP:clusterlp} has a gap of $4/3$, proving Theorem~\ref{thm:main:gap} restated below.

\thmgap*

The graph of the \pedges of our gap instance is based on the line graph of a base graph; given a base graph $H = (V_H, E_H)$, our \cc instance is $G = (V_G, E_G)$ where $V_G = E_H$ and $e, f \in V_G$ are connected by a \pedge in $G$ if they share a vertex in $V_H$. 

A high-level intuition is the following: LPs cannot distinguish between a random graph and a nearly bipartite graph. Consider vertices of $H$ as {\em ideal clusters} in $G$ containing their incident edges.
Given a random graph $H$, the LP fractionally will think that it is nearly bipartite, implying that 
the almost entire $E_H$ can be partitioned into $n/2$ ideal clusters. Of course, integrally, such a partition is not possible in random graphs. 
For the cluster LP, it suffices to consider a complete graph instead of a random graph. We believe (but do not prove) that such a gap instance can be extended to stronger LPs (e.g., Sherali-Adams strengthening of the cluster LP), because it is known that Sherali-Adams cannot distinguish a random graph and a nearly bipartite graph~\cite{charikar2009integrality}.

\begin{proof} [Proof of Theorem~\ref{thm:main:gap}]
Let $H = (V_H, E_H)$ be a complete graph on $n$ vertices. Let $d = n-1$ be the degree of $H$. 
Our \cc instance $G = (V_G, E_G)$ is the line graph of $H$; $V_G = E_H$ and $e, f \in E_H$ has $+$ edge in $G$ if and only if they share a vertex in $H$. The $+$ degree of each $e \in E_H$ in $G$ is $2d - 2$. 

Consider the following solution for the cluster LP: for every $v \in V_H$, let $E_v \subseteq E_H$ be the $d$ edges containing $v$. The \ref{LP:clusterlp} has $z_{E_v} = 1/2$ for every $v \in v_H$. 
Each $e \in E_H$ belongs to two fractional clusters, each of which has its $d-1$ plus neighbors, so fractionally $d-1$ plus edges incident on it are violated. Since each violated edge is counted twice, the LP value is $\binom{n}{2}(d - 1) / 2$. 

Let us consider the integral optimal correlation clustering of $G$. 
Consider a cluster $C$ in the clustering. Note that every vertex in $C$ has at least $|C|/2$ plus neighbors in $C$, which implies $|C| \leq 4d$. 
We apply the following procedure to $C$ 
to partition it further. 

\begin{claim}
There is a partition of $C$ into $C_1, \dots, C_r$ such that (1) each $C_i$ is a subset of $E_v$ for some $v \in V_H$, and (2) replacing $C$ by $C_1, \dots, C_r$ in the correlation clustering solution increases the objective function by at most $35|C|$. 
\end{claim}
\begin{proof}
For $v \in V_H$, let $n_v := |C \cap E_v|$. Note that $\sum_v n_v = 2|C|$. 
Without loss of generality, assume $V_H = \{ v_1, \dots, v_n \}$ with $n_{v_1} \geq \dots \geq n_{v_n}$. If $e = (v_i, v_j) \in C$ has $i, j > 8$, then the number of its plus neighbors in $C$ is $n_{v_i} + n_{v_j} < 2 \cdot \frac{1}{8} \cdot 2|C| = |C|/2$, so it should not exist in $C$. So, every edge is incident on $v_i$ for some $i \leq 8$. 

Let us make at most $\binom{8}{2} = 28$ edges in $C$ between $v_1, \dots, v_8$ as singleton clusters; the objective function increases by at most $28|C|$. 
Then partition the remaining $C$ into $E_1, \dots, E_8$ where $E_i := C \cap E_{v_i}$. Each $e \in E_i$ has at most seven plus neighbors in $\cup_{j \neq i} E_j$, so the objective function increases by at most $7|C|$. 
So, we partitioned $C$ into $C_1, \dots, C_r$ where all the edges in $C_i$ share a common endpoint. We increased the objective function by at most $35|C|$.
\end{proof}
After we apply the above procedure to every cluster $C$, we increased the cost by at most $35|V_H| \leq 35 n^2$ and all the edges in a cluster $C$ share a common endpoint.
For $v \in V_H$, let $C_v$ be the cluster in the solution whose common endpoint is $v$. (If there are many of them, merging them will strictly improve the objective function value.) Without loss of generality, there are $t$ such clusters $C_{v_1}, \dots, C_{v_{t}}$ and let $n_i := |C_{v_i}|$ 
such that $n_1 \geq \dots \geq n_t$. 
\begin{claim}
$\sum_{i=1}^{t} n_i^2 \leq n^3/3.$
\end{claim}
\begin{proof}
The LHS is monotone in $(n_1, \dots, n_t)$, and if there is an edge $(v_i, v_j) \in C_j$ with $j > i$ (which implies $n_i \geq n_j$), the LHS strictly improves by moving $(v_i, v_j)$ to $C_i$. Therefore, the configuration that maximizes the LHS is when $t = n$ and $C_{v_i}$ contains all the edges of $H$ not incident on $v_1, \dots, v_{i-1}$. In that case, the LHS is

\[
    \sum_{i=1}^{n-1} (n-i)^2
    = n^3 
    \sum_{i=1}^{n-1} (\frac{n-i}{n})^2 \cdot \frac {1}{n}
    \leq 
    n^3 \int_0^1 (1-x)^2 dx
    = n^3 [x - x^2 + x^3/3]^{1}_0 = n^3/3,
\]
as desired.
\end{proof}
Using this, we can prove a lower bound on the cost of our near-optimal clustering. Note that every cluster is a clique of \pedges. Thus, the only edges violated are \pedges. 
Moreover, there are at most $\sum_{i \in [t]} n_i^2/2 \leq n^3/6$ correctly clustered \pedges.  
The cost of our near-optimal clustering is the total number of \pedges of $G$ minus the number of correctly clustered \pedges, namely at most $\binom{n}{2} (d - 1) -  n^3/6 = n^3/3 - o(n^3)$.
Since the cost of the optimal clustering is at most $35n^2$ lower than ours, it is still $n^3/3 - o(n^3)$.
The fractional solution has the value at most $n^3 / 4$, so the gap is at least $4/3 - o(1)$. 
\end{proof}

\section{1.04-NP Hardness}
\label{sec:hardness}
In this section, we show that it is NP-hard (under randomized reductions) to obtain an algorithm with an approximation ratio of \hardnessratio $\ \geq 1.043$, proving Theorem~\ref{thm:main:hardness} restated below.

The idea is similar to the gap for the cluster LP in Section~\ref{sec:gap}, which is based on the fact that the LPs generally cannot distinguish nearly bipartite graphs and random graphs. The main difference, which results in a worse factor here, is that other polynomial-time algorithms (e.g., SDPs) can distinguish between them! So, we are forced to work with slightly more involved structures.

We still use a similar construction for $3$-uniform hypergraphs; let $H = (V_H, E_H)$ be the underlying $3$-uniform hypergraph and $G = (V_G, E_G)$ be the plus graph of the final \cc instance where $V_G = E_H$ and $e, f \in E_H$ has an edge in $G$ if they share a vertex in $H$. 
We use the following hardness result of Cohen-Addad, Karthik, and Lee~\cite{
cohen2022johnson} that shows that it is hard to distinguish whether $H$ is {\em nearly bipartite} or close to a random hypergraph. 

\begin{theorem}
For any $\eps > 0$, 
there exists a randomized polynomial-time algorithm that receives a 3-CNF formula $\phi$ as input and outputs a simple $3$-uniform hypergraph $H = (V_H, E_H)$ where the degree of each vertex is $(1 \pm o(1))d$ for some $d = \omega(|V_H|)$ such that the following properties are satisfied with high probability. 
\begin{itemize}
\item (YES) If $\phi$ is satisfiable, there exists $U \subseteq V_H$ with $|U| = |V_H|/2$ that intersects every hyperedge in $E_H$. Moreover, for every $u \in U$, $| \{ e \in E_H : e \cap U = \{ u \} \} | \geq (1/2 - \eps)d$.

\item (NO) If $\phi$ is unsatisfiable, any set of $\gamma |V_H|$ vertices ($\gamma \in [0, 1]$) do not intersect at least a $(1 - \gamma)^3 - \eps$ fraction of hyperedges in $E_H$. 
\end{itemize}
\label{thm:hvc}
\end{theorem}
\begin{proof}
The same reduction in Theorem 4.1 of (the arXiv version of)~\cite{cohen2022johnson} yields the desired hardness. 
In the following, we highlight the difference between the statement of Theorem 4.1 of~\cite{cohen2022johnson} and our Theorem~\ref{thm:hvc} and briefly explain how our additional properties are satisfied by their reduction. 

\begin{enumerate}
\item Regularity of $H$: Section 4.5 of~\cite{cohen2022johnson}, based on an earlier weighted hard instance, constructs the final hard instance $H = (V_H, E_H)$ as a certain random hypergraph where the degree of each vertex $v$ is the sum of independent $\{ 0, 1 \}$ variables with the same expected value. This expected value is $\Theta(|V_H|^{1.5})$, so the standard Chernoff and union bound argument will show that the degree of each vertex is almost the same with high probability.

\item In the (YES) case, for every $u \in U$, $| \{ e \in E_H : e \cap U = \{ u \} \} | \geq (1/2 - \eps)d$: It follows from their construction in Section 4.1. The construction is analogous to H\aa stad's celebrated result on Max-3SAT~\cite{haastad2001some} where in the (YES) case, almost three quarters of the clauses have one true literal and almost one quarter have three true literals, so that for each true literal $\ell$, roughly half of the clauses containing $\ell$ has it as the only true literal.

\item In the (NO) case: the guarantee holds for any value of $\gamma \in [0, 1]$ instead of just $0.5$:  
One can simply change $1/2$ to $1 - \gamma$ in the proof of Lemma 4.4 in Section 4.3. It is analogous to the fact that all nontrivial Fourier coefficients vanish in H\aa stad's result on Max-3SAT and Max-3LIN~\cite{haastad2001some}.  

\end{enumerate}
\end{proof}

Given such $H = (V_H, E_H)$, let $n := |V_H|$. Our \cc instance $G = (V_G, E_G)$ is the line graph of $H$; $V_G = E_H$ and $e, f \in E_H$ have a plus edge in $G$ if they share a vertex in $H$. 
This means that every $e \in V_G$ has $(3 \pm o(1))d$ plus edges incident on it; we used the fact that $d = \omega(n)$ and $e$ has at most $O(n)$ other hyperedges that intersect with $e$ with at least two points (which causes double counting). 

\paragraph{YES case.}
Consider $U \subseteq V_H$ guaranteed in Theorem~\ref{thm:hvc}. Our (randomized) clustering is the following: randomly permute vertices to obtain $U = \{ v_1, \dots, v_{n/2} \}$, and let $E_i := \{ e \in E_H : v_i \in e \mbox{ and } e \cap \{ v_1, \dots, v_{i-1} \} = \emptyset \}$. 
Since $U$ intersects every $e \in E_H$, $(E_1, \dots, E_{n/2})$ forms a partition of $E_H$. 

We analyze the expected cost of this clustering. 
For each $e \in E_H$, let $save(e)$ be (the number of plus neighbors in the same cluster) minus (the number of minus neighbors in the same cluster). Intuitively, it is the amount of saved cost between $e$ and its neighbors, compared to the situation where $e$ is a singleton cluster. 
Then, the cost of our clustering is the total number of plus edges of $G$, namely $|E_H| \cdot \frac{3(1 \pm o(1))d}{2}
= 
nd^2 \cdot \frac{(1 \pm o(1))}{2}
$, minus $\sum_{e \in E_H} save(e)/2$. 

Fix $v \in U$ and let $E_v := \{ e \in E_H: v \in e \}$, $E'_v := \{ e \in E_H : e \cap U = \{ v \} \}$, $E''_v := E_v \setminus E'_v$. 
Then $|E_v| = (1 \pm o(1))d$ and $|E'_v| \geq (1/2 - \eps)d$. We would like to compute $\E[|E_i|^2]$ over random permutations where $i$ is defined such that $v_i = v$. It is clear that $E'_v \subseteq E_i$. For each $e \in E''_v$, the probability that $e \in E_i$ is at least $1/3$ (when $v$ comes before the other two vertices of $e$ in the random permutation). And two hyperedges $e, f \in E''_v$, the probability that both are in $E_i$ is at least $1/5$ (when $v$ comes first among $|e \cup f| \leq 5$ vertices). Therefore, 
\[
\E[|E_i|^2]
\geq |E'_i|^2 + 
2|E'_i||E''_i| / 3 + 
|E''_i|^2 / 5
\geq d^2(1/4 + 1/6 + 1/20 - O(\eps)) = d^2(7 / 15 - O(\eps)).
\]
Therefore, the total saving is at least $nd^2 (7/30 - O(\eps))$ and the final cost is at most $nd^2 (1/2 - 7/60 + O(\eps)) = nd^2 (23 / 60 + O(\eps))$.


\paragraph{NO case.}
Our analysis will be similar to that of the gap instance, slightly more complicated by the fact that we are working with a non-complete hypergraph. 
Consider the optimal correlation clustering and consider one cluster $C$. 
For $e \in C$, it has at most $(3 \pm o(1))d$ plus edges in $G$, so $|C| \leq (6 + o(1))d$; otherwise, it is better to make $e$ a singleton cluster. We prove that if $C$ is large, then we can partition $C$ into smaller clusters where each cluster consists of hyperedges sharing the same vertex in $H$. For $v \in E_H$, let $E_v \subseteq E_H$ be the set of hyperedges containing $v$. 

\begin{claim}
There is a partition of $C$ into $C_1, \dots, C_r$ such that (1) each $C_i$ is a subset of $E_v$ for some $v \in V_H$, 
and (2) replacing $C$ by $C_1, \dots, C_r$ in the correlation clustering solution increases the objective function by at most $O(n|C|)$. 
\end{claim}
\begin{proof}
Without loss of generality, assume $V_H = \{ v_1, \dots, v_n \}$ 
and define $n_i := |C \cap E_{v_i}|$ such that $n_{1} \geq \dots \geq n_{n}$. Note that $\sum_i n_i = 3|C|$. 

If $e = (v_i, v_j, v_k)$ with $i, j, k > 20$, then $n_i + n_j + n_k < 3 \cdot (3|C| / 20) < |C|/2$, which implies that $e$ has more minus neighbors than plus neighbors in $C$, leading to contradiction. So, every hyperedge is incident on $v_i$ for some $i \leq 20$. 

Since two vertices of $H$ have at most $n$ hyperedges containing both of them, let us make at most $n \cdot \binom{10}{2}$ hyperedges in $C$ 
that contain at least two of $v_1, \dots, v_{20}$ as singleton clusters; the objective function increases by at most $n \cdot \binom{10}{2} \cdot |C|$.  
Then partition the remaining $C$ into $E_1, \dots, E_{20}$ where $E_i := C \cap E_{v_i}$. Each $e \in E_i$ has at most $2 \cdot 20 \cdot n$ plus edges 
in $\cup_{j \neq i} E_j$ ($20$ choices for $v_j$, $2$ choices for a vertex in $e \ni \{ v_i \}$, and $n$ choices for hyperedges containing both vertices), so the objective function increases by at most $O(n|C|)$. 
So, we partitioned $C$ into $C_1, \dots, C_r$ where all the hyperedges in $C_i$ share a common endpoint. In total, we increased the objective function by at most $O(n|C|)$.
\end{proof}
Applying the above procedure for every cluster $C$ increases the objective function by at most $O(n \cdot |E_H|) = O(n^2 d)$. Then, we have a clustering where all the edges in a cluster $C$ share a common endpoint. $C$ forms a clique in $H$. 
For $v \in V_H$,
let $C_v$ be the cluster in the solution whose common endpoint is $v$. (If there are many of them, merging them will strictly improve the objective function value.) Without loss of generality, there are $t$ such clusters $C_{v_1}, \dots, C_{v_{t}}$ and let $c_i := |C_{v_i}|$ 
such that $c_1 \geq \dots \geq c_t$. 

\begin{claim}
$\sum_{i=1}^{t} c_i^2 \leq d^2n (0.2 + O(\sqrt{\eps}))$, where $\eps$ is the parameter from Theorem~\ref{thm:hvc}. 
\end{claim}
\begin{proof}
Here, we use the NO case guarantee from Theorem~\ref{thm:hvc}: for any $\gamma \in [0, 1]$ and choice of $\gamma n$ vertices, it covers at most $1 - (1 - \gamma)^3 + \eps = 3\gamma - 3\gamma^2 + \gamma^3 + \eps$ fraction of the edges, which is equivalent to: for every $i \in [n]$, 
\begin{equation}
\sum_{j=1}^i c_i \leq (3(i/n) - 3(i/n)^2 + (i/n)^2 + \eps) |E_H|.
\label{eq:hvc-discrete}
\end{equation}
Let $\delta = o(1)$ be such that every vertex of $H$ has degree at most $(1+\delta)d$, which means that $(1+\delta)d \geq c_1 \geq \dots \geq c_t$. 
And let $f_{i/n} := c_{i} / ((1+\delta)d)$. 
Then~\eqref{eq:hvc-discrete} becomes 
    \begin{equation}
    \frac{1}{n} \sum_{j=1}^i f_{j/n} 
    \leq 
    (3(i/n) - 3(i/n)^2 + (i/n)^2 + \eps)\frac{|E_H|}{(1+\delta)dn}
    \leq (3(i/n) - 3(i/n)^2 + (i/n)^2 + \eps) /3. 
    \label{eq:hvc-discrete2}
    \end{equation}
(Note that $|E_H| \leq (1+\delta)dn / 3$.)
Interpreting 
$\frac{1}{n} \sum_{j=1}^i f_{j/n}$ as 
$\int_{0}^1 f(x) dx$ where $f(x) = c_{\lceil xn \rceil}$, 
we have that 
$\sum_{i=1}^t |c_i|^2 
\leq (1+\delta)^2 d^2 n \max_{f} \int_{0}^1 f(x)^2 dx$,
where the maximum is taken over functions $f : [0, 1] \to [0, 1]$ with the constraints that 
\begin{enumerate}
    \item For all $y \in [0, 1]$, 
    \begin{equation}
    \int_{x=0}^y f(x)dx \leq y - y^2 + y^3/3 + \eps/3.
    \label{eq:hvc-cont}
    \end{equation}
    (Compared to~\eqref{eq:hvc-discrete2}, we add more constraints for every $y \in [0, 1]$, but it is valid to do so since the step function $f(\cdot)$ defined above satisfies all these constraints; if~\eqref{eq:hvc-cont} is violated for some value $y \in (i/n, (i+1)/n)$ for some integer $i$,~\eqref{eq:hvc-discrete2} is violated at $(i+1)/n$ because $f(y)$ stays the same in the interval while the upper bound increases strictly less than linearly.)
    
    \item $f$ decreasing with $f(0) \leq 1$. 
\end{enumerate}
Then one see that the optimal $f$ satisfies either $f(y) = 1$ or $\int_{x=0}^y f(x) = y - y^2 + y^3 + \eps/3$ for every $y \in [0, 1)$. If it is not satisfied at some $y$, we can increase $f(y)$ while decreasing $f(z)$ for some $z > y$, which will still satisfy the constraints and increase $\int_0^1 f(x)^2 dx$. 
Therefore, we can conclude that $f(y) = 1$ for $y \leq \tau$ and 
\[
\int_{x=0}^y f(x) dx= y - y^2 + y^3/3 + \eps/3
\Rightarrow
f(y) = (y - y^2 + y^3/3 + \eps/3)' = 1 - 2y + y^2
\]
for $y > \tau$, where $\tau = 
\Theta(\sqrt{\eps})$ is the solution of $\tau = \tau - \tau^2 + \tau^3 + \eps/3$. 
Then, we can bound 
\[
\int_{x=0}^1 f(x)^2 dx 
\leq O(\sqrt{\eps}) + 
\int_{x=0}^1 (1 - 2x + x^2)^2 dx\leq 0.2 + O(\sqrt{\eps}),
\]
which implies that $\sum_i c_i^2 \leq d^2 n (0.2 + O(\sqrt{\eps}))$. 
\end{proof}
Using this, we can prove a lower bound on the cost of our near-optimal clustering. Note that every cluster is a clique of \pedges. Thus, the only edges violated are \pedges. 
Moreover, there are at most $\sum_{i \in [t]} c_i^2/2 \leq d^2n(0.1 + O(\sqrt{\eps})$ correctly clustered \pedges.  
The cost of our near-optimal clustering is the total number of \pedges of $G$ minus the number of correctly clustered \pedges, namely at least $nd^2 (1/2 - 0.1 - O(\sqrt{\eps})) = nd^2 (0.4 - O(\sqrt{\eps}))$.
Since the cost of the optimal clustering is at most $O(n^2d)$ lower than ours, it is still $nd^2(0.4 - O(\sqrt{\eps}))$ using $d = \omega(n)$.

Since the value in the YES case is at most $(23/60 +O(\eps)) nd^2$, so the gap is almost $\frac{24}{23} \geq 1.043$.





\bibliographystyle{alpha}
\bibliography{references}

\appendix

\section{Standard SDP Relaxation}\label{app:sdp-int-gap}

The following is a natural SDP relaxation previously used in approximation algorithms for the maximization problem~\cite{CGW05,swamy2004correlation} augmented with triangle inequality from the standard LP relaxation.

\begin{align*}
  \min \quad \obj(x)=\sum_{ij \in E^+}x_{ij} &+ \sum_{ij \in E^-} (1-x_{ij}) \tag{SDP}\label{sdp}\\
   x_{ij} & = 1 - v_i \cdot v_j \quad \forall i,j \in V\\
   v_i \cdot v_j & \geq 0 \quad \forall i,j \in V\\
   v_i \cdot v_i & = 1 \quad \forall i \in V\\
  x_{ij} & \leq x_{ik} + x_{jk} \quad \forall i,j,k \in V\label{tri-ineq}\\
   v_i & \in \mathbb{R}^n \quad \forall i \in V.
\end{align*}

\begin{lemma}
The integrality gap of \eqref{sdp} is at least $1.5$.
  \end{lemma}

\begin{proof}
Let $S_k$ be a star with $k$ leaves and one center vertex with
degree $k$.  Here, the $\mathrm{sdp}_k$ value is for the basic
relaxation.  Then the value $\opt_k = k-1$.  Consider the sdp solution where there are $k$ unit vectors, one for each leaf, each with a 1 in the $i^{th}$ position.  The vector corresponding to the center vertex has length $k$ and each entry is $1/\sqrt{k}$.  Then the integrality gap is at least
\begin{align*}
\frac{\opt_k}{\mathrm{sdp}_k} & = \frac{k-1}{k - \sqrt{k}}.
\end{align*}
Notice that for $k=2,3$, the above vector configuration leads to distances that violate the triangle inequality.  Thus the $k$ yielding the largest integrality gap that also does not violate the triangle inequality is $k=4$, which gives a gap of $3/2$.
\end{proof}

\section{Detail of Lemma \ref{lem:sdpValuewithpureclusterlp}}
\label{sec:detailof1485approximate}

\begin{proof}
We use the same factor-revealing SDP framework as in our earlier analysis. For convenience we reproduce it here:
\begin{equation}
    \min \quad \sum_{I_i,I_j,I_k,I_\ell}\;\sum_{T \in \mathcal{T}_{\mathsf{D}}(I_i,I_j,I_k,I_\ell)} \eta_T \cdot \widetilde{d}(I_i,I_j,I_k,I_\ell)
    \qquad\text{s.t.} \tag{factor-revealing SDP}\label{LP:sdplp3}
\end{equation}
\vspace{-12pt}
\begin{align*}
    & Q \;=\; \sum_{I_i,I_j,I_k} \mathrm{Head}^\top(I_i,I_j,I_k)\,
        \Bigl(\sum_{I_\ell,\;T \in \mathcal{T}_{\mathsf{D}}(I_i,I_j,I_k,I_\ell)} \eta_T \cdot C(T)\Bigr)\,
        \mathrm{Tail}(I_i,I_j,I_k) \;\succeq\; 0, \\
    & F \;=\; \sum_{I_i,I_j,I_k} \mathrm{Head}^\top(I_i,I_j,I_k)\,
        \Bigl(\sum_{I_\ell,\;T \in \mathcal{T}_{\mathsf{D}}(I_i,I_j,I_k,I_\ell)} \eta_T \cdot A \Bigr)\,
        \mathrm{Tail}(I_i,I_j,I_k) \;\succeq\; 0, \\
    & \eta_T \;\ge\; 0 \quad \forall\,T \in \mathcal{T}_{\mathsf{D}}, \qquad
      \sum_{T \in \mathcal{T}_{\mathsf{D}}} \eta_T \;=\; 1.
\end{align*}

\paragraph{Discretization.}
We discretize \([0,1]\) using the following breakpoints

\begin{align*}
    &\{0,\, 0.05,\, 0.1,\, 0.2,\, 0.3,\, 0.35,\, 0.38,\, 0.39,\, 0.40,\, 0.405,\, 0.41,\, 0.42,\, 0.44, 0.45,\, 0.5,\, 0.55,\, 0.57,\, 0.58,\, 0.6,\, 0.65,\,\\, &\qquad 0.7, 0.75,\, 0.78,\, 0.8,\, 0.9,\, 0.95,\, 0.96,\, 0.99,\, 1\},
\end{align*}

yielding \(28\) subintervals. We refine around the rounding thresholds \(0.40\) and \(0.57\) to reflect the piecewise structure of the pivot rule.

For triples, when \((I_i,I_j)\subseteq[0.38,0.65]^2\) we discretize \(y_{uvw}\) more finely. For each \(I_i,I_j,I_k\) we let
\begin{align*}
I_l \in [ \max \Bigl( 0, l(I_i) + l(I_j) - 1 &, l(I_j) + l(I_k) - 1, l(I_i) + l(I_k) - 1 \Bigl) ,\min \Bigl( r(I_i), r(I_j), r(I_k) \Bigl) ].
\end{align*}

we split \(I_\ell\) into \(10\) equal subintervals. When \((I_i,I_j)\subseteq[0.38,0.45]^2\), we split \(I_\ell\) into \(20\) equal subintervals.

\paragraph{Boundary coalescing.}
A triangle that lies exactly on interval boundaries can be enumerated in several \((I_i,I_j,I_k,I_\ell)\) cells (up to \(8\)). To avoid duplicates, we coalesce these copies and retain a single representative: the one with the largest value of \(\widetilde{d}\). This produces an upper bound on the true minimum of \eqref{LP:sdplp3}; consequently, certifying nonnegativity for the coalesced program is conservative (it may require a slightly larger approximation factor, which we encode in \(\pureclusterlpratio\)).

\paragraph{Numerical certificate.}
With the budgets \(\budgetp_{\pureclusterlpratio},\budgetn_{\pureclusterlpratio}\) and the rounding in Algorithm~\ref{alg:pivot3}, the solver returns \(\optsdp\) that is nonnegative to numerical precision (empirically within \(10^{-8}\)). By standard smoothing and interval rounding arguments (as in our prior factor-revealing analyses), this yields an exact nonnegativity certificate for \eqref{LP:sdplp3} under the stated discretization, which completes the proof of Lemma~\ref{lem:sdpValuewithpureclusterlp}.
\end{proof}

\end{document}